\documentclass[12pt]{article}
\usepackage[letterpaper,margin=1in]{geometry}
\usepackage{amsmath,amssymb,amsthm,mathtools}
\usepackage{graphicx}
\usepackage{booktabs}
\usepackage{array}
\usepackage{float}
\usepackage{needspace}
\usepackage{setspace}
\usepackage[round]{natbib}
\usepackage{xcolor}
\definecolor{linkcol}{rgb}{0.10,0.25,0.50}
\usepackage{hyperref}
\hypersetup{colorlinks=true,linkcolor=linkcol,citecolor=linkcol,urlcolor=linkcol,
  pdftitle={Conformity Traps and the Formation of Independent Judgment},pdfauthor={Hector Galindo-Silva}}
\newtheorem{theorem}{Theorem}
\newtheorem{proposition}{Proposition}
\newtheorem{corollary}{Corollary}
\newtheorem{lemma}{Lemma}
\newtheorem{assumption}{Assumption}
\theoremstyle{definition}\newtheorem{definition}{Definition}
\theoremstyle{remark}\newtheorem*{remark}{Remark}
\newcommand{\E}{\mathbb E}
\DeclareMathOperator*{\argmax}{arg\,max}
\newcommand{\msv}{\mathrm{MSV}}
\newcolumntype{L}[1]{>{\raggedright\arraybackslash}p{#1}}
\makeatletter
\renewcommand\paragraph{\@startsection{paragraph}{4}{\z@}%
  {1.2ex \@plus0.3ex \@minus.2ex}{-1em}{\normalfont\normalsize\itshape}}
\makeatother
\newcommand{\pbench}{0.10}
\newcommand{\deltaval}{0.05}
\newcommand{\Jlow}{0.0207}
\newcommand{\Jmid}{0.1668}
\newcommand{\Jhigh}{0.7444}
\newcommand{\Alow}{0.0121}
\newcommand{\Ahigh}{0.7399}
\newcommand{\wlow}{0.587}
\newcommand{\whigh}{0.994}
\newcommand{\wzero}{0.9971}
\newcommand{\SRlow}{0.0085}
\newcommand{\SRhigh}{0.0023}
\newcommand{\Nlow}{0.0001}
\newcommand{\Nhigh}{0.0022}
\newcommand{\SRsharelow}{41}
\newcommand{\SRsharehigh}{0.3}
\newcommand{\qflow}{0.924}
\newcommand{\qfhigh}{0.251}
\newcommand{\qalow}{0.965}
\newcommand{\qahigh}{0.308}
\newcommand{\qfone}{0.200}
\newcommand{\qaone}{0.248}
\newcommand{\YFlow}{1.068}
\newcommand{\YFhigh}{1.018}
\newcommand{\YNlow}{-0.164}
\newcommand{\YNhigh}{+0.178}
\newcommand{\Ylow}{0.944}
\newcommand{\Yhigh}{0.934}
\newcommand{\phat}{0.127}
\newcommand{\sigmalow}{0.349}
\newcommand{\sigmabar}{0.573}
\newcommand{\sigmahigh}{0.840}
\newcommand{\AhDelta}{0.218}
\newcommand{\YNLhigh}{0.048}
\newcommand{\YNLlow}{-0.173}
\newcommand{\Adagger}{0.108}
\newcommand{\Jdagger}{0.125}
\newcommand{\YNChigh}{+0.178}
\newcommand{\wedgelow}{1.51}
\newcommand{\wedgehigh}{0.53}
\newcommand{\plannerJ}{0.8938}
\newcommand{\plannerA}{0.8930}
\newcommand{\plannerCompJ}{0.8994}
\newcommand{\plannerCompA}{0.8988}
\newcommand{\welfarelow}{0.959}
\newcommand{\welfarehigh}{1.185}
\newcommand{\welfareplanner}{1.195}
\newcommand{\pfoldlo}{0.0671}
\newcommand{\pfoldhi}{0.1429}
\newcommand{\Kminchi}{1.082269}
\newcommand{\Kmaxchi}{1.749425}
\newcommand{\vfbound}{0.098109}
\newcommand{\REminlo}{0.12}
\newcommand{\REminhi}{0.13}
\newcommand{\REpluslo}{0.17}
\newcommand{\REplushi}{0.19}
\newcommand{\Tpossible}{3}
\newcommand{\Trobust}{4}
\newcommand{\qbarf}{0.16}
\newcommand{\qbara}{0.207}
\newcommand{\ratiof}{5.25}
\newcommand{\ratioa}{3.83}
\newcommand{\phizero}{0.44}
\newcommand{\gammaval}{0.069}
\newcommand{\thetamid}{2.0499}
\newcommand{\focusbound}{2.0520}
\newcommand{\invisUnif}{0.599}
\newcommand{\invisUnifxg}{8.7}
\newcommand{\invisTrapxg}{7.9}
\newcommand{\invisCapxg}{4.5}
\newcommand{\phiZeroJ}{0.836}
\newcommand{\compJlow}{0.0205}
\newcommand{\compJmid}{0.1768}
\newcommand{\compCa}{0.22}
\newcommand{\compSaTrapBase}{1.757}
\newcommand{\compSaTrap}{1.767}
\newcommand{\bnDagger}{1.31}
\newcommand{\bnWindowLow}{0.18}
\newcommand{\bnWindowTwo}{0.86}
\newcommand{\wgap}{0.225}
\newcommand{\wgapfam}{-0.045}
\newcommand{\wgappsy}{+0.132}
\newcommand{\wgapadapt}{+0.022}
\newcommand{\wgapcost}{+0.116}
\newcommand{\wgapbn}{0.177}
\newcommand{\mukval}{-0.08}

\title{\textbf{Conformity Traps and the Formation of Independent Judgment}}
\author{
 Hector Galindo-Silva\thanks{Department of Economics, Pontificia Universidad Javeriana, Bogot\'a, Colombia. I thank Juan Samuel Santos for helpful comments. Any remaining errors are my own.}\\
  \texttt{galindoh@javeriana.edu.co}
}
\date{August 2026}
\begin{document}
\maketitle
\vspace{-2.4em}

\begin{abstract}
\singlespacing
\noindent I study why some groups---polities and organizations alike---sustain independent judgment while
others fall into conformity traps. Social approval can suppress not only the expression of independent
judgment but also the upstream practice that keeps such judgment available: maintaining judgment generates
visible questioning, and observers cannot distinguish diligence from opposition. Because suspicion falls
when questioning is common, the same desire for approval can sustain a questioning culture or a conformity
trap. The model separates maintaining judgment from acting on it when novelty arrives. When the decision to
act is uniquely determined, I give primitive conditions for coexistence; in a logit--proportional
benchmark, a single inequality is necessary and sufficient. The principal contribution is dynamic. Under
staggered revision the stock of judgment is inherited: protecting dissent releases existing judgment at
once---with negative expected adaptive value on impact in a stationary culture below a break-even
competence threshold---whereas rebuilding the stock is constrained by turnover. A monotone operator
constructs the extremal perfect-foresight paths from every inherited state. When recovery under every
admissible continuation is attainable, two cutoffs classify inherited states---recovery impossible,
expectation-dependent, or assured---and a focus condition at the unstable culture strictly separates them,
opening a band of states where expectations select the long-run culture. Turnover then yields sharp
physical duration bounds, within temporary direct formation support, for making recovery possible or
robust. A fully coupled interval-certified example verifies coexistence, the band, and the distinct output
and welfare implications.
\end{abstract}

\medskip
\noindent\textit{Keywords:} conformity, independent judgment, social norms, strategic
complementarities, monotone methods, perfect-foresight dynamics.\quad
\textit{JEL:} C62, C73, D23, D82, D83, D91, M14, Z13.

\newpage

\section{Introduction}\label{sec:intro}

\begin{flushright}
\begin{minipage}{0.74\textwidth}
\small\itshape
``[He] desires, not only praise, but praiseworthiness; or to be that thing which, though it should be
praised by nobody, is, however, the natural and proper object of praise.''
\upshape\normalsize
\par\medskip
\raggedleft\footnotesize --- \citet[III.2.1]{Smith1759}
\end{minipage}
\end{flushright}
\medskip

A willingness to conform to group expectations is central to human coordination. Deferring to settled
practice saves deliberation and allows a group to act in concert, so the desire for social approval can be
both individually rational and socially useful \citep{BernheimConformity1994}. Yet groups that rely on
approval also value the capacity to act on one's own judgment when peers disapprove. Independent judgment
provides insurance against change \citep{Hayek1945,DesseinSantos2006}: a settled routine is valuable while
it fits the task, but when conditions change and that routine fails, corrective action may require someone
able to diagnose the failure independently.

This value creates a puzzle. Novel situations are infrequent and unpredictable, whereas familiar situations
dominate everyday activity; and when novelty arrives, the capacity to diagnose it either is already in place
or is not. The use of independent judgment is therefore intermittent, but its maintenance cannot be. Keeping
that capacity alive often requires members to question procedures while those procedures still appear to
work, exposing diligence to being read as opposition. The puzzle is how a group can sustain such
readiness---apparently in defiance of the social approval that coordinates ordinary conduct---through long
stretches in which routine appears sufficient.

This paper's answer is that a questioning culture need not persist in defiance of social approval. The
prevalence of questioning changes how observers interpret it and therefore changes its social cost.
Maintaining judgment leaves visible traces before any novel task arrives: a member requests a justification,
probes a premise, or tests an alternative.\footnote{The mechanism formalizes Mill's idea that visible
nonconformity can reduce the stigma attached to later nonconformity \citep[ch.~III]{Mill1859}.} Observers see
the probe but not its motive. When questioning is common, they interpret it as ordinary diligence and impose
little stigma; maintaining readiness is then worthwhile. When questioning is rare, the same probe is read as
opposition, cooperation is withdrawn, and fewer members maintain judgment. The same desire for approval can
thus sustain either a questioning culture or a conformity trap.

One consequential reading of the mechanism is political. Polities differ persistently, at fixed formal institutions,
in whether citizens publicly ask why; where questioning official routine is rare, the citizen who requests
a justification is read as an oppositionist, and where it is common, as engaged. Sustained stigma on
questioning then does more than silence views: it reaches back to their formation, and after long enough
there is little considered judgment left to suppress---which is why the passivity a regime enforces can
persist long after the regime is gone.\footnote{East Germans exposed to denser Stasi surveillance remain
less civically engaged decades after reunification \citep{LichterLoefflerSiegloch2021}, and regime
exposure durably reshaped what citizens expect of state and society \citep{AlesinaFuchsSchundeln2007};
\citet{BesleyPersson2019} model the coevolution of civic values and institutions that such persistence
suggests. Section~\ref{subsec:outscope} delimits the scope of this reading.} In this reading, repression
is the technology that depletes the inherited stock, and the dynamics explain its aftermath---why the
depletion persists once formal coercion is gone, and why its repair is bounded by turnover.

I formalize this mechanism in an environment with two sequential margins. At the formation margin, agents
decide whether to maintain the capacity to assess established practice independently. At the enactment
margin, those who maintained it decide whether to act on their judgment if a novel state arrives. In the
civic reading, probing is ordinary-times questioning of authority, and enactment is departing from
official routine when it visibly fails. On an investment committee---the compact illustration used
throughout---questioning a proposal's premise is the formation-stage probe, and voting against the
received view in an unusual case is enactment. The same inference problem operates at both
margins: conduct generated by independent judgment is pooled with observationally identical opposition.
Because the audience's withdrawal of cooperation is ex-post optimal---a best response to its posterior at the
point of decision---and agents anticipate it, aggregate
conduct and reputational pressure must be jointly determined. An equilibrium is a fixed point of this
feedback; no hostile preferences or external coordination device is needed for a conformity trap.

The paper makes two contributions. The first gives primitive conditions under which the trap can arise.
When the enactment subgame is single-valued, the two-margin system reduces to a scalar fixed-point problem
organized by a threshold function built from primitives. Theorem~\ref{thm:trap} turns that function into a
sign test: an ordering of its values at two interior stocks is sufficient for multiplicity and, under a
shape condition and transverse crossings, characterizes exactly two stable cultures separated by an
unstable one. Two results then deliver the ordering from primitives: in the fully coupled model, sufficiently
concentrated formation costs whose zero-pressure return clears the median cost somewhere generate an open
coexistence interval, and in a logit--proportional benchmark, a
single primitive inequality is necessary and sufficient for a nonempty multiplicity interval
(Theorem~\ref{thm:analytic-threshold}). A worked example, certified by interval arithmetic, serves as an
existence witness; the sign test is the device that translates primitives into equilibrium counts, not a
result in its own right.

The second contribution is dynamic. Under staggered revision the
maintenance stock becomes an inherited group-level state: current interventions act on what members do
with the existing stock, but the stock itself moves only at revision opportunities. A
monotone path operator constructs the least and greatest perfect-foresight paths from every inherited stock
(Theorem~\ref{thm:history-expectations}). When the stationary map has three cultures and escape is assured
from some inherited stock, the extremal paths define two cutoffs: below the first, escape to the
questioning culture is impossible; between them, it depends on continuation expectations; above the
second, it occurs under every admissible continuation. The cutoffs may coincide; in the certified
benchmark all three regions are nonempty, and Theorem~\ref{thm:band} makes the separation a matter of
theory rather than certification: complex perfect-foresight roots at the unstable culture---a condition
checkable from the slope of the static reduced map and the revision primitives---force the
expectation-dependent region to contain an open band. That complex roots open a history--expectations
overlap is a classical insight \citep{Krugman1991,FukaoBenabou1993}; the theorem establishes it by a
global monotone construction, in a nonsmooth discrete-time system whose complementarity is itself an
equilibrium object, beyond the reach of the classical phase-plane arguments.

Turnover then bounds institutional change: no temporary intervention can rebuild the stock faster than
revision allows. Theorem~\ref{thm:frontier} converts the cutoffs into a physical stock ceiling and minimum
durations---sharp within temporary direct formation support---for making recovery possible or guaranteeing
it. Speed, however, does not determine value: enactment protection releases existing judgment at once, but
the released departures are worth acting on only above a break-even competence threshold, so the impact
sign of protection is separated across stationary inherited cultures by a model-implied stock threshold. Identification and welfare are implications of the same machinery:
under a single exposure regime with unknown recording and background rates, observed conduct cannot
identify how much judgment is unformed, privately withheld, or reputationally suppressed; and comparison
with a same-technology planner yields underprovision only under stated conditions
(Theorem~\ref{thm:central-welfare}), so the paper does not
rank the two decentralized cultures unconditionally.

\medskip

The mechanism builds on canonical models of social interactions and social image, in which aggregate conduct
feeds back into private incentives and reputational returns
\citep{BernheimConformity1994,BenabouTirole2006,BenabouTirole2011,BrockDurlauf2001}, with more recent
extensions in \citet{MichaeliSpiro2017,AliBenabou2020,Braghieri2024,BenabouTiroleLawsNorms2026}, experimental
evidence that a stigmatized act's sanction falls with its prevalence and that bad norms persist and tip
\citep{BursztynEgorovFiorin2020,SmerdonOffermanGneezy2020,AndreoniNikiforakisSiegenthaler2021}, and
pluralistic-ignorance foundations for miscoordinated conformity \citep{FernandezDuque2022}. Further back
stands the critical-mass logic of threshold models \citep{Schelling1978,Granovetter1978}. Multiplicity
through endogenous social meaning is therefore familiar; the new element is an upstream maintenance margin
that carries its own posterior, is observed before enactment, and becomes an inherited group-level state
under staggered revision.

The closest static benchmarks are \citet{SwankVisser2023}, where reputations inferred from deliberation make
preparation efforts strategic complements, and \citet{KaramychevSwank2022}, where social image shapes
information acquisition, discussion, and voting, including information traps. This paper instead separates
keeping a diagnostic capacity alive from later acting on it, and that separation produces the stock--flow
distinction: enactment protection can release the silenced share of the inherited stock on impact but cannot
contemporaneously rebuild it.

A complementary organizational literature studies distortions in what identified members report, vote, or
acquire within a decision episode: conformity and herding under career concerns
\citep{ScharfsteinStein1990,Prendergast1993,Zwiebel1995,Holmstrom1999,OttavianiSorensen2006}, with
supporting evidence that short-record agents conform most \citep{ChevalierEllison1999,HongKubikSolomon2000},
committees in which reputational
incentives shape deliberation and its value \citep{VisserSwank2007,SuurmondSwankVisser2004}, the design of
optimal dissent \citep{LandierSraerThesmar2009}, and transparency's effect on the incentive to acquire
expertise \citep{BarIsaac2012}. Those mechanisms operate through individually attributed records. The margin
here is upstream and anonymous: whether the capacity to diagnose is maintained at all between
episodes---closer to \citet{March1991}'s tension between exploiting settled routines and preserving
exploratory capability, and to culture as shared beliefs and practice \citep{VandenSteen2010}---and what the
audience reads is the population prevalence of that maintained practice, not any individual's record.

Operationally, the formation choice is upstream of \citet{Hirschman1970}'s voice: it concerns maintaining the
capacity that makes expression possible, not merely revealing a view already held. Preference
falsification, self-censorship, and career-concerns conformity
\citep{Loury1994,Kuran1995,Morris2001,Prat2005} distort the use or expression of an existing view; here
sustained stigma reaches back to formation, so that eventually there is little left to falsify.
\citet{Kuran1995} anticipates this feedback informally---persistent falsification eventually corrupts
private knowledge itself (see also \citealp{Kuran1987} on collective conservatism)---and it is formalized
here as an equilibrium object with its own prevalence-priced posterior and inherited stock. The closest
antecedent of the formation margin is the signaling model of \citet{AustenSmithFryer2005}, in which a
costly productive investment is read by peers as a type signal; there, however, the investment's social
meaning is fixed by the signaling equilibrium rather than priced by its own prevalence, and no inherited
stock arises. Acquisition models ask whether information is obtained or communicated
\citep{CheKartik2009,Swank2010,Marshall2019}, and informational cascades conceal information about the
state \citep{BikhchandaniHirshleiferWelch1992}; the hidden object here is instead the readiness to depart
from a failing routine.

The history-versus-expectations question goes back to \citet{Krugman1991} and \citet{Matsuyama1991}; its
game-theoretic form is the perfect-foresight dynamics of \citet{MatsuiMatsuyama1995}, developed in
\citet{HofbauerSorger1999,Oyama2002,OyamaTakahashiHofbauer2008} and, with aggregate shocks,
\citet{BurdzyFrankelPauzner2001}---a program, like its stochastic-evolutionary branch
\citep{KandoriMailathRob1993,Young1993,Young2015}, of \emph{selecting} among static equilibria. This paper
shares the skeleton but not the question: the payoff to revising runs through a Bayesian audience's
posterior rather than a fixed stage game, no selection is imposed---the operator characterizes the entire
set of perfect-foresight paths from each inherited stock---and the resulting cutoffs convert into duration
bounds for temporary policy (Section~\ref{sec:dynamics} makes the comparison precise). In
\citet{AcemogluJackson2015} history operates through the finite window of play new cohorts observe---here
long-lived members adjust only at revision opportunities, so turnover itself bounds temporary
interventions---and the duration bounds give the law-versus-norm tension of \citet{AcemogluJackson2017} a
stock--flow form. Nor is the inherited state \citet{Tirole1996}'s collective reputation, the audience's
belief about past conduct at a known distribution of types, or culture as an equilibrium-selection device
\citep{Kreps1990}: it is the prevalence of a maintained practice, whose social meaning moves with that
prevalence.\footnote{Persistence comes from staggered revision of a maintained practice, not from the
socialization of preferences \citep{BisinVerdier2001}. Persistent-state analogues---individually
maintained quality \citep{BoardMeyerterVehn2013}, organizational stocks and capital
\citep{HalacPrat2016,DesseinPrat2022}---are not the population prevalence of a visible practice
interpreted reputationally. Work on slowly built capabilities, organizational inertia, and cultural
persistence \citep{Chassang2010,GibbonsHenderson2012,GuisoSapienzaZingales2015,DelfgaauwSwank2016}
supplies broader organizational neighbors, and speaking-up research an empirical referent for the
maintenance margin \citep{Edmondson1999}. Motivated-belief, self-image, and echo-chamber mechanisms
instead alter beliefs or the information environment
\citep{Benabou2013,GrossmanVanderWeele2017,LevyRazin2019}.}

Section~\ref{sec:model} presents the environment. Section~\ref{sec:trap} delivers the first contribution, the
primitive conditions for the trap. Section~\ref{sec:dynamics} delivers the second, the dynamics of escape and policy
duration, together with the impact effect of protection. Sections~\ref{sec:hidden}--\ref{sec:welfare} draw out
the diagnostic and welfare implications, and Section~\ref{sec:numerical} verifies the fully coupled benchmark.

\section{The environment}\label{sec:model}

Agents first decide whether to form and maintain independent judgment and, if novelty arrives, whether to act
on it. At each margin, a distinct recorded public event is interpreted by a Bayesian audience whose withdrawal
of cooperation creates a reputational sanction.

\subsection{Tasks, margins, and observability}\label{subsec:env}

\paragraph{Tasks and states.}
A unit continuum of agents confronts recurring task episodes. Within each episode the state $\omega$ is
aggregate: all agents face the same task environment. Let $p\in(0,1)$. With probability $1-p$, the state is
\emph{familiar}: an established routine $r$ is efficient. With probability $p$, it is \emph{novel}: the
routine is inappropriate and considered judgment prescribes $a^\ast(\omega)\neq r$. Thus $p$ is the group's
long-run frequency of novel episodes. Novelty creates a wedge between considered judgment and socially
settled conduct.

\paragraph{Two margins of judgment.}
Before the state is realized, an agent chooses whether to form and maintain the disposition to assess
established practice through self-endorsed judgment, $e\in\{0,1\}$. The idiosyncratic net cost of doing so is
$k$, with strictly increasing c.d.f.\ $H$ on $\mathbb R$ and continuous density $h>0$. The cost includes the
effort of maintaining that disposition net of any intrinsic value attached to it and may therefore be
negative. Let $J$ denote the mass of agents with $e=1$, which we call \emph{formed judgment}.

In a familiar state, formed judgment endorses $r$, and all agents take the same action. In a novel state, a
formed agent identifies $a^\ast(\omega)$ as appropriate but must still decide whether to depart from the
routine. Her idiosyncratic resolve cost $\ell$ is drawn after formation, independently of $k$, from c.d.f.\
$G$ on $\mathbb R$ with continuous density $g>0$. Let $A\leq J$ denote the mass of formed agents who depart in
a novel state, which we call \emph{enacted judgment}.

All primitives and distributions are common knowledge. The two margins are distinct: formation determines whether considered judgment is
available, whereas enactment determines whether the agent remains faithful to it. The benchmark private
motive is fidelity to a considered judgment rather than confidence that the judgment is correct.

\paragraph{Timing and diagnostic technology.}
An agent's formation status $e$ is in place before the task is known. Nature draws $\omega$, and agents
observe a diagnostic cue $y$ before the episode is recognized as familiar or novel. A formed agent applies a
diagnostic rule when $y$ belongs to a set $\mathcal D$ and incurs a real cost $c_D>0$. This activity may
generate a \emph{diagnostic probe}: a public request to justify a routine, a challenge to one of its premises,
a small experiment, or another low-stakes test of established practice. On a committee, for example, a probe
is a tabled question asking a sponsor to defend the rationale for a proposal.\footnote{Related evidence
documents avoidance of help-seeking when exposing informational need can threaten standing
\citep{RyanPintrichMidgley2001}.}

In a familiar episode, define
\begin{equation}\label{eq:gamma-diagnostic}
 \zeta:=\Pr(y\in\mathcal D\mid\omega\text{ is familiar})>0,
 \qquad
 \gamma:=c_D\zeta,
 \qquad
 \phi_0:=\zeta\bar\phi,
\end{equation}
where $\bar\phi\in(0,1]$ is the conditional probability that a diagnostic probe is recorded as a public
event. Thus $\gamma$ is the expected familiar-state diagnostic cost and $\phi_0$ is raw probe visibility.
The same diagnostic frequency $\zeta$ links the real cost of scrutiny to its public exposure. Because the
diagnosis occurs before the task is recognized as familiar, its cost has already been incurred when the agent
learns that the routine fits.

Novel-state diagnostic costs do not vary with the aggregates $(J,A)$, and the baseline normalizes them to
zero, so the net formation cost $k$ is an ex-ante primitive whose distribution $H$ is invariant to $p$:
comparative statics in $p$ vary the frequency of novelty at fixed formation-cost
primitives.\footnote{With a per-novel-episode diagnostic cost $c_N>0$ instead, the formation return
\eqref{eq:invest} would read $p[\Omega_P(A)-c_N]-(1-p)[\gamma+\phi s_f]$, and
$\widehat R_p=\Omega_P(a)-c_N+\gamma+\phi s_f$ would be positive if{}f
$\Omega_P(a)+\gamma+\phi s_f>c_N$. The boundary case $\zeta=0$ corresponds to perfect ex-ante recognition
and eliminates both the familiar-state diagnostic cost and this formation-stage event channel.}

\paragraph{Public recording and information.}
In a familiar episode, a recorded diagnostic probe generates the public event $x_f=1$. Once the task context
becomes recognizable, a formed agent in a novel state chooses whether to depart from the routine. A
consummated departure is registered as the public event $x_a=1$ with source-neutral probability
$\psi\in(0,1]$.

For accounting, formation-stage probes are recorded only in familiar episodes. In a novel episode,
diagnostic activity is not independently attributable as a separate event; only a registered consummated
departure generates $x_a=1$. The audience observes the event and whether its task context is familiar or
novel, but it observes neither the actor's motive nor $a^\ast(\omega)$ and cannot link personal histories
across the two margins. Probing therefore precedes the recognition of novelty, whereas enactment is
conditional on novelty.

\subsection{Bayesian beliefs and a reduced-form sanction schedule}\label{subsec:audsanction}

\paragraph{Authentic and background sources.}
At each margin $m\in\{f,a\}$, underlying independent acts come from two sources. Authentic acts arise at rate
$\alpha_mX_m$, where $\alpha_m>0$ is a visibility normalization and $X_f=J$, $X_a=A$.
A background source of oppositional or attention-seeking conduct generates acts at rate $\nu_m>0$. This
source affects event frequencies and audience beliefs but lies outside the unit continuum of agents over
which private choices, output, and welfare are defined. The exclusion has an economic reading, which is
also why the two sources respond differently to stigma: background acts come from outsiders to the
relational network---actors with no standing or future cooperation at stake with this audience---so the
reputational sanction, which operates entirely through withdrawn cooperation, has little purchase on them.
The audience need not distinguish insiders from outsiders for this to be coherent: it withdraws
cooperation from whoever performed the act, and the withdrawal simply carries no cost for an actor with no
relational capital at stake.
At the formation margin both sources generate
familiar-state probes; at the enactment margin both generate departures conditional on a novel-state
context.

Recording is source-neutral within each margin: formation events are recorded with raw probability $\phi_0$
and enactment events with probability $\psi$. The audience observes the event and its task context but not
its source. Let $o$ denote the event that an observed act comes from the background source. Conditional on
observing $x_m=1$, Bayes' rule gives
\begin{equation}\label{eq:posterior}
 q_m(X_m):=\Pr(o\mid x_m=1;X_m)
 =\frac{\nu_m}{\nu_m+\alpha_mX_m}.
\end{equation}
The posterior measures how suspicious an observed independent act appears. It is high when authentic
conduct is rare and falls as authentic prevalence increases.

Observed event frequencies are $\phi_0(\nu_f+\alpha_fJ)$ at formation and $\psi(\nu_a+\alpha_aA)$ at
enactment. Because recording is source-neutral, $\phi_0$ and $\psi$ cancel
from the posterior. They nevertheless govern an authentic agent's exposure to the corresponding sanction.
The normalization $\alpha_m$ scales aggregate authentic activity; only the ratio $\nu_m/\alpha_m$ matters
for beliefs, and the baseline sets $\alpha_m=1$ by absorbing its units into $\nu_m$.\footnote{If background acts have relative recording probability $\chi>0$, the same formulas apply
after replacing $\nu_m$ by $\chi\nu_m$. A common change in the visibility of both sources changes exposure
but not the posterior.}

Inference is local to each margin, and in equilibrium $X_f=J$ and $X_a=A$ are the prevalences generated by
agents' own choices. Positive $\nu_m$ supplies residual ambiguity and keeps observed events on path. When
$\nu_m=0$, the posterior of a background source is zero for every $X_m>0$, and the reputational-inference
channel disappears away from the boundary.

\paragraph{Vindication and effective exposure.}
Standing is withdrawn at the matching stage, before an act's consequences are realized and attributed to its
author. Let $v_f\in[0,1]$ be the probability that a recorded probe is promptly and individually vindicated,
in which case the resulting standing loss is restored. Define unresolved formation exposure as
\begin{equation}\label{eq:effective-formation-exposure}
 \phi:=(1-v_f)\phi_0.
\end{equation}
Observed probe frequency continues to depend on raw visibility $\phi_0$, whereas formation incentives depend
on unresolved exposure $\phi$.

\paragraph{Matching and reputational sanctions.}
After an independent act, potential partners with heterogeneous withdrawal costs decide whether to retain
or exclude its author; the aggregate standing loss therefore rises with the posterior $q_m(X_m)$ and falls
with authentic prevalence $X_m$.

Let $s_m(X_m;\theta_m)\geq0$ denote the \emph{reputational sanction schedule} at margin $m$: the expected loss
of standing and future cooperation borne by an authentic actor conditional on her independent act being
observed. Like every payoff in the model, $s_m$ is measured in units of a single episode's output, and it is
a capitalized value: the present value of the standing and cooperation forgone after the act, charged when
the act is observed, not a recurring per-period levy. The parameter $\theta_m$ is a generic index of the pressure applied to suspicious conduct. The next
section derives this schedule from decentralized withdrawal decisions. The analysis here requires only the
following monotonicity discipline.

\begin{assumption}[Reputational monotonicity]\label{ass:socialmeaning}
For each margin $m$, $s_m(X_m;\theta_m)$ is continuous and weakly decreasing in $X_m$. On every smooth
sanctioned interior branch used for derivative statements, $s_m>0$, $s_m'<0$, and
$\partial s_m/\partial\theta_m>0$. Flat tolerated or capped regions are allowed; derivative comparative
statics involving $\theta_m$ are then interpreted on the interior branch, while lattice and monotonicity
results use weak monotonicity only.
\end{assumption}

\subsection{The two decision problems}\label{subsec:decisions}

Given the sanction schedules, the individual problems are solved backward: enactment conditional on formed
judgment, followed by formation conditional on enactment's option value.

\paragraph{Enacting judgment.}
In a novel state, following one's judgment produces private benefit $b_N>0$ relative to the routine. The
stake $b_N$ is measured, like every payoff, in units of one episode's output: it is the conviction
stake---what standing by a considered judgment is worth to the agent in the episode at hand. A
consummated departure is registered with probability $\psi$, so the agent acts on her judgment if and only
if $\ell\leq z(A):=b_N-\psi s_a(A;\theta_a)$.
The cutoff is the private value of following a considered view net of the expected standing loss generated by
a registered departure.

Define the option-value function $\Omega(x):=\E[(x-\ell)_+]$,
which is assumed finite for every finite $x$. The conditional enactment rate and the private option value of
formed judgment are $w(A):=G(z(A))$ and
$\Omega_P(A):=\Omega(z(A))=\int_{-\infty}^{z(A)}G(u)\,du$.
Thus $w(A)$ is the share of formed agents who enact, while $\Omega_P(A)$ is the ex-ante value of having the
option to depart after the resolve cost is known. Aggregate enacted judgment must be consistent with the
prevalence that determines its social interpretation:
\begin{equation}\label{eq:AS}
 A=Jw(A).
\end{equation}

\paragraph{Forming judgment.}
On a familiar task, diagnosis forgoes the economy of relying immediately on settled practice. Its expected
real cost is $\gamma$, while a recorded probe exposes an authentic agent to the formation-stage sanction with
unresolved probability $\phi$. The familiar-state frequency $1-p$ scales both channels. An agent forms and
maintains judgment if and only if
\begin{equation}\label{eq:invest}
 e=1\iff k\leq R(J,A;p):=
 \underbrace{p\Omega_P(A)}_{\text{option value of standing by judgment}}
 -\underbrace{(1-p)\bigl[\gamma+\phi s_f(J;\theta_f)\bigr]}_
 {\text{diagnostic and reputational cost}}.
\end{equation}
Consequently, formed judgment satisfies
\begin{equation}\label{eq:fp}
 J=H(R(J,A;p)).
\end{equation}
Equations \eqref{eq:AS} and \eqref{eq:fp} jointly determine the stationary pair $(J,A)$.

The binary formation decision has a maintained-practice interpretation.

\begin{proposition}[Formation as maintained practice]\label{prop:practice}
Suppose each agent chooses in every period whether to maintain diagnostic practice: maintaining means
applying the diagnostic rule on cued episodes, incurring the per-period net cost $k$ and generating the
familiar-state cost $\gamma$ and raw exposure $\phi_0$ of \eqref{eq:gamma-diagnostic}, and the ability to
identify $a^\ast(\omega)$ in a novel episode is available in period $t$ if and only if practice is
maintained in period $t$. If the aggregates $(J,A)$ are stationary and agents are atomistic, the optimal
policy is stationary: maintain if and only if $k\le R(J,A;p)$, as in \eqref{eq:invest}, with the
indifferent type of zero mass. The formation margin is therefore the reduced form of currently maintained
diagnostic readiness.
\end{proposition}

The proof is in Online Appendix OA.13. The proposition gives $J$ the interpretation of current readiness;
the dynamic extension introduces infrequent revision opportunities and thereby persistence in its aggregate
prevalence.

Let $w^0:=G(b_N)$ denote enactment in the absence of reputational pressure. Formed judgment decomposes as
\begin{equation}\label{eq:decomposition}
 J=
 \underbrace{A}_{\text{enacted judgment}}
 +\underbrace{J[w^0-w(A)]}_{S^R:\ \text{reputationally silenced}}
 +\underbrace{J(1-w^0)}_{N:\ \text{privately not enacted}}.
\end{equation}
Among agents who have formed a judgment, $A$ enact it, $S^R$ are deterred specifically by anticipated
reputational sanctions, and $N$ would refrain even in the absence of an audience. The decomposition therefore
separates social suppression from non-enactment caused by the real resolve cost $\ell$.

\subsection{Output, institutional interpretation, and scope}\label{subsec:outscope}

\paragraph{Output.}
The positive benchmark treats a formed agent's judgment as accurate; a later extension allows accuracy to
depend on enacted judgment. On a familiar task, an agent who relies immediately on the routine produces
$V_F+\gamma$, whereas an agent who maintains diagnostic readiness produces $V_F$. Here $\gamma$ is the same
diagnostic cost that enters the formation decision and is counted once in output. Hence
\[
 Y_F(J)=V_F+\gamma(1-J).
\]
In a novel state, only enacted judgment changes conduct, so
\[
 Y_N^C(A)=AV_N-(1-A)D,
 \qquad
 Y^C(J,A;p)=(1-p)Y_F(J)+pY_N^C(A),
\]
where $V_F,V_N>0$, $D\geq0$, and the superscript $C$ denotes the competent benchmark. These measures are
defined over the unit continuum and exclude the background event-generating source. The rate of adaptive
action is $A$; equation \eqref{eq:decomposition} identifies how much non-enactment is attributable specifically
to reputational pressure.

\paragraph{Institutional interpretation.}
A motivating civic interpretation is a polity's culture of questioning: a probe is a public request to
justify established practice, and enactment is a departure from routine when the routine visibly fails.
The audience is the succession of neighbors, officials, and counterparties encountered in a fluid society;
what matters for anonymous matching is record linkage---current identity visible enough for standing to be
withdrawn, past acts difficult to associate with the same individual. The background source requires a
real institutional counterpart---an actual persistent stratum of oppositional or attention-seeking public
conduct, of the kind a departed regime's residues, factional contention, or notoriety seeking
supply---since the audience is correctly Bayesian and its suspicion must be justified in equilibrium.

The same architecture applies to rotating committees and other organizations when evaluators have limited
access to personal histories; persistent cross-unit differences in whether established practice is
questioned at all, at fixed task and technology, are the cross-sectional pattern this multiplicity predicts
\citep{Edmondson1996}. Stable teams and named committees with persistent individual records lie
closer to a career-concerns environment, in which linked histories provide information absent from the
anonymous-matching benchmark. In the political application, the relevant environment begins after formal
coercion has ended but decentralized inference and social withdrawal remain; centralized, identity-linked
surveillance belongs to a different interaction structure.

\paragraph{Scope: prompt vindication and task recognition.}
The baseline sets $v_f=0$: standing is withdrawn before a probe's consequences are observed; when prompt
vindication is possible, it lowers unresolved exposure $\phi$ in \eqref{eq:effective-formation-exposure}
and weakens the formation-stage channel. An enactment analogue $v_a$ would enter exactly as
$\psi(1-v_a)$; that the formal treatment carries a vindication parameter and a certified bound at
formation (Corollary~\ref{cor:anonymity}) but neither at enactment is motivated, though not forced, by the
timing---a departure's consequences take longest to attribute precisely when the state is novel, so the
learning technology of Online Appendix OA.1 routes outcome signals to successors' shared model rather than
to the audience sanctioning the current act.\footnote{How much slow or noisy outcome-based updating the
three-culture configuration survives is a quantitative question the certificates do not answer and is left
open.}

Recognizing that a task is unsettled does not reveal whether a particular departure reflects judgment or
opposition: the audience observes the context but not the appropriate action or the actor's motive. The
next section derives the sanction schedules from the audience's ex-post matching problem and
closes the stationary equilibrium.

\section{When the conformity trap arises}\label{sec:trap}

This section derives reputational sanctions from decentralized matching, defines cultural equilibrium, and
reduces the two-margin system to a scalar formation map. It then characterizes coexistence through
formation-pressure thresholds and provides a global characterization in a logit--proportional benchmark.

\subsection{The Bayesian audience and endogenous tolerance}\label{subsec:audience}

\paragraph{The audience's decision.}
At each margin $m\in\{f,a\}$, the audience observes the event $x_m=1$, holds the posterior
$q_m(X_m)$ in \eqref{eq:posterior}, and decides how much cooperation to withdraw. Because the posterior
depends on the true prevalence $X_m$, audience sanctions and agents' choices must be jointly consistent.

After observing an independent act, potential partners decide whether to retain or exclude its author.
Excluding an oppositional actor avoids a bad match, whereas excluding an authentic actor sacrifices useful
cooperation. Let $\pi_m>0$ denote the per-partner value of excluding an oppositional actor and
$\Lambda_m>0$ the loss from excluding an authentic actor. These are relational values of cooperation and
standing.
Partner $u\in[0,\bar s_m]$ incurs withdrawal cost $c_mu$, where $c_m>0$, and withdraws if and only if
\[
 q\pi_m-(1-q)\Lambda_m\geq c_mu.
\]
Anonymous matching means that partners observe the current event but cannot condition on a sufficiently
informative personal history.\footnote{As individual records become more informative, authentic and
opportunistic conduct become easier to distinguish, attenuating the posterior channel
\citep{ElyValimaki2003}. Online Appendix OA.8 formalizes this attenuation.}

If a mass $s$ of partners withdraws, integrating their heterogeneous costs gives the aggregate audience
payoff
\begin{equation}\label{eq:audience-payoff}
 U_m(s;q)
 =
 q\pi_m s-(1-q)\Lambda_m s-\frac{c_m}{2}s^2,
 \qquad
 s\in[0,\bar s_m].
\end{equation}
The quadratic term is the aggregate withdrawal cost
$\int_0^s c_mu\,du$. Ex-post optimization yields the exact decentralized sanction rule
\begin{equation}\label{eq:endogenous-tolerance}
 W_m(q):=\argmax_{s\in[0,\bar s_m]}U_m(s;q)
 =
 \left[
 \frac{(\pi_m+\Lambda_m)q-\Lambda_m}{c_m}
 \right]_{0}^{\bar s_m},
\end{equation}
where $[x]_0^{\bar s}:=\min\{\bar s,\max\{0,x\}\}$.
Define the \emph{marginal sensitivity} $\kappa_m:=(\pi_m+\Lambda_m)/c_m$ and the \emph{tolerance
threshold} $\bar q_m:=\Lambda_m/(\pi_m+\Lambda_m)$. The resulting prevalence-dependent sanction is
$s_m^W(X_m)=\bigl[\kappa_m\bigl(q_m(X_m)-\bar q_m\bigr)\bigr]_0^{\bar s_m}$.
The audience tolerates an independent act when $q_m(X_m)\leq\bar q_m$ and withdraws increasingly above that
threshold, up to the cap $\bar s_m$. At the formation margin, we call $\kappa_f$ \emph{formation pressure}
and $\phi\kappa_f$ \emph{effective formation pressure}.

\begin{proposition}[Endogenous tolerance]\label{prop:tolerance}
Define $s_m^W(X_m):=W_m(q_m(X_m))$, and call a region \emph{interior} when $0<s_m^W<\bar s_m$ holds on it.
The exact Bayesian sanction is weakly decreasing in $X_m$ and strictly decreasing on every interior
region, with
\[
 \frac{d s_m^W}{dX_m}
 =
 -\frac{\pi_m+\Lambda_m}{c_m}
 \frac{\nu_m\alpha_m}{(\nu_m+\alpha_mX_m)^2}
 <0.
\]
Hence the exact rule satisfies Assumption~\ref{ass:socialmeaning} on every sanctioned interior branch and
weakly satisfies it globally.

Moreover, for any two prevalences $X_h>X_\ell$, every tolerance threshold
\[
 \bar q_m\in\bigl(q_m(X_h),q_m(X_\ell)\bigr)
\]
induces a strictly positive sanction at $X_\ell$ and zero sanction at $X_h$ under the same audience
primitives.
\end{proposition}

\begin{remark}[Relational standing]
The audience prices the future relationship, not the present act: $\pi_m$ and $\Lambda_m$ are fixed
primitives of the match. A competence-aware audience would let $\Lambda_a$ rise with the accuracy
$\sigma(A)$ of Section~\ref{subsec:sequencing}, so that excluding an authentic dissenter costs less exactly
where dissent is least competent; Online Appendix OA.6 shows this strengthens the prevalence feedback
rather than weakening it, preserves Assumption~\ref{ass:socialmeaning}, and at full proportional strength
leaves the benchmark's three cultures with a slightly deeper trap and a wider trap basin ($J_i$ rises from
$\Jmid$ to $\compJmid$).
\end{remark}

\paragraph{Smooth sanctions.}
Some arguments require global differentiability. For these, define the softened Bayesian schedule
\begin{equation}\label{eq:sanction}
 s_m^\varepsilon(X_m;\kappa_m,\bar q_m,\bar s_m)
 :=
 \bar s_m
 \left(
 1-\exp\left\{
 -\frac{\kappa_m[q_m(X_m)-\bar q_m]_\varepsilon}{\bar s_m}
 \right\}
 \right),
\end{equation}
where $[x]_\varepsilon:=\varepsilon\log(1+\exp(x/\varepsilon))$.
For finite $\bar s_m$, this is a smooth saturating sanction. In the no-cap limit,
\[
 s_m^\varepsilon(X_m)
 =
 \kappa_m[q_m(X_m)-\bar q_m]_\varepsilon,
\]
which converges as $\varepsilon\downarrow0$ to the uncapped hard-tolerance rule. The softening can be
interpreted as small idiosyncratic heterogeneity in audience tolerance.

\begin{lemma}[Smooth Bayesian sanctions]\label{lem:reputation}
The no-cap limit of \eqref{eq:sanction} satisfies
Assumption~\ref{ass:socialmeaning} globally and has $s_m''>0$. Thus its sanction is strictly decreasing in
authentic prevalence and convex on the physical domain.
\end{lemma}

Generic results below require only Assumption~\ref{ass:socialmeaning}. Derivative statements use smooth
sanctioned branches. Hard tolerated or capped regions retain weak monotonicity but may replace smooth folds
with border cases.

\paragraph{Cultural equilibrium.}
To distinguish a sanction schedule from its equilibrium realization, write $S_m(X_m;\theta_m)$ for the
schedule and $s_m:=S_m(X_m;\theta_m)$ for its value at the equilibrium prevalence.

\begin{definition}[Cultural equilibrium]\label{def:equilibrium}
Fix sanction schedules $S_m(X_m;\theta_m)$ satisfying
Assumption~\ref{ass:socialmeaning}. A \emph{cultural equilibrium relative to those schedules} is a tuple
$(J,A,q_f,q_a,s_f,s_a)$ satisfying:

\emph{(i) Enactment:} $A=J\,G(b_N-\psi s_a)$,
so a formed agent departs in a novel state if and only if
$\ell\leq b_N-\psi s_a$.

\emph{(ii) Formation:} $J=H(R(J,A;p))$.

\emph{(iii) Bayesian beliefs:}
$q_m$ is the posterior \eqref{eq:posterior} evaluated at the true prevalences
$X_f=J$ and $X_a=A$.

\emph{(iv) Sanction consistency:} $s_m=S_m(X_m;\theta_m)$; in the exact Bayesian matching game
$S_m=W_m\circ q_m$, and under the smooth specification
$S_m=s_m^\varepsilon(\cdot\,;\kappa_m,\bar q_m,\bar s_m)$.
\end{definition}

Under the exact rule, the last two conditions impose correct Bayesian beliefs and an ex-post optimal audience
response; because $\nu_m>0$ keeps every event on path, no off-path belief selection is involved. Under the
smooth rule, they define the corresponding regularized fixed point. This is a stationary
population equilibrium with rational expectations.

\subsection{The equilibrium map and scalar reduction}\label{subsec:map}

Define the simultaneous best-response map $\mathcal T:[0,1]^2\to[0,1]^2$ by
\begin{equation}\label{eq:2dmap}
 \mathcal T(J,A):=
 \begin{pmatrix}
 H(R(J,A;p))\\
 JG(z(A))
 \end{pmatrix}.
\end{equation}
Both coordinates are weakly increasing in both arguments and strictly increasing in the relevant directions
on sanctioned interior branches. Greater formation lowers formation stigma; greater enactment lowers
enactment stigma, increases current enactment, and raises the option value of forming judgment. The two
margins therefore generate connected strategic complementarities
\citep{MilgromRoberts1990,Vives1990,CooperJohn1988,VanZandtVives2007}.

To reduce this two-dimensional system to one scalar equation, impose the following regularity condition on
the enactment subgame:
\begin{equation}\label{eq:enactment-regularity}
 \mathcal C_a:=
 \sup_{A\in(0,1]}
 \frac{A g(z(A))}{G(z(A))}
 \,\psi[-s_a'(A)]<1.
\end{equation}
The condition limits the feedback through which greater enactment lowers stigma and thereby induces still
more enactment. For a hard rule, it is imposed on both one-sided derivatives at a kink.

For every $J>0$, the map $A\mapsto JG(z(A))-A$ is positive at $A=0$ and nonpositive at $A=1$. Condition
\eqref{eq:enactment-regularity} makes $A/G(z(A))$ strictly increasing, so the enactment equation has a unique
solution $A=a(J)$, with $a(0)=0$.
Where differentiable,
\begin{equation}\label{eq:aprime}
 a'(J)
 =
 \frac{w(a(J))}
 {1-Jg(z(a(J)))\psi[-s_a'(a(J))]}
 >0.
\end{equation}
Define the induced formation return $\widehat R(J):=R(J,a(J);p)$ and the reduced formation map
$\Psi(J):=H(\widehat R(J))$.

\begin{proposition}[Two-margin equilibrium and multiplicity]\label{prop:mult}
The equilibrium set of \eqref{eq:2dmap} is nonempty and forms a complete lattice. Hence minimal and maximal
judgment cultures exist.

Under \eqref{eq:enactment-regularity}, two-margin equilibria are in one-to-one correspondence with fixed
points of the increasing scalar map $\Psi$. The equilibrium is unique if $\Psi$ is Lipschitz with modulus
below one; on a smooth branch, $\sup_J\Psi'(J)<1$ is sufficient.

Suppose $\Psi'-1$ has exactly two zeros
$\underline J<\overline J$, is negative outside
$(\underline J,\overline J)$, and positive inside. Then exactly three equilibria exist if and only if
\[
 \Psi(\underline J)-\underline J
 <0<
 \Psi(\overline J)-\overline J.
\]
The two outer equilibria are scalar-stable ($\Psi'<1$) and the middle equilibrium scalar-unstable
($\Psi'>1$).
\end{proposition}

Here \emph{stable} and \emph{unstable} refer only to iteration of the reduced map $\Psi$; they do not describe
forward-looking dynamics.\footnote{Global-games perturbations or aggregate shocks could select among
equilibria \citep{MorrisShin2003,FrankelPauzner2000}. The analysis here characterizes the equilibrium set.}

Prevalence dependence is essential for the conformity feedback. Constant sanctions can lower judgment but
cannot make the social interpretation of independence depend on how common independent conduct is.

\begin{proposition}[Fixed disclosure costs do not create cultural traps]\label{prop:nosocialmeaning}
If $s_f(J)\equiv s_f^0$ and $s_a(A)\equiv s_a^0$ with $s_f^0,s_a^0\geq0$,
the equilibrium is unique: $\bar z=b_N-\psi s_a^0$, $\bar w=G(\bar z)$,
\[
 J^0=
 H\!\left(
 p\Omega(\bar z)
 -(1-p)(\gamma+\phi s_f^0)
 \right),
 \qquad
 A^0=\bar wJ^0.
\]
The reduced formation response is constant in $J$, so fixed disclosure costs can depress judgment but cannot
generate a prevalence-driven fold or an unstable threshold.
\end{proposition}

The proof is in Online Appendix OA.13. When exactly three fixed points exist, denote them by
\[
 J_\ell<J_i<J_h,
 \qquad
 A_j:=a(J_j).
\]
The low fixed point is the conformity trap, the middle point is the threshold culture, and the high fixed
point is the questioning culture. Figure~\ref{fig:cultures} illustrates the reduced response.

\begin{figure}[t]
\centering
\caption{Two-margin judgment at $p=0.10$.}
\label{fig:cultures}
\includegraphics[width=0.62\textwidth]{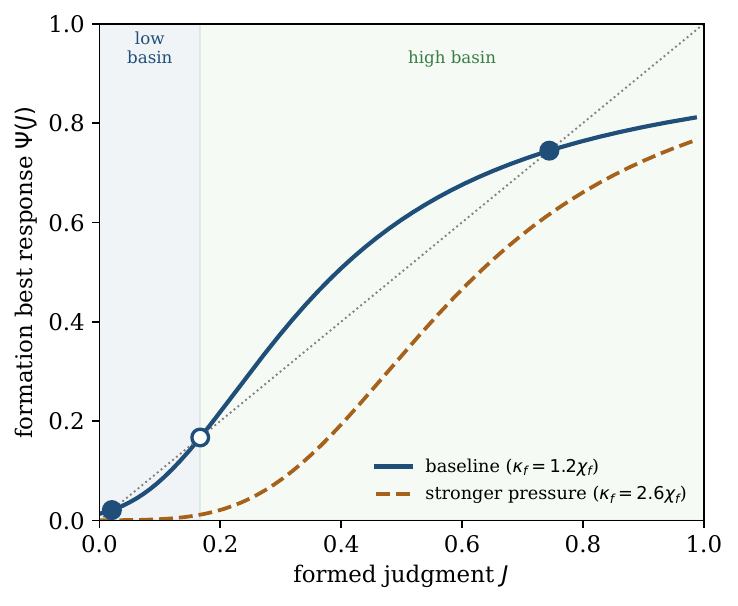}\\[4pt]
\begin{minipage}{0.88\textwidth}
\footnotesize\emph{Notes:}
The reduced formation response $\Psi(J)$, after substituting the enactment solution $A=a(J)$, crosses the
$45^\circ$ line at two scalar-stable cultures (filled) separated by a scalar-unstable threshold (open).
Greater formation pressure shifts the reduced response downward.
\end{minipage}
\end{figure}

\subsection{Pressure thresholds and coexistence}\label{subsec:trap}

Consider a smooth interior formation-sanction branch of the form $s_f(J;\kappa_f)=\kappa_f b_f(J)$ with
$b_f(J)>0$ and $b_f'(J)<0$.
Here $b_f$ is the base sanction schedule. It equals $[q_f]_\varepsilon$ under the zero-threshold no-cap smooth
rule and $q_f-\bar q_f$ on an interior hard-tolerance branch; $\bar q_f=0$ arises as the limit
$\Lambda_f\downarrow0$ of the matching microfoundation, which itself delivers $\bar q_m\in(0,1)$.

Holding enactment primitives fixed, let $a(J)$ be the unique regular enactment solution and define the
formation return at zero formation pressure,
$B(J;p):=p\,\Omega\!\left(b_N-\psi s_a(a(J))\right)-(1-p)\gamma$.
For $\phi>0$ and $J\in(0,1)$, define
\begin{equation}\label{eq:pressure-threshold}
 \mathcal K(J;p):=
 \frac{B(J;p)-H^{-1}(J)}
 {(1-p)\phi b_f(J)}.
\end{equation}
The threshold $\mathcal K(J;p)$ is the formation-pressure level at which $J$ lies exactly on the reduced
formation response. Since $H$ is strictly increasing and the denominator is positive,
\begin{equation}\label{eq:psi-K}
 \Psi(J;\kappa_f)>J
 \quad\Longleftrightarrow\quad
 \kappa_f<\mathcal K(J;p).
\end{equation}
Pressure below the threshold places the best response above the diagonal; pressure above it places the best
response below.

\begin{definition}[Pressure-threshold pair]\label{def:pressure-threshold}
Fix $p\in(0,1)$ and $\phi>0$, and suppose the enactment subgame satisfies
\eqref{eq:enactment-regularity}. Suppose the multiplicative formation branch is valid on an interval
containing $0<J_L<J_H<1$. The pair $(J_L,J_H)$ is a \emph{pressure-threshold pair} at $p$ if
$\mathcal K(J_H;p)>\max\{0,\mathcal K(J_L;p)\}$; its induced positive-pressure interval is
$\mathcal I_{LH}(p):=\bigl(\max\{0,\mathcal K(J_L;p)\},\,\mathcal K(J_H;p)\bigr)$.
\end{definition}

For exact root counts, use the following three-piece threshold-shape condition: $\mathcal K$ is continuously
differentiable and there exist $0<J_-<J_+<1$ such that
$\mathcal K'<0$ on $(0,J_-)$, $\mathcal K'>0$ on $(J_-,J_+)$, and $\mathcal K'<0$ on $(J_+,1)$.
Under this condition, every horizontal positive-pressure line intersects $\mathcal K$ at most three times
(Proposition~\ref{prop:atmost}).

\begin{theorem}[Pressure-threshold sign test for coexistence]\label{thm:trap}
Fix $p\in(0,1)$ and $\phi>0$, and suppose the enactment subgame satisfies
\eqref{eq:enactment-regularity}. If the formation-sanction branch admits a pressure-threshold pair
$(J_L,J_H)$ at $p$ and is valid at $J_L$ and $J_H$ for every pressure in
$\mathcal I_{LH}(p)$---automatic for uncapped rules, since their multiplicative form holds at every
pressure---then every
\[
 \kappa_f\in\mathcal I_{LH}(p)
\]
generates at least three fixed points of the reduced map $\Psi(\cdot;\kappa_f)$.

If, in addition, the multiplicative branch is valid on all of $(0,1)$ at the pressure under
consideration---so that every fixed point lies on it, as for uncapped rules---then under the three-piece
threshold-shape condition there are exactly three fixed points. If the crossings are
transverse, the outer two are scalar-stable and the middle one scalar-unstable. At any
endpoint of a maximal multiplicity interval satisfying
\[
 \Psi(J^\ast;\kappa_f^\ast)=J^\ast,
 \qquad
 \Psi_J(J^\ast;\kappa_f^\ast)=1,
\]
\[
 \Psi_{JJ}(J^\ast;\kappa_f^\ast)\neq0,
 \qquad
 \Psi_{\kappa_f}(J^\ast;\kappa_f^\ast)\neq0,
\]
the entering or exiting pair has the nondegenerate saddle-node form.

Conversely, suppose that at some $\kappa_f>0$, on a smooth multiplicative branch valid on an interval
containing them, there are three transverse
fixed points with the stable--unstable--stable scalar pattern. Then there exists a pressure-threshold pair
$(J_L,J_H)$ such that
\[
 \kappa_f\in\mathcal I_{LH}(p).
\]
Thus a pressure-threshold pair is sufficient for at least three fixed points and necessary for a transverse
stable--unstable--stable configuration. Together with the threshold-shape condition, it characterizes the
exact transverse three-root configuration.
\end{theorem}

The proof in Appendix~\ref{app:logit-geometry} applies the intermediate value theorem to the signs delivered by
\eqref{eq:psi-K}; full support of $H$ supplies the boundary signs
$\Psi(0)>0$ and $\Psi(1)<1$. The local saddle-node statement follows from the normal-form argument in
Lemma~\ref{lem:regular-fold}.

\begin{corollary}[Private or vindicated formation closes the formation-stigma channel]
\label{cor:anonymity}
On a multiplicative formation-sanction branch, raw visibility $\phi_0$, prompt vindication $v_f$, and
formation sensitivity $\kappa_f$ enter formation incentives only through
$\phi\kappa_f=(1-v_f)\phi_0\kappa_f$, because $\phi s_f(J)=\phi\kappa_fb_f(J)$.
For fixed $\phi$, $\mathcal I_{LH}(p)$ is an interval in $\kappa_f$; the corresponding interval in effective
formation pressure is $\phi\mathcal I_{LH}(p)$.
A lower recording probability $\bar\phi$ at fixed diagnostic frequency $\zeta$---which lowers $\phi_0$
without changing the real diagnostic cost $\gamma$---greater prompt vindication, and lower formation
sensitivity are therefore equivalent for stationary formation incentives. As
$v_f\uparrow1$ or $\phi_0\downarrow0$, effective formation pressure vanishes.\footnote{At the worked
benchmark, $\phi_0=\phizero$ and $\kappa_f/\chi_f=1.2$, where $\chi_f$ is the normalization scale used to
report formation pressure; the certified lower pressure fold is $\mathcal K(J_-)/\chi_f\approx\Kminchi$.
Holding the other primitives fixed, the three-root conclusion survives for
$0\le v_f<1-\Kminchi/1.2\approx\vfbound$, with the low and middle roots meeting at the nondegenerate fold
at the endpoint.}
\end{corollary}

When $\phi=0$, the formation-stigma channel disappears, but enactment stigma may still generate
prevalence feedback through
$\Psi_0(J)=H\!\left(p\,\Omega_P(a(J))-(1-p)\gamma\right)$.
Multiplicity can therefore survive private or fully vindicated formation. Uniqueness additionally follows
if $\sup_J\Psi_0'(J)<1$.

Institutionally, a mandatory pre-mortem makes probing universal and hence uninformative about motive, the
logic of Proposition~\ref{prop:nosocialmeaning}; anonymous petition and review channels lower raw probe
visibility $\phi_0$, the margin of Corollary~\ref{cor:anonymity}; and the secret ballot or legal protection
of dissent shields the enactment act itself.

In the fully coupled model, sufficiently concentrated location--scale formation costs whose zero-pressure
return clears the median cost at some prevalence deliver a pressure-threshold pair even with an active
enactment channel (Theorem~\ref{thm:coupled-coexistence}). The benchmark below instead makes the threshold
geometry globally explicit.

\subsection{A sharp characterization in the logit--proportional benchmark}\label{subsec:canonical}

With logit formation costs, a proportional uncapped formation sanction, and a prevalence-independent
zero-pressure return, the threshold geometry admits a closed-form characterization. A single primitive
inequality determines whether the formation-pressure interval supporting three cultures is nonempty.

\begin{theorem}[Global threshold characterization in the logit--proportional benchmark]
\label{thm:analytic-threshold}
Fix $p\in(0,1)$ and $\phi>0$. Suppose the formation return at zero formation pressure is constant in
prevalence, $B(J;p)\equiv B_0(p)$---as with a fixed enactment sanction $s_a^0$, for which
$B_0(p)=p\Omega(b_N-\psi s_a^0)-(1-p)\gamma$. Suppose formation costs are logit,
$H^{-1}(J)=\mu_k+\varsigma\log\bigl(J/(1-J)\bigr)$ with $\varsigma>0$, and the formation audience uses the
globally interior, uncapped proportional rule
\[
 s_f(J;\kappa_f)=\kappa_fb_f(J),
 \qquad
 b_f(J)=q_f(J)-\bar q_f,
 \qquad
 q_f(J)=\frac{\nu_f}{\nu_f+\alpha_fJ},
\]
with $0\leq\bar q_f<q_f(1)$. Let $r:=\nu_f/\alpha_f$ and $c_0(p):=(B_0(p)-\mu_k)/\varsigma$, and define
$m(r,\bar q_f):=\log\bigl(r/(1+r)\bigr)+2\bigl[1+2r-2\bar q_f(1+r)\bigr]$.

The primitive inequality
\begin{equation}\label{eq:global-logit-inequality}
 c_0(p)>m(r,\bar q_f)
\end{equation}
is necessary and sufficient for $\mathcal K(\cdot;p)$ to have exactly two stationary points. Under this
inequality, there are unique $0<J_-<J^\circ<J_+<1$, with $J^\circ:=r/(1+2r)$, such that
$\mathcal K'<0$ on $(0,J_-)$, $\mathcal K'>0$ on $(J_-,J_+)$, and $\mathcal K'<0$ on $(J_+,1)$.
The point $J_-$ is a strict local minimum, $J_+$ is a strict local maximum, and
$0<\mathcal K(J_-;p)<\mathcal K(J_+;p)$.

Consequently, for every $\kappa_f>0$, the reduced equilibrium equation has exactly three interior fixed
points if and only if
\begin{equation}\label{eq:canonical-pressure-window}
 \mathcal K(J_-;p)
 <
 \kappa_f
 <
 \mathcal K(J_+;p).
\end{equation}
The outer two are scalar-stable and the middle one scalar-unstable. Outside the closed
interval in \eqref{eq:canonical-pressure-window}, there is exactly one interior fixed point. At either endpoint
there are two distinct fixed points, one of which is a nondegenerate saddle-node. Hence
\eqref{eq:global-logit-inequality} is necessary and sufficient, within this benchmark, for a nonempty
multiplicity interval in formation pressure.

If $c_0(p)<m(r,\bar q_f)$,
$\mathcal K$ is strictly decreasing. If equality holds, it remains strictly decreasing but has an isolated
horizontal tangent at $J^\circ$. In either case every positive formation pressure supports a unique interior
fixed point.
\end{theorem}

The proof is in Appendix~\ref{app:logit-geometry}. The constant-return restriction delivers the global
equivalence. The two-turning-point geometry, the positive pressure interval, the scalar-stability signs, and
the two nondegenerate saddle-nodes persist under sufficiently small $C^2$ enactment feedback
(Corollary~\ref{cor:logit-feedback}). The fully coupled result
Theorem~\ref{thm:coupled-coexistence} instead provides a primitive sufficient condition without requiring the
return $B(J;p)$ to be constant.

The parameters have distinct roles: $\kappa_f$ is the bifurcation parameter, novelty affects the
zero-pressure return and the interval's endpoints, and audience tolerance reshapes the pressure window.
Multiplicity is generated by the dependence of Bayesian sanctions on authentic prevalence.

\subsection{Comparative statics and scope}\label{subsec:scope}

\begin{proposition}[Pressure at formation and enactment]\label{prop:capacity}
Under the enactment regularity condition \eqref{eq:enactment-regularity}, $a_J>0$.
Away from kinks, on any smooth branch,
$a_{\theta_a}\leq0$, $\widehat R_J\geq0$, $\widehat R_{\theta_f}\leq0$, $\widehat R_{\theta_a}\leq0$,
and
\[
 \widehat R_p
 =
 \Omega_P(a)+\gamma+\phi s_f(J)
 >0.
\]
Novelty raises the frequency of the option value of judgment and simultaneously reduces the frequency of
familiar-state diagnostic and reputational costs.

At every regular scalar-stable equilibrium and for
$\theta\in\{\theta_f,\theta_a,p\}$,
\[
 \frac{dJ_j}{d\theta}
 =
 \frac{h(\widehat R_j)\widehat R_\theta}
 {1-\Psi'(J_j)}.
\]
Thus, locally along any regular scalar-stable equilibrium branch, greater pressure at either margin weakly
lowers formed and enacted judgment, strictly when
the relevant effective sanction responds, and greater novelty strictly raises both. Holding $J$ fixed,
enactment pressure weakly raises $S^R$, strictly on a strictly sanctioned smooth enactment branch.
\end{proposition}

On strictly sanctioned smooth branches the weak inequalities above are strict ($\widehat R_{\theta_f}$
additionally requires $\phi>0$), and $\widehat R_J>0$ whenever either margin responds strictly; a capped or
tolerated branch may carry a positive but locally unresponsive sanction.

\paragraph{Audience memory.}
Personal records delimit the anonymous channel. As partners observe longer histories, authentic and
opportunistic conduct become easier to distinguish and expected formation stigma declines. The mechanism is
therefore strongest when counterparties turn over faster than informative personal records accumulate.
Online Appendix OA.8 provides the corresponding bound.

\paragraph{Private maintenance.}
Because the formation sanction prices the visible probe, the natural deviation is to keep the practice and
hide it \citep{Prat2005}. The baseline treats diagnosis as interactive by technology: a probe requests a
justification held by others, and the knowledge needed to test a premise is dispersed \citep{Hayek1945}, so
a private substitute exists only at a cost premium. Online Appendix OA.8 makes the menu explicit. Whenever
the premium exceeds the largest exposure the formation sanction can generate, $\phi\kappa_f(1-\bar q_f)$,
the visible mode is strictly dominant at every prevalence and every equilibrium survives the menu
unchanged. At the benchmark this bound is $\invisUnifxg$ times the diagnostic cost $\gamma$; the certified
configuration survives numerically down to premiums of $\invisCapxg\gamma$, and free invisibility---the
limit $\phi=0$ of Corollary~\ref{cor:anonymity}---leaves a unique questioning culture at $J=\phiZeroJ$. The
trap thus requires that hiding the practice forfeit the information others hold.

\paragraph{Responsive background activity.}
The background source may itself shrink when stigma rises. Let opportunists differ in expressive benefit
$v\sim F_v$ on $[0,\infty)$ and act when that benefit covers the expected sanction, so the active formation
pool is $\nu_f(s)=\bar\nu_f[1-F_v(\phi s)]$, with an analogous construction at enactment. Writing
$\mathcal K_\nu$ for the pressure-threshold function at formation pool $\nu$ (the enactment pool held
fixed), the induced equilibrium map lies between the fixed-pool maps generated by the largest pool
$\bar\nu_f$ and the smallest undeterred mass
$\underline\nu_f:=\bar\nu_f[1-F_v(\phi s_f^{\max})]$, $s_f^{\max}:=\kappa_f(1-\bar q_f)$, so at least
three fixed points survive whenever
$\mathcal K_{\bar\nu_f}(J_H;p)>\kappa_f>\mathcal K_{\underline\nu_f}(J_L;p)$: a positive undeterred mass
and overlapping threshold bounds suffice for the prevalence feedback to survive responsive opposition.
Online Appendix OA.8 gives the sandwich argument, and Proposition~\ref{prop:responsive} establishes
persistence under sufficiently small smooth responsiveness.

Together, the results isolate the ingredients of coexistence through the formation-stigma route:
residual ambiguity, a sanction
that falls with authentic prevalence, positive effective formation exposure, a single-valued enactment
response, and an ordered pair of formation-pressure thresholds; at $\phi=0$, coexistence can instead
survive through the enactment channel alone. The three-piece threshold shape pins down
exactly three cultures. The next section makes formed judgment an inherited stock and studies forward-looking
paths from each initial condition.
\section{Forward-looking cultures: history, expectations, and policy duration}\label{sec:dynamics}

From one inherited stock, more than one rational-expectations future can be consistent with equilibrium.
This section constructs the extremal perfect-foresight paths, classifies inherited stocks into up to three
convergence regions separated by two cutoffs, and converts those cutoffs into duration bounds for temporary
policy. Throughout, the enactment subgame has a unique continuous nondecreasing solution $a(J)$ (condition
\eqref{eq:enactment-regularity} suffices on smooth branches), and $\widehat R(J):=R(J,a(J);p)$ and
$\Psi(J):=H(\widehat R(J))$ are continuous and nondecreasing.

\subsection{The dynamic stock and intervention timing}\label{subsec:dynstock}

The aggregate maintenance status is predetermined within a date, but the agents who revise it need not be
myopic. Current flow at date $t$ is realized against inherited $J_t$; at the end of the date, a fraction
$\delta\in(0,1)$ receives an opportunity to
choose its formation status for date $t+1$, independently of current status and cost. A revision remains in force
until the next opportunity. Agents discount at $\beta\in(0,1)$ and are atomistic. Interpret $k$ as the
per-period-equivalent net cost of maintaining formed judgment during the ensuing revision
spell.\footnote{Because its costs and returns begin when the new status becomes effective, their common
one-period discount cancels from the revision cutoff; this convention preserves the static $H$ rather than
mechanically rescaling it.} Let
$d:=1-\delta$ and $\lambda:=\beta(1-\delta)=\beta d\in(0,d)$, so $d$ is the per-period survival of an
existing maintenance status.\footnote{The audience primitives $\pi_m$, $\Lambda_m$, $c_m$ of
Section~\ref{subsec:audience} are separate objects and always carry a margin subscript.}
Given an anticipated aggregate path $\mathbf J=(J_t)_{t\ge0}$, the normalized value of formed judgment at date
$t$ is
\begin{equation}\label{eq:forward-value}
 V_t(\mathbf J):=(1-\lambda)\sum_{n=0}^{\infty}\lambda^n\widehat R(J_{t+n}),
 \qquad V_t=(1-\lambda)\widehat R(J_t)+\lambda V_{t+1},
\end{equation}
where the unique contemporaneous enactment solution is $A_t=a(J_t)$. The cohort revising at $t$ chooses a status
effective at $t+1$ and therefore forms iff $k\le V_{t+1}$,\footnote{Explicitly: a future revision resets
the status independently of the current one, so continuation payoffs beyond the next opportunity are
common to both choices, and forming changes the discounted stream by
$\sum_{n\ge0}\beta^{\,n+1}d^{\,n}\bigl[\widehat R(J_{t+1+n})-k\bigr]
=\tfrac{\beta}{1-\lambda}\bigl[V_{t+1}-k\bigr]$: the new status takes effect at $t+1$ and survives $n$
further periods with probability $d^{\,n}$. The revising cohort's cost distribution is $H$ because
opportunities arrive independently of $k$ and of current status, so each cohort is a representative
cross-section of the population---under the exact law-of-large-numbers convention for a continuum
\citep{Sun2006}---and this holds whether $k$ is permanent or redrawn at revision; both readings deliver
\eqref{eq:forward-stock}.} so
\begin{equation}\label{eq:forward-stock}
 J_{t+1}=dJ_t+\delta H(V_{t+1}).
\end{equation}

Two microfoundations induce the same equations: long-lived members with revision opportunities at rate
$\delta$, and literal turnover---exit at hazard $\delta$, entrants drawing $k\sim H$. Both deliver
$\lambda=\beta(1-\delta)$, because what matters is the survival of a maintenance status; the reading of
$\delta$ as generational replacement and migration, or seat rotation on a committee
(Section~\ref{subsec:outscope}), is exact rather than an analogy. A period is the interval over which the
fraction $\delta$ of the group turns over, and the duration statements below inherit this
unit.\footnote{At the worked benchmark $\delta=\deltaval$, a period replaces one member in twenty, so the
``three intervention dates'' of Section~\ref{sec:numerical} amount to roughly a $15\%$ replacement of the
body rather than three calendar periods.}

A \emph{perfect-foresight cultural path} from $J_0\in[0,1]$ is any bounded sequence
$(J_t,V_t)_{t\ge0}$ satisfying \eqref{eq:forward-value}--\eqref{eq:forward-stock}. An \emph{admissible
continuation} from date $T$ is the tail $(J_t,V_t)_{t\ge T}$ of any such path---under a temporary policy, of
any perfect-foresight path of the policy-augmented system; universal statements about continuations below
quantify over exactly this set.

\begin{corollary}[Immediate and delayed effects of enactment protection and formation support]\label{cor:stock-flow}
Fix a date $t$ and an inherited stock $J_t$.
\begin{enumerate}
\item[(i)] Removing enactment stigma immediately raises current enactment from
$A_t=J_tw(A_t)$ to $A_t^0=J_tw^0$---weakly, and strictly whenever $S_t^R>0$---without changing $J_t$, and
releases exactly
\[
 A_t^0-A_t=J_t[w^0-w(A_t)]=S_t^R.
\]
If departures are measured directly as acts---or as registered events with known $\psi$---and the background
flow is regime-invariant, the impact change in the measured departure mass identifies reputationally silenced judgment.
\item[(ii)] A direct cutoff shift $\tau_t$ offered to the cohort revising at $t$ leaves $(J_t,A_t)$ unchanged on
impact and changes the next stock according to
\[
 J_{t+1}(\tau_t)=dJ_t+\delta H(V_{t+1}+\tau_t),\qquad
 \left.\frac{\partial J_{t+1}}{\partial\tau_t}\right|_{V_{t+1}}
 =\delta h(V_{t+1}+\tau_t)>0.
\]
Conditional on the continuation value, its enactment effect first appears through
$A_{t+1}=a(J_{t+1})$.
\end{enumerate}
Thus enactment support has a front-loaded flow response at the inherited maintenance state, whereas formation support has a
delayed stock response followed by enactment. This timing restriction is the operational two-margin prediction.
\end{corollary}
\begin{proof}
Part~(i) is the decomposition \eqref{eq:decomposition} at fixed $J_t$; part~(ii) follows from
\eqref{eq:forward-stock}, the cutoff shift, and the predetermined nature of $J_t$.
\end{proof}

In the committee, part~(i) is the protective chair, whose effect appears in this quarter's votes; part~(ii)
is training or selection, whose effect appears only as seats rotate. In the civic reading, part~(i) is
legal protection of dissent, visible in current conduct, and part~(ii) is civic formation, whose effect
arrives only as cohorts replace.

\subsection{Local and global perfect-foresight dynamics}\label{subsec:pfdyn}

The system \eqref{eq:forward-value}--\eqref{eq:forward-stock} has one predetermined variable ($J_t$) and one
that jumps ($V_t$), so an outer culture is locally a saddle for the right reason: the saddle property is
determinacy, not fragility---near it, exactly one value of $V_t$ places the system on the stable arm. Rational
expectations pin down that value. Under the slope condition of Online Appendix OA.9 the middle culture is a
source; outside that region it can be a saddle with a negative stable root, so convergence along its stable
arm is oscillatory. The benchmark lies in the source
region, so the scalar and perfect-foresight classifications agree there.

\begin{proposition}[Stationary cultures, saddle paths, and convergence rates near a fold]\label{thm:forward-local}
Suppose the enactment subgame has a unique $C^1$ solution $a(J)$, $\widehat R$ is $C^1$ and nondecreasing,
and $H$ is $C^1$ with continuous density $h$ near each stationary value $\widehat R(J_j)$.
A stationary perfect-foresight culture is exactly a static
cultural equilibrium. At a transverse fixed point $J_j$, let $m_j:=\Psi'(J_j)$ and define
$\vartheta_j:=d+\lambda^{-1}-\delta(1-\lambda)m_j/\lambda$.
The characteristic roots solve
\begin{equation}\label{eq:forward-characteristic}
 P_j(r)=r^2-\vartheta_jr+\frac d\lambda=0,
 \qquad
 P_j(0)=\frac1\beta>1,
 \qquad
 P_j(1)=\frac{\delta(1-\lambda)}{\lambda}(m_j-1).
\end{equation}
Hence every outer culture, where $m_j<1$, is a saddle with one root in $(0,1)$ and one above one. Its
convergence factor is
\begin{equation}\label{eq:forward-stable-root}
 r_j^s=\frac12\left[\vartheta_j-
 \sqrt{\vartheta_j^2-\frac{4d}{\lambda}}\right].
\end{equation}
At a nondegenerate fold along an outer branch, $m_j\uparrow1$ and $r_j^s\uparrow1$, so the convergence half-life
$\log(1/2)/\log r_j^s$ diverges. Online Appendix OA.9 gives the Jacobian and the complete classification of
the middle culture.
\end{proposition}

Those local labels do not determine global history. The next theorem instead constructs the extremal
rational-expectations paths, adapting monotone lattice methods for dynamic economies with complementarities
\citep{BalbusReffettWozny2014} to the perfect-foresight selection skeleton of
\citet{MatsuiMatsuyama1995,HofbauerSorger1999,Oyama2002,BurdzyFrankelPauzner2001,OyamaTakahashiHofbauer2008}---with
no selection imposed, and with the complementarity itself an equilibrium object.\footnote{The cutoffs $J_-^{RE}$ and $J_+^{RE}$ are state-dependent
cousins of the absorption and accessibility notions that \citet{OyamaTakahashiHofbauer2008} also compute
from extremal paths, obtained here as certified enclosures from finitely checkable conditions rather than
as limit criteria. \citet{PierriReffett2026} construct recursive competitive equilibria on minimal state
spaces under price-dependent collateral constraints, with computable existence theory and constructive
comparative statics; the operator here acts instead on anticipated aggregate paths from a predetermined
social stock.} The construction uses an operator on conjectured
futures: $\mathcal T_j$ takes an anticipated aggregate path, computes the value it implies, and returns the
path that the induced formation decisions would actually generate; its fixed points are exactly the
perfect-foresight paths.

For fixed $j$, let $\mathbb X(j)$ be the set of sequences $\mathbf x\in[0,1]^{\mathbb N_0}$ with $x_0=j$,
ordered coordinatewise, and define
\begin{align}
 (\mathcal T_j\mathbf x)_0&=j,\nonumber\\
 (\mathcal T_j\mathbf x)_{t+1}
 &=dx_t+\delta H\!\left((1-\lambda)\sum_{n=0}^{\infty}
 \lambda^n\widehat R(x_{t+1+n})\right).\label{eq:path-operator}
\end{align}

\begin{theorem}[Extremal paths and expectation-dependent convergence]\label{thm:history-expectations}
Suppose the enactment subgame has a unique continuous nondecreasing solution $a(J)$ and the induced
$\widehat R$ and $H$ are continuous and nondecreasing.
\begin{enumerate}
\item[(i)] For every inherited stock $j$, perfect-foresight paths form a nonempty complete lattice. Its least
and greatest elements $\underline{\mathbf J}(j)$ and $\overline{\mathbf J}(j)$ are obtained by iterating
$\mathcal T_j$ from $(j,0,0,\ldots)$ and $(j,1,1,\ldots)$, respectively.
\item[(ii)] Both extremal paths are nondecreasing in $j$. If they converge to the same stationary culture,
every perfect-foresight path from $j$ converges there; if the least converges to $J_\ell$ and the greatest to
$J_h$, different continuations lead to different cultures from the same inherited stock.
\item[(iii)] Suppose $\Psi$ has exactly three stationary cultures and the sign of $\Psi(J)-J$ alternates
across them, beginning positive near zero. Let
\[
 \mathcal E_h^{\max}:=\{j:\overline J_t(j)\to J_h\},\qquad
 \mathcal E_h^{\min}:=\{j:\underline J_t(j)\to J_h\}.
\]
Both are upper sets and $\mathcal E_h^{\min}\subseteq\mathcal E_h^{\max}$. When
$\mathcal E_h^{\min}$ is nonempty, define
\begin{equation}\label{eq:expectation-band}
 J_-^{RE}:=\inf\mathcal E_h^{\max},\qquad
 J_+^{RE}:=\inf\mathcal E_h^{\min}.
\end{equation}
Then $J_-^{RE}\le J_+^{RE}$. No path from $j<J_-^{RE}$ converges to $J_h$, whereas every path from
$j>J_+^{RE}$ does. Thus the cutoffs separate up to three open regions: below $J_-^{RE}$ high convergence is
impossible; strictly between the cutoffs it depends on continuation expectations; and above $J_+^{RE}$ it
occurs under every admissible continuation. Cutoff stocks and non-convergent paths are left unclassified:
the sets classify convergence to $J_h$ only.
\end{enumerate}
\end{theorem}

The proof, in Appendix~\ref{app:dynproofs}, is a Tarski argument on the coordinatewise lattice of paths,
plus a squeeze on limit points. The theorem does not by itself make the band strict: the cutoffs may
coincide and any region may be empty; Theorem~\ref{thm:band} below gives a checkable
sufficient spectral condition for strictness. At the certified benchmark the enclosures
$[\REminlo,\REminhi]$ and $[\REpluslo,\REplushi]$ are disjoint, so $J_-^{RE}<J_+^{RE}$ strictly and all
three regions are nonempty. Between the cutoffs, failure to converge high does not imply convergence to
$J_\ell$ in general---Online Appendix OA.11 verifies the stronger low--high selection on the benchmark's
certified interior band---and the boundary stocks need not equal the static culture $J_i$.

The cutoffs admit finite certificates based on monotone iterates and finite-tail sub- and supersolutions.
Online Appendix OA.11 states the sufficient inequalities, controls the discounted tail, and verifies the
benchmark brackets $J_-^{RE}\in[\REminlo,\REminhi]$ and $J_+^{RE}\in[\REpluslo,\REplushi]$.

The certificates quantify the band at a parameter point; the next result makes it strictly open from a
condition on the slope of the static reduced map and the revision primitives $(\beta,\delta)$, checkable
without solving any dynamic path. The condition is local at
the unstable culture and is exactly the \emph{focus} case of the
local classification (Proposition~\ref{thm:forward-local}; Online Appendix OA.9): at a transverse middle
culture with $m_i:=\Psi'(J_i)>1$, the characteristic roots of \eqref{eq:forward-characteristic} are
nonreal precisely when
\begin{equation}\label{eq:focus-condition}
 \left(d+\lambda^{-1}-\frac{\delta(1-\lambda)m_i}{\lambda}\right)^{2}<\frac{4d}{\lambda}.
\end{equation}
The mechanism is simple. The forward system \eqref{eq:forward-value}--\eqref{eq:forward-stock} has constant
Jacobian determinant $d/\lambda=1/\beta$, so it is invertible wherever it is $C^1$. Under
\eqref{eq:focus-condition}, its \emph{backward} dynamics at $(J_i,\widehat R(J_i))$ is a stable spiral:
exact equilibrium prefixes recovered backward from a nearby target circle the unstable culture and cross
$J=J_i$ from both sides. Pasting such a prefix to a constant continuation above (below) $J_i$ yields a
subsolution (supersolution) of the path operator, and with it an equilibrium path that escapes from an
inherited stock strictly below $J_i$---or falls to the trap from one strictly
above. That complex roots at the unstable steady state open a history--expectations overlap---a regime in
which expectations can overturn history in both directions---is the classical criterion of
\citet{Krugman1991}, characterized exactly by \citet{FukaoBenabou1993} for smooth planar
perfect-foresight systems; Theorem~\ref{thm:band} transplants it to the present setting, where the
dynamics is a nonsmooth lattice system on paths, the complementarity runs through an equilibrium
posterior, and the classical phase-plane arguments do not apply.

\begin{theorem}[The focus condition opens the expectation band]\label{thm:band}
Maintain the hypotheses of Theorem~\ref{thm:history-expectations}.
Let $\Psi$ have exactly three stationary cultures $J_\ell<J_i<J_h$ with the alternating sign chart of
Theorem~\ref{thm:history-expectations}(iii); suppose the enactment subgame has a unique $C^1$ solution
and $\widehat R$ is $C^1$ and nondecreasing, each near $J_i$, and $H$ is $C^1$ with continuous density
$h$ near the stationary value $\widehat R(J_i)$;
and let $m_i>1$ satisfy the focus condition \eqref{eq:focus-condition}. Then there exist inherited stocks
$j^-<J_i<j^+$ with
\[
 [\,j^-,j^+\,]\subseteq\mathcal E_h^{\max}\setminus\mathcal E_h^{\min}:
\]
from every $j\in[j^-,j^+]$, some perfect-foresight path converges to $J_h$ and some converges to
$J_\ell$. Consequently, when $\mathcal E_h^{\min}$ is nonempty,
\[
 J_-^{RE}\;\le\;j^-\;<\;J_i\;<\;j^+\;\le\;J_+^{RE},
\]
so the cutoffs are strictly separated and the expectation-dependent region contains an open interval
around the unstable culture.
\end{theorem}

The proof, in Appendix~\ref{app:dynproofs}, combines an elementary planar lemma---backward orbits of an
unstable focus spiral into it, crossing any line through it infinitely often---with the sub- and
supersolution pasting already used for the extremal paths. The finite backward witnesses of the OA.11
certificate are nonlinear instances of the same construction. The condition fails as $\lambda\downarrow0$
whenever $\delta m_i\neq1$, myopic revision making the roots
real---consistent with the band being a genuinely forward-looking phenomenon.

\subsection{Duration bounds and direct formation support}\label{subsec:duration}

The two boundaries now become policy targets: how many cohorts must have the opportunity to revise before
the stock can physically cross each one?

\begin{theorem}[Sharp duration bounds and sufficient direct formation shifts]
\label{thm:frontier}
Maintain the hypotheses of Theorem~\ref{thm:history-expectations}, including the three-culture sign-chart
supposition of its part~(iii), and suppose $\mathcal E_h^{\min}$ is nonempty, so that the cutoffs
\eqref{eq:expectation-band} are defined. Let $B_t\in[0,1]$ be the share of the
revising cohort that forms under a temporary policy withdrawn after
date $T-1$, and suppose $J_0<J_-^{RE}\le J_+^{RE}<1$, where $J_-^{RE}$ and $J_+^{RE}$ are the
post-withdrawal convergence cutoffs.
\begin{enumerate}
\item[(i)] \emph{(Universal bound.)} Every policy and expectation path, for any instrument at fixed
$\delta$, satisfies
\begin{align}
 J_T&=d^TJ_0+\delta\sum_{t=0}^{T-1}d^{T-1-t}B_t,
 \label{eq:general-stock-bounds}\\
 J_T&\le1-(1-J_0)d^T.
 \label{eq:physical-stock-ceiling}
\end{align}
\item[(ii)] \emph{(Duration necessity.)} The first durations at which the physical ceiling lies strictly
above each threshold are
\begin{align}
 T_{\mathrm{possible}}
 &=\left\lfloor\frac{\log[(1-J_-^{RE})/(1-J_0)]}{\log d}\right\rfloor+1,
 \label{eq:Tpossible}\\
 T_{\mathrm{robust}}
 &=\left\lfloor\frac{\log[(1-J_+^{RE})/(1-J_0)]}{\log d}\right\rfloor+1;
 \label{eq:Trobust}
\end{align}
no shorter policy can make $J_T$ strictly exceed the corresponding threshold.
\item[(iii)] \emph{(Constructive sufficiency.)} Suppose $H$ is strictly increasing with full support and a
direct formation instrument shifts the revising cohort's cutoff to $k\le V_{t+1}+\tau$. Let
$\underline R:=\min_{J\in[0,1]}\widehat R(J)$. For either target $c\in\{J_-^{RE},J_+^{RE}\}$, whenever
$1-(1-J_0)d^T>c$, define the uniform worst-case support bound
$\tau_T^{\mathrm{env}}(c):=\bigl[H^{-1}\!\bigl((c-d^TJ_0)/(1-d^T)\bigr)-\underline R\bigr]_+$.
Every finite uniform shift $\tau>\tau_T^{\mathrm{env}}(c)$ guarantees $J_T>c$ for every admissible
continuation: direct formation support can make high convergence possible at $T_{\mathrm{possible}}$ and
guarantee it at $T_{\mathrm{robust}}$.
\item[(iv)] \emph{(Sharpness within the class.)} If every finite direct cutoff shift has $B_t<1$ while
arbitrarily large finite shifts make $B_t$ arbitrarily close to one, both duration bounds are sharp within
this class.
\end{enumerate}
\end{theorem}

The proof is in Appendix~\ref{app:dynproofs}.

\paragraph{Beyond the bounds.}
The theorem gives duration bounds and conservative sufficient shift sizes within temporary direct cutoff
shifts, not a least-cost policy; Online Appendix OA.10 gives the capped-instrument duration, the full
envelope, nonuniform schedules, and the extremal-path construction. Enactment protection changes current
flow, obeys the same physical ceiling, and inherits constructive sufficiency only if promised future
protection supplies a uniform continuation-value floor. Reforms that raise $\delta$ itself---rotation,
hiring, selective replacement---or permanently change primitives lie outside the bounds: no temporary
intervention at fixed turnover can rebuild the stock faster than turnover permits, while forward-looking
expectations determine whether a given terminal stock makes high convergence merely possible or guarantees
it under every admissible continuation.

Figure~\ref{fig:dynamic-regions} summarizes the benchmark classification. The left panel shows the certified
convergence regions; the right panel compares their boundary boxes with the fastest stock rebuilding physically
possible after each intervention date.
\begin{figure}[t]\centering
\caption{History, expectations, and policy duration in the worked benchmark.}
\label{fig:dynamic-regions}
\includegraphics[width=0.92\textwidth]{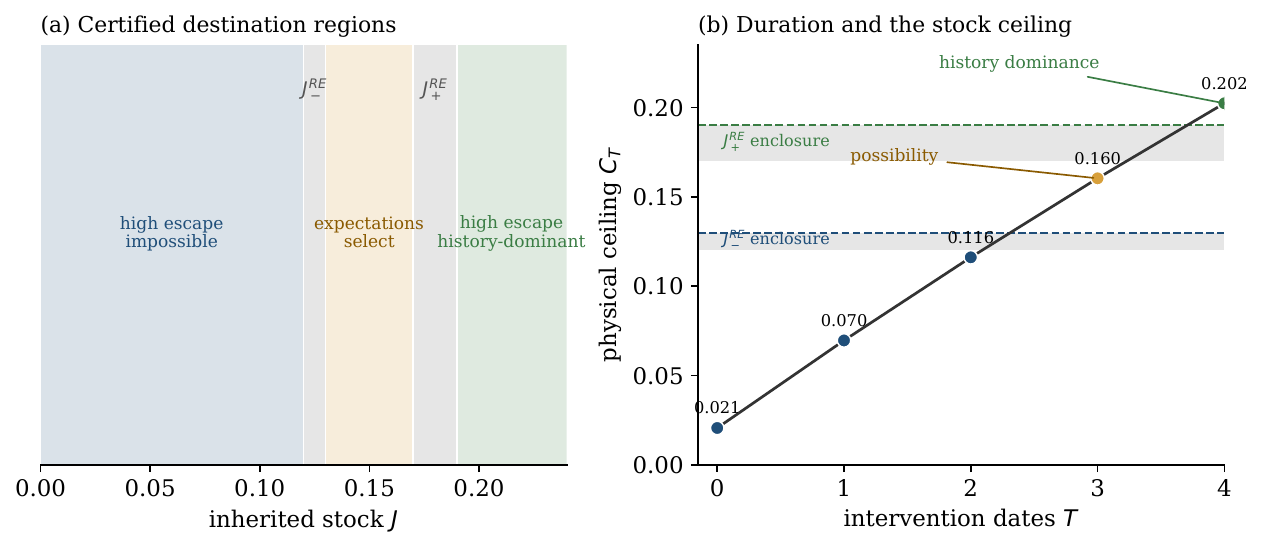}
\begin{minipage}{0.90\textwidth}
\footnotesize\emph{Notes:} The shaded intervals in the left panel enclose $J_-^{RE}$, below which high
convergence is impossible, and $J_+^{RE}$, above which it occurs under every admissible continuation. The right panel plots the physical stock ceilings from the trap
for zero through four intervention dates. Three dates cross the first boundary box; four dates cross the
second.
\end{minipage}
\end{figure}

\subsection{The inherited impact value of enactment protection}\label{subsec:sequencing}

The timing asymmetry of Corollary~\ref{cor:stock-flow} says the two instruments act at different speeds. It
does not say which to use first, and neither does what follows. What the model does deliver is a prerequisite
for any such comparison: the sign of the impact effect of protection, which is not a constant of the
environment but a function of the inherited state. Releasing judgment is worth something only when the
judgment released is likely to be right, and how likely it is to be right depends on how much judgment the
group has been enacting.

Let shared competence $\sigma(A)$ be strictly increasing, concave, and $C^1$: the accuracy at diagnosing novelty that
prevalent enacted practice sustains. It is a steady-state object, the accuracy supported by a culture in which
a mass $A$ routinely departs from failing routines, not a within-period return to a single departure; a
one-date intervention therefore does not move it. A formed agent identifies the appropriate novel action with
probability $\sigma(A)$; a correct departure produces $V_N$, an incorrect one destroys $L>0$, and either
avoids the routine's novel-state damage $D\ge0$. The expected output value of a departure relative to the
routine is
\begin{equation}\label{eq:delta-sigma}
 \Delta(\sigma)=\sigma V_N-(1-\sigma)L+D,
\end{equation}
which is increasing in $\sigma$ and vanishes at the \emph{break-even competence threshold}
\begin{equation}\label{eq:sigma-bar}
 \bar\sigma:=\frac{L-D}{V_N+L}.
\end{equation}
To make the crossing well defined, assume $\bar\sigma\in\sigma\bigl(a((0,1))\bigr)$: the break-even
competence threshold is attained at an enacted mass that some interior inherited stock actually generates.
Then $A^\dagger:=\sigma^{-1}(\bar\sigma)$ is well defined and lies in $a((0,1))$, and it is the enacted
mass at which a marginal departure has zero output value.\footnote{The condition rules out two corners: if
$\bar\sigma>\sigma(a(1))$, competence never reaches the threshold and protection lowers impact output at
every inherited stock, while if $\bar\sigma\le\sigma(a(0^+))$ it raises it at every stock. Both are
degenerate rather than impossible; the benchmark of Section~\ref{sec:numerical} satisfies the condition
with room to spare.}

\begin{proposition}[The impact sign of protecting dissent in a stationary inherited culture]\label{thm:sequencing}
Fix an inherited culture that is stationary at stock $J_t$, so that its enacted mass is $A_t=a(J_t)$ and its
inherited competence is the steady-state accuracy $\sigma(A_t)$ that the mean-count learning technology of
Online Appendix OA.1 attaches to that culture. Let a one-date enactment-protection
intervention remove enactment stigma at date $t$ only. The intervention
releases $S_t^R=J_t[w^0-w(A_t)]$ departures, and the impact change in adaptive output is
\begin{equation}\label{eq:impact-sign}
 S_t^R\,\Delta\bigl(\sigma(A_t)\bigr).
\end{equation}
At every positive inherited stock on a strictly sanctioned enactment branch $S_t^R>0$, so the sign is that
of $\Delta(\sigma(A_t))$: strictly negative when $A_t<A^\dagger$, strictly positive when $A_t>A^\dagger$. (At
$J_t=0$ nothing is released and the impact is zero.) Because $a(\cdot)$ is strictly
increasing---under \eqref{eq:enactment-regularity} and full support of $G$, \eqref{eq:aprime} gives
$a'>0$---the two signs are separated in the stock by the single threshold
\begin{equation}\label{eq:seq-stock}
 J^\dagger:=a^{-1}(A^\dagger):
\end{equation}
on a strictly sanctioned enactment branch, every stationary culture with stock below $J^\dagger$ has
strictly negative impact and every one above it strictly positive; on a tolerated branch nothing is
silenced and the impact is zero. The threshold is a point of the state space, not itself generically a stationary stock;
both signs occur across stationary cultures exactly when stationary stocks lie on both sides of
$J^\dagger$.
\end{proposition}

\begin{proof}
At fixed inherited competence, each released departure contributes $\Delta(\sigma(A_t))$ relative to the
routine, and $S_t^R$ of them are released; \eqref{eq:impact-sign} follows. The sign statement is
$\Delta(\sigma)\gtrless0$ as $\sigma\gtrless\bar\sigma$ together with monotonicity of $\sigma$ and $a$.
\end{proof}

The economic content is that the sign of protecting dissent is inherited rather than intrinsic. The same
chair who shields a dissenting vote adds value in a group that has been exercising judgment and destroys it in
one that has not, and \eqref{eq:seq-stock} locates the crossing. In the trap the released judgment carries
negative expected adaptive value. The stationarity hypothesis is what licenses reading competence off the
current stock: along a transition the robust object is $S_t^R\Delta(\sigma_t)$ with $\sigma_t$ the accuracy
actually inherited, and two groups at the same transitional stock can carry different accumulated accuracy
from different histories, so the identification $\sigma_t=\sigma(a(J_t))$---and with it the threshold
$J^\dagger$---applies to stationary inherited cultures, which is how the benchmark of
Section~\ref{sec:numerical} uses it.

\paragraph{Toward sequencing.} The impact sign is not a sequencing rule (rebuild first, protect
afterwards), and I do not claim one. Two forces push the other way. Releases are public experiments: across
stationary cultures, higher enacted judgment sustains higher competence, an externality that no private
enactor counts (Theorem~\ref{thm:central-welfare}(i)). And a \emph{promised continuation} of protection,
unlike the one-date surprise of Corollary~\ref{cor:stock-flow}, raises the returns entering the revising
cohort's value $V_{t+1}$ and so accelerates rebuilding through the same channel formation support targets.
Settling the order would require a dynamic competence state (for instance
$C_{t+1}=(1-\eta)C_t+\eta f(A_t)$), an intertemporal objective, and instrument costs---an extension left
open; what the model delivers is the input any such exercise needs: the impact sign, the threshold
\eqref{eq:seq-stock}, and the timing contrast of Corollary~\ref{cor:stock-flow}.

\subsection{What stationary behavior conceals}\label{subsec:conceal}

The thresholds $J_-^{RE}$ and $J_+^{RE}$ depend on continuation values off the stationary path, so a
stationary action level does not identify either boundary without the transition primitives. In particular,
$\beta$ enters neither stationary condition yet moves both boundaries: at the worked primitives, greater
patience widens the expectation-dependent region---more patient members make a group both easier to rescue
and easier to lose (Online Appendix OA.10, a local computation rather than global
monotonicity).\footnote{The duration pair $(\Tpossible,\Trobust)$ is unchanged across that comparison
because both boundaries move within the same steps of the physical ceiling.}

\paragraph{Slowing convergence near a fold.}
At a nondegenerate fold where an outer culture meets $J_i$, $m_j\to1$ and the stable root in
\eqref{eq:forward-stable-root} tends to one, so the convergence half-life $\log(1/2)/\log r_j^s$ diverges. The
stationary level remains continuous up to the fold, whereas convergence becomes arbitrarily slow; this is the
phenomenon known as critical slowing down \citep{SchefferEtAl2009}. Tolerance or cap kinks can generate nonsmooth
transitions rather than nondegenerate folds; the worked baseline stays on a globally interior branch, so its
folds are nondegenerate.

Even sanctions that are ex-post rational filters are chosen after the stock is already low, so a
commitment to protect dissent or a norm of interpretive charity is valuable precisely as a commitment not
to apply the suspicious interpretation the trap itself creates---not a call for unconditional
permissiveness, since released fallible judgment can spread confident error (Online Appendix OA.1).

\section{Diagnosing absent and suppressed judgment}\label{sec:hidden}

What can an observer learn about the maintenance stock from conduct? Less than the stock's importance would
suggest, and in one respect the ranking is actively misleading. Conduct does not reveal the composition of
silence: the same visible action may be independent for one member and conformist for another. For a group
emerging from enforced conformity this is the transition question: how much did pressure deplete
maintained judgment, and how much of the surviving stock is socially suppressed rather than privately
withheld? Throughout, the analyst does not know the recording and background rates; were those primitives
known, event frequencies alone would recover the stocks.

\begin{proposition}[What conduct cannot reveal]\label{prop:limits}
Maintain the baseline $v_f=0$, the normalization $\alpha_m=1$, and a single exposure regime, and consider
an analyst who observes conduct---event frequencies and actions---but neither motives nor the primitives
$(\phi_0,\psi,\nu_m,H,G,b_N)$. Then: (i)
whenever $X_m>0$, every observed independent act has a
posterior probability of being authentic strictly inside $(0,1)$, and a conforming agent may be unformed,
silenced, or privately withheld; (ii) the observed event frequencies confound the stocks $(J,A)$ with the
background rates and visibilities; and (iii) even granted $(J,A)$, conduct restricts the split of
non-enacted judgment into its silenced and privately withheld parts only to the segment
$\{(S^R,N)\ge0:\,S^R+N=J-A\}$. The segment is sharp on its relative interior when $J>0$ and the enactment
margin is strictly sanctioned: primitives differing only in the unobserved cost distributions generate
identical conduct at every margin while assigning any interior split.
\end{proposition}

Online Appendix OA.12 states the observational framework and shows that an invariant-stock
enactment-protection intervention, truthful private elicitation, and measured nuisance
levels---instruments assumed rather than derived---restore point identification of the full composition;
the proofs are in Online Appendix OA.13. What remains here is an implication an outside evaluator should
not miss: when deliberation is costly, normal performance ranks adaptive capacity in reverse---an
accounting implication of the output definitions, not a new equilibrium mechanism.

\begin{corollary}[Normal performance reverses the capacity ranking]\label{cor:flip}
When maintaining judgment carries a positive deliberation cost, $\gamma>0$, normal performance does not merely
fail to identify capacity; it ranks it in reverse (at $\gamma=0$ calm output is flat in the stock).
Familiar-state output
$Y_F=V_F+\gamma(1-J)$ is strictly decreasing in $J$, while competent novel-state output $Y_N^C$ is strictly increasing in
$A=a(J)$, itself increasing in $J$. For any pair of cultures $(J_h,A_h)>(J_\ell,A_\ell)$, the high-capacity
culture produces less in calm and more under novelty, and more in expectation iff
$p>\widehat p_{\ell h}$, where
\begin{equation}\label{eq:pbreak}
 \widehat p_{\ell h}:=
 \frac{\gamma(J_h-J_\ell)}
 {\gamma(J_h-J_\ell)+(V_N+D)(A_h-A_\ell)}\in(0,1).
\end{equation}
\end{corollary}

The reversal requires the continuing deliberation cost, not the formation-stigma mechanism itself: the
worked benchmark becomes unique at $\gamma=0$, though other admissible primitives can still support
multiplicity.\footnote{Along continuous coexisting branches an output reversal solves
$p=\widehat p_{\ell h}(p)$, uniquely so when the output difference is strictly increasing across the
coexistence interval. With practice-built competence (Section~\ref{sec:welfare} defines the learning-adjusted
novel output $Y_N^L$), the same comparison replaces the competent novel-output gap by
$Y_N^L(A_h)-Y_N^L(A_\ell)$ whenever that gap is positive; Online Appendix OA.4 gives a sufficient
condition. Because $Y_N^L(A)\le Y_N^C(A)$ whenever competence is at most one, weak practice further
depresses novel output in the trap; the pairwise reversal threshold rises at the worked benchmark.}

\paragraph{Scope of the reversal.} The strict reversal treats deliberation as pure option value in calm
states. If a fraction $\rho_c\in[0,1)$ of cue-labeled-familiar episodes are covert novelties that only a
formed agent's diagnostic detects (net gain $\Delta_c>0$, no additional registered event), the equilibrium
system is unchanged but calm-labeled output acquires slope $-(1-\rho_c)\gamma+\rho_c\Delta_c$, so
Corollary~\ref{cor:flip} survives iff $(1-\rho_c)\gamma>\rho_c\Delta_c$: the reversal is confined to
settings where covert novelty is rare, or its gains small or slow to attribute (derivation in Online
Appendix OA.5).

\paragraph{Diagnostic content.} The timing asymmetry of Corollary~\ref{cor:stock-flow} yields the
familiar stock--flow signatures---an immediate flow response to enactment protection and to its
withdrawal, a delayed stock response to formation support, and heterogeneous post-withdrawal persistence
across otherwise comparable units starting between the cutoffs, since continuation expectations are
unobserved---consistency checks the mechanism shares with any stock--flow adjustment
technology.\footnote{The civic counterparts are as in
Section~\ref{subsec:outscope}: legal protection of dissent for the chair's protection, civic formation for
training, generational replacement for seat rotation.} What is specific to the mechanism is the sanction's
signature profile: it is \emph{act-contingent} and levied before outcomes are realized,
\emph{prevalence-priced}---the penalty attached to a given probe or dissenting vote falls with the
prevalence of the practice, because the posterior \eqref{eq:posterior} is decreasing in $X_m$---and
insensitive to individual vindication within the episode; under career concerns the penalty is instead
outcome-contingent and individually priced \citep{Holmstrom1999,ScharfsteinStein1990}. One fingerprint
discriminates most sharply: under career concerns a record that licenses departures reveals
\emph{competence}, whereas here it reveals \emph{motive}---authentic members probe only on the diagnostic
cue, so a \emph{sparse} probing record, not a successful one, identifies authenticity.\footnote{No single comparative static separates the two regimes: blame-sharing and
relative evaluation also make the cost of deviating fall with the mass deviating
\citep{ScharfsteinStein1990,Zwiebel1995}, and career-concern conformity fades as an individual's record
lengthens \citep{ChevalierEllison1999,HongKubikSolomon2000}, the same direction in which record visibility
erodes the anonymous channel here (Online Appendix OA.8, a bound taken at fixed prevalence, which loses
force at trap-level stocks). The sparse-record likelihood ratio is in Online Appendix OA.8.}
These are testable implications, not an unconditional welfare case for more independent action.

\section{Welfare benchmarks}\label{sec:welfare}

The planner--equilibrium wedge combines adaptive output, learning, and sanctions, so the whole gap should not
be attributed to reputation. Its sanction component has a timing source: withdrawal is optimal after the act
but chosen after judgment is formed, so no partner internalizes its effect on future formation.
Theorem~\ref{thm:central-welfare} characterizes the combined wedge; call a margin \emph{effectively
sanctioned} when its sanction survives exposure ($\psi s_a(A_j)>0$ or $\phi s_f(J_j)>0$---since $\psi>0$
is a primitive, enactment always qualifies when strictly sanctioned, and formation exactly when
$\phi>0$).

The planner benchmarks use the competence technology $\sigma(A)$ and the departure value
$\Delta(\sigma)$ of \eqref{eq:delta-sigma}, introduced for the sequencing result in
Section~\ref{subsec:sequencing}. Recall that $\Delta(1)=V_N+D$ and that $\Delta(\sigma)$ turns negative below
the break-even competence threshold \eqref{eq:sigma-bar}: the conditionality below reflects confident error, not technical
caution. The resulting novel-state output is $Y_N^L(A):=A\Delta(\sigma(A))-D$,
where the superscript distinguishes learning-built output from the competent benchmark $Y_N^C$ of
Section~\ref{sec:model}. Departures are public experiments---each release generates information for the group much as
trial-and-error search does \citep{Callander2011}---generating the externality that
Theorem~\ref{thm:central-welfare}(i) prices; Online Appendix OA.1 derives the learning technology as a
transparent mean-count approximation.

Consider a planner benchmark that treats standing losses as transfers. Let
$C(J)=\int_0^JH^{-1}(u)\,du$ and $K(q)=\int_0^qG^{-1}(u)\,du$; assume $\E|k|<\infty$ and
$\E|\ell|<\infty$, so that $C$ and $K$ are finite and continuous on $[0,1]$; both moment conditions hold
at the logistic benchmark.\footnote{The left tail of $\ell$ is already controlled by the maintained
finiteness of $\Omega$; the right tail is what $K(1)$, reached on the feasible boundary $A=J$,
additionally requires.} Stationary welfare is
\begin{equation}\label{eq:welfare}
 \mathcal W(J,A;p)=(1-p)[V_F+\gamma(1-J)]
 +p[A b_N+A\Delta(\sigma(A))-D]-C(J)-pJ K(A/J).
\end{equation}
Its four terms are familiar-state output net of diagnostic costs, novel-state surplus (the private benefit
$Ab_N$ plus adaptive output), aggregate formation costs, and aggregate resolve costs, with the convention
$JK(A/J):=0$ at $J=0$ (the enacting share is empty there); standing losses do not
appear---the benchmark treats them as transfers---and the background source stays outside the ledger, so
welfare comparisons are conditional on excluding both the harm undeterred opposition causes and the
sanctions it bears (Online Appendix OA.3 prices the former; the closing paragraph returns to the
accounting).
The resulting planner is the same-technology benchmark,
\[
 (J^{\ast\ast},A^{\ast\ast})\in\argmax_{0\le A\le J\le1}\ \mathcal W(J,A;p),
\]
facing the same endogenous competence as the decentralized cultures. The fully competent benchmark is its
$\sigma\equiv1$ specialization: define
\[
 w^\ast:=G(b_N+V_N+D),\qquad
 \Omega_S:=\E[(b_N+V_N+D-\ell)_+],\qquad
 \msv(p):=p\Omega_S-(1-p)\gamma,
\]
with unique allocation $J^\ast=H(\msv(p))$ and $A^\ast=J^\ast w^\ast$ (derivation in OA.1). Two
qualifications are independent and should be kept distinct: transfers versus real match-surplus accounting
(welfare accounting), and endogenous versus perfect competence (the accuracy technology).

\begin{theorem}[Learning externality and conditional underprovision]\label{thm:central-welfare}
Suppose $\sigma(A)$ is increasing, concave, and $C^1$; standing losses are treated as transfers; and the
baseline private-payoff specification is maintained, so that private enactors do not internalize output
consequences (the internalization parameter $\rho$ of Online Appendix OA.1 is zero).
\begin{enumerate}
\item[(i)] \emph{(Enactment wedge.)} At a permanent culture, the planner's marginal value of enactment
differs from the private cutoff by
\begin{equation}\label{eq:learning-wedge}
 \mathcal L_a(A):=\Delta(\sigma(A))+A\sigma'(A)\Delta'(\sigma(A))+\psi s_a(A).
\end{equation}
The learning term $A\sigma'(A)\Delta'(\sigma(A))$ is nonnegative, and strictly positive wherever $A>0$ and
$\sigma'(A)>0$; wherever $\mathcal L_a(A)>0$, enactment pressure thus creates a positive marginal
private--social enactment wedge.
\item[(ii)] \emph{(Conditional underprovision against the same-technology planner.)} Let
$(J^{\ast\ast},A^{\ast\ast})$ be an interior optimum of the planner facing the same endogenous
competence. If
\[
 \Delta(\sigma(A^{\ast\ast}))+A^{\ast\ast}\sigma'(A^{\ast\ast})\Delta'(\sigma(A^{\ast\ast}))\ge0,
\]
every equilibrium with at least one effectively sanctioned margin has
$J_j<J^{\ast\ast}$ and $A_j<A^{\ast\ast}$.
\item[(iii)] \emph{(Unconditional comparison with the competent benchmark.)} Independently of the sign
condition, every equilibrium satisfies $J_j<J^\ast$ and $A_j<A^\ast$. The gap is strict even without any
sanction, reflecting the uninternalized social value $V_N+D$ of a competent departure; the fully competent
planner is a counterfactual upper benchmark, not the same-technology comparison.
\end{enumerate}
\end{theorem}

The proof is in Appendix~\ref{app:welfare-proof}. The theorem compares decentralized cultures with
planners; it does not rank the two cultures. That
ranking depends on formation and resolve costs, the economy of routine, competence, and real
screening value; the worked benchmark's comparison is reported in Section~\ref{sec:numerical}. Online
Appendix OA.4 records how learning compounds the output gap between cultures.

\paragraph{Beyond the transfer convention.} The benchmark treats standing losses as transfers, whereas the
matching microfoundation of Section~\ref{sec:trap} gives withdrawal real match-surplus consequences. Three
statements survive the change in accounting. The global planner--equilibrium ordering is robust to small
match-surplus terms: Online Appendix OA.2 gives $\bar\varrho>0$ such that underprovision survives for every
$\varrho\in[0,\bar\varrho)$ at any equilibrium with an effectively sanctioned margin---a bound that need not
reach full accounting ($\varrho=1$). Under full accounting the statement is local and one-directional:
marginal enactment protection is guaranteed to raise surplus whenever the released act satisfies
$V_N+D>\psi\bigl[\Lambda_as_a+\tfrac{c_a}{2}s_a^2\bigr]$, a sufficient condition since the release also
relaxes the stigma priced on inframarginal departures. At inherited competence $\sigma$ the same comparison
replaces $V_N+D$ by $\Delta(\sigma)$, and a permanent stationary comparison adds the learning term of
\eqref{eq:learning-wedge}. And when undeterred opposition causes large real harm, pressure becomes locally
valuable (Online Appendix OA.3). What survives the accounting is the direction of the distortion at the
margin, not a case against pressure as such.

\section{A fully coupled worked example}\label{sec:numerical}

This section applies the equilibrium, output, welfare, and duration results to a single fully coupled
parameter vector---a constructed example, not a calibration. Formation and resolve costs are
logistic and the audience follows the hard Bayesian withdrawal rule, which remains on its sanctioned interior
branch over the whole feasible domain: multiplicity arises because prevalence changes the intensity of
withdrawal, not because any equilibrium crosses a tolerance threshold. Online Appendix OA.11 reports the
exact-decimal inputs, which the computer-assisted proofs treat as rational numbers.

\paragraph{Stationary cultures.}
At $p=\pbench$ the reduced map has fixed points
$J_\ell=\Jlow$, $J_i=\Jmid$, and $J_h=\Jhigh$, with the stable--unstable--stable ordering under the reduced
scalar adjustment. Holding the other primitives fixed (under the Section~\ref{subsec:env} normalization
for $p$), coexistence
obtains for $p\in(\pfoldlo,\pfoldhi)$ within the certified novelty domain
$p\in[0,\tfrac12]$.\footnote{Beyond the
certified domain, uncertified numerical continuation indicates a single high-judgment equilibrium for
$p\in(\tfrac12,1)$, with $J$ approaching a value extremely close to---though, because $H$ has full support
and returns are finite, strictly below---one as $p\uparrow1$.} Below the lower fold only the trap survives, above the upper fold only the questioning
culture. Novelty is therefore infrequent throughout the coexistence region---between one episode in fifteen
and one in seven---as the environment of Section~\ref{sec:model} requires. Nor need the audience dislike
authentic independence: at $\bar q_f=\qbarf$ and $\bar q_a=\qbara$, a partner values excluding an opponent
about four to five times as much as it costs to lose an authentic one ($\pi_f/\Lambda_f=\ratiof$,
$\pi_a/\Lambda_a=\ratioa$). A familiar-state probe is recorded with probability $\phi_0=\phizero$ rather
than certainly, and the diagnostic cost $\gamma=\gammaval$ is under seven percent of familiar-state
output. The window in
$p$ is nonetheless narrow: the example is chosen for certifiability, and Appendix~\ref{app:folds} guarantees
persistence of its structure under sufficiently small $C^2$ perturbations, not a wide coexistence region.

The audience's posterior interpretation shifts sharply across the two cultures: the posterior probability that a probe comes from
the oppositional source falls from $\qflow$ in the trap to $\qfhigh$ in the questioning culture, and the
corresponding posterior for a departure falls from $\qalow$ to $\qahigh$. In the trap, that is, suspicion
is largely correct---more than nine in ten observed independent acts come from the background source---so
withdrawal is ex-post optimal and the audience makes no inference error; the pathology lies in what that
correct pricing does to formation.
Table~\ref{tab:cultures} reports the two stable cultures.
\begin{table}[H]\centering
\caption{Stable cultures at $p=0.10$.}
\label{tab:cultures}
\small\setlength{\tabcolsep}{4.5pt}
\begin{tabular}{lcccccccc}
\toprule
culture & $J$ & $w$ & $A$ & $S^R$ & $N$ & $Y_F$ & $Y_N^C$ & $Y^C$\\
\midrule
low  & $\Jlow$ & $\wlow$ & $\Alow$ & $\SRlow$ & $\Nlow$ & $\YFlow$ & $\YNlow$ & $\Ylow$\\
high & $\Jhigh$ & $\whigh$ & $\Ahigh$ & $\SRhigh$ & $\Nhigh$ & $\YFhigh$ & $\YNhigh$ & $\Yhigh$\\
\bottomrule
\end{tabular}\\[4pt]
\begin{minipage}{0.88\textwidth}
\footnotesize\emph{Notes:} Rows: \emph{low} is the conformity trap, \emph{high} the questioning culture.
$J$ formed judgment; $w=G(z(A))$ enactment rate of the formed; $w^0=G(b_N)=\wzero$ the counterfactual
enactment rate absent reputational pressure; $A=Jw$ enacted judgment; $S^R=J[w^0-w]$ reputationally silenced;
$N=J(1-w^0)$ privately withheld; $Y_F$, $Y_N^C$, $Y^C$ familiar-state, competent novel-state, and competent
expected output (output measures, not welfare). All entries are computed at the unrounded equilibrium roots
and then rounded independently, so $A+S^R+N$ can differ from $J$ in the last displayed digit; the identity
$A+S^R+N=J$ holds exactly at the unrounded roots. Statistics quoted in the text (e.g.\ $q_a=\qalow$) likewise
use the unrounded roots.
\end{minipage}
\end{table}
Three features of the table carry the section's economics. In the trap almost no one maintains judgment and
much of the surviving stock is silenced: $S^R/J=\SRsharelow\%$ against $\SRsharehigh\%$ in the questioning
culture, with the trap's absolute silenced mass the larger of the two ($\SRlow$ versus $\SRhigh$) though its
stock is one thirty-sixth the size. The questioning culture enacts almost everything it forms ($w=\whigh$).
And the trap economizes on familiar-state diagnosis ($Y_F=\YFlow$ versus $\YFhigh$) while the questioning
culture dominates under novelty ($Y_N^C=\YNhigh$ versus $\YNlow$).

\paragraph{Output and the reversal.}
Because novelty is rare here, those two effects do not cancel in favor of the questioning culture: at the
benchmark the \emph{trap} produces the higher expected competent output ($Y^C=\Ylow$ versus $\Yhigh$).
Re-solving the two stable branches as $p$ varies, the questioning culture overtakes the trap in $Y^C$ only at
$p=\phat$, which lies inside the coexistence window $(\pfoldlo,\pfoldhi)$.\footnote{The displayed value is
computed from the reversal condition \eqref{eq:pbreak} at the benchmark allocations; re-solving both
branches locates the crossing at the same displayed precision.} The output ranking of the two cultures
therefore reverses \emph{within} the region where both exist. This is the sharpest form of
Corollary~\ref{cor:flip}: over a range of environments in which a questioning culture is sustainable, an
observer who ranked the two groups by realized output would prefer the trap---right about output, wrong
about adaptive capacity. Nor is the trap an unambiguous inefficiency: the welfare comparison below does
not support that reading either.

\paragraph{Competence and the impact sign.}
With endogenous competence, two verifications must be kept distinct. The cultures sit on opposite sides of
the break-even competence threshold $\bar\sigma=(L-D)/(V_N+L)$ of \eqref{eq:sigma-bar}, at which an additional departure
has zero incremental value:
$\sigma(A_\ell)=\sigmalow<\bar\sigma=\sigmabar<\sigmahigh=\sigma(A_h)$,
so a marginal departure destroys value in the trap and creates it in the questioning culture. This is the
configuration Proposition~\ref{thm:sequencing} describes: the zero-value enacted mass is $A^\dagger=\Adagger$,
corresponding to the stock threshold $J^\dagger=\Jdagger$, which is not itself a stationary stock; the
stationary cultures straddle it. One-date enactment protection applied to a stationary culture below that
threshold---here, the trap---therefore releases judgment with negative expected adaptive value. Separately, the \emph{level} of learning-built
novel output, $Y_N^L=A\Delta(\sigma(A))-D$, is $\YNLhigh$ at the high culture against $\YNLlow$ in the trap
(both differing from the competent levels in Table~\ref{tab:cultures}).

\paragraph{Welfare and the planner.}
This parameter vector satisfies the sign condition of Theorem~\ref{thm:central-welfare}(ii), and the
certified unique planner for these primitives, $(J^{\ast\ast},A^{\ast\ast})=(\plannerJ,\plannerA)$, chooses
more formed and more enacted judgment than every decentralized equilibrium; uniqueness is computationally
certified for this parameter vector, not a consequence of concavity. The questioning culture has the higher
stationary welfare, $\welfarehigh$ versus $\welfarelow$, and comes close to the planner's $\welfareplanner$---so the output
comparison above and the welfare comparison point in opposite directions. What carries the gap is not output: term by term in
\eqref{eq:welfare}, the enactors' private benefit $pAb_N$ contributes $\wgappsy$ of the $\wgap$ gap and
the cost side $\wgapcost$ (the median net formation cost is negative, $\mu_k=\mukval$, so about two thirds of
members derive intrinsic value from forming), while measured output \emph{nets against} the questioning
culture ($\wgapfam$ familiar against $\wgapadapt$ adaptive).\footnote{The gap and its components are
computed at the unrounded roots and rounded independently, as in Table~\ref{tab:cultures}, so the
components need not sum exactly and the gap differs in the last digit from the difference of the displayed
welfare levels.} The welfare case here therefore rests on the taste
components $b_N$ and $\mu_k$, not on production, and the proximity to the planner should be read the same
way. Nor is $b_N$ a free dial: it moves behavior as well as welfare. Re-solved at lower enactment
benefits, coexistence itself fails below $b_N=\bnDagger$, and at $b_N=1.4$ the gap narrows to
$\wgapbn$. Silencing on the trap's scale requires stakes on $b_N$'s scale; the taste-driven gap is a feature
of the mechanism, not an artifact of one number. The stake is not a knife-edge either: mapped jointly over
$(b_N,\kappa_f)$, coexistence occupies a widening wedge---at the benchmark pressure it requires
$b_N>\bnDagger$, while every tested stake down to $b_N=1.20$ supports a nonempty pressure window whose
width grows with the stake, from $\bnWindowLow$ at $b_N=1.2$ to $\bnWindowTwo$ at $b_N=2$ (Online
Appendix OA.6). Online Appendices OA.1 and OA.11 report the
competence inputs, planner solutions, wedge decomposition, and the uniqueness certificate.

\paragraph{Escape and duration.}
The dynamic exercise asks how long sufficiently strong direct formation support must last. The certified
escape boundaries are $J_-^{RE}\in[\REminlo,\REminhi]$ and $J_+^{RE}\in[\REpluslo,\REplushi]$
(Figure~\ref{fig:dynamic-regions}); from $J_0=J_\ell=\Jlow$ at the benchmark turnover and patience, three
intervention dates let the physical ceiling clear the first bracket and four dates the second, so
$T_{\mathrm{possible}}=\Tpossible$ and $T_{\mathrm{robust}}=\Trobust$.\footnote{These are
feasibility thresholds---the minimum durations at which the physical ceiling clears each boundary---not
solutions to an intertemporal cost-minimization problem, and ``all continuation expectations'' means all
perfect-foresight continuations admitted by the model, not robustness to shocks or misspecification.}
Online Appendix OA.10
derives the conservative envelope-based support schedules, and OA.11 certifies schedules attaining the
bounds.\footnote{One numerical coincidence: the break-even inherited stock $J^\dagger=\Jdagger$ of
Proposition~\ref{thm:sequencing} falls inside the certified enclosure $[\REminlo,\REminhi]$ for the lower
convergence cutoff $J_-^{RE}$---a property of these numbers, not a theorem, and one that does not order the
two thresholds. The stationary trap sits below both: protection there releases judgment with negative
expected adaptive value \emph{and} no continuation from the trap converges high.}

\paragraph{The focus margin.}
The strict separation of the certified brackets is also implied by the slope of the static reduced map
together with the revision primitives. At the certified middle
slope $m_i$, the composite root parameter of \eqref{eq:forward-characteristic} is
$\vartheta_i=\thetamid$, strictly inside the complex-root bound $2\sqrt{d/\lambda}=\focusbound$, so the
middle culture is an unstable focus and Theorem~\ref{thm:band} applies: the expectation-dependent band is
strictly open around $J_i$ as a matter of theory, and the certificates locate its boundaries.\footnote{The
comparison is rigorous: $\vartheta_i$ is strictly decreasing in $m_i$, so the certified enclosure for
$m_i$ (Online Appendix OA.9) maps to an enclosure for $\vartheta_i$ well inside the bound. The thin margin---$\thetamid$ against
$\focusbound$---is consistent with the narrowness of the certified band.}

\paragraph{What is certified.}
The Arb certificates enclose and exhaustively count all equilibrium roots and both folds, verify the
benchmark planner, and certify the enclosures for $J_-^{RE}$ and $J_+^{RE}$: a rigorous existence witness
for the fully coupled model. Online Appendix OA.11 records the proof
obligations; OA.6 and OA.7 report robustness to alternative sanction technologies and responsive opposition.

\section{Conclusion}\label{sec:conclusion}

Approval can regulate not only visible conduct but the upstream practice that keeps independent judgment
available. When motives remain ambiguous, a Bayesian audience reads rare questioning as opposition and
withdraws cooperation, and the resulting prevalence feedback can sustain a conformity trap alongside a
questioning culture. Two ingredients carry the mechanism in different roles: residual ambiguity
($\nu_m>0$) is essential for any prevalence feedback, whereas formation visibility ($\phi>0$) matters only
for the formation-stigma route---at $\phi=0$ the enactment channel can still sustain multiplicity on its
own.

Staggered revision makes aggregate maintenance an inherited group-level state. Enactment protection can
release existing judgment on impact, whereas formation support changes that state only as members revise. When
escape under every admissible continuation is attainable, the extremal perfect-foresight paths classify
inherited stocks---escape impossible, expectation-dependent, or assured---with the middle regime strictly
present whenever the unstable culture is a local focus; the duration bounds then state how long temporary
formation support must last to cross each boundary. No temporary intervention at fixed turnover can rebuild the stock faster than turnover permits, and the same observed stationary culture can conceal
different transition risks.

What protection accomplishes is itself inherited. Because competence at diagnosing novelty is the
steady-state accuracy sustained by a culture that enacts judgment, protecting dissent below the break-even
threshold releases judgment whose expected adaptive value on impact is negative---an accounting at
inherited competence, before anything the released departures teach successors---and a model-implied stock
threshold separates the impact sign across stationary inherited cultures. Whether that makes rebuilding the efficient first move is
a question this model poses without answering, since protection also raises the option value of holding
judgment; a dynamic competence state would be needed to settle it. The scope of the mechanism is
correspondingly sharp rather than universal: it requires opposition that is not fully deterrable by the stigma
it attracts, and audiences whose memory of individuals is short relative to the rate at which a record would
identify a type.

Silence pools members who never maintained judgment with those who withhold or suppress it, while
appropriately timed interventions separate release from rebuilding. The planner comparison identifies
conditional underprovision: at the worked benchmark the conformity trap
even produces \emph{more} expected output than the questioning culture over part of the region where both
exist, which is precisely what makes the capacity it destroys hard to see. The central distinction is
therefore temporal as well as informational. A group that rewards praise over praiseworthiness, in Smith's
distinction, may not merely silence the judgment it has; it may gradually stop maintaining the readiness on
which future correction depends. Read at the scale of a polity, this offers one reason why civic passivity
can outlive the
regimes that produce it, and why its repair may be constrained by cohort turnover rather than by the
calendar of a campaign.

\appendix
\numberwithin{theorem}{section}
\numberwithin{proposition}{section}
\numberwithin{lemma}{section}
\numberwithin{corollary}{section}
\section{Threshold geometry and the number of cultures}\label{app:atmost}\label{app:logit-geometry}

\emph{Roadmap.} This appendix establishes the exact-count condition and proves Theorems~\ref{thm:trap} and~\ref{thm:analytic-threshold}; Appendix~\ref{app:folds} collects the fold normal form, robustness to primitive perturbations, the primitive coexistence conditions (Theorem~\ref{thm:coupled-coexistence}), and responsive opposition; Appendix~\ref{app:dynproofs} proves the dynamic theorems and Appendix~\ref{app:welfare-proof} the welfare theorem. Every numbered theorem is thus proved in the manuscript; proofs of propositions, lemmas, and corollaries, together with the extensions, are in the Online Appendix.

The prevalence-dependent posterior is what makes multiplicity possible: with sanctions fixed at both margins
the equilibrium is unique (Proposition~\ref{prop:nosocialmeaning} in the main text), so a standard disclosure
friction cannot create the trap. The exact count of cultures, when the posterior does vary with prevalence, is
pinned down by a shape condition on the primitive formation-pressure threshold function.

\begin{proposition}[Exact count under a threshold-shape condition]\label{prop:atmost}
Fix $p$ and a smooth multiplicative formation branch $s_f(J;\kappa_f)=\kappa_f b_f(J)$ with $b_f>0$,
valid on all of $(0,1)$ at every pressure considered. Let
$\mathcal K(J;p)$ be the primitive formation-pressure threshold function in \eqref{eq:pressure-threshold}. Suppose
$\mathcal K$ is continuously differentiable and there exist
$0<J_-<J_+<1$ such that
$\mathcal K'<0$ on $(0,J_-)$, $\mathcal K'>0$ on $(J_-,J_+)$, and $\mathcal K'<0$ on $(J_+,1)$.
Then for every positive formation pressure $\kappa_f$, the reduced equilibrium equation has at most three
interior fixed points. Hence any pressure interval delivered by Theorem~\ref{thm:trap} contains exactly
three fixed points, with the outer two stable and the middle one unstable whenever the crossings are transverse.
The condition is primitive for the construction: $\mathcal K$ is built from $B$, $H$, and $b_f$, and its shape
can be checked before solving for equilibrium roots. At $p=0.10$, the computer-assisted certificate computed with Arb and documented in Online Appendix OA.11 establishes the
condition globally for the exact-decimal full baseline, isolates the two stationary points, and certifies the
monotonicity pattern $(-,+,-)$.
\end{proposition}

\begin{proof}
For a positive pressure on the smooth branch,
\[
 \Psi(J;\kappa_f)=J
 \quad\Longleftrightarrow\quad
 \kappa_f=\mathcal K(J;p).
\]
The stated shape partitions $(0,1)$ into three monotone intervals, so any horizontal line $\kappa_f$ has at
most one intersection in each. Theorem~\ref{thm:trap} supplies at least three crossings in the interval it
supplies; transversality gives the alternating scalar-stability signs.
\end{proof}

\begin{proof}[Proof of Theorem~\ref{thm:trap}]
Since $H$ is strictly increasing and $b_f(J)>0$,
\[
 \Psi(J;\kappa_f)>J
 \quad\Longleftrightarrow\quad
 \kappa_f<\mathcal K(J;p),
\]
with the inequality reversed for $\Psi(J;\kappa_f)<J$. For any
$\kappa_f\in\mathcal I_{LH}(p)$, therefore,
$F(J_L;\kappa_f):=\Psi(J_L;\kappa_f)-J_L<0$ and
$F(J_H;\kappa_f)>0$. Full support of $H$ gives $F(0;\kappa_f)>0$ and
$F(1;\kappa_f)<0$. The intermediate value theorem then gives one fixed point in each of
$(0,J_L)$, $(J_L,J_H)$, and $(J_H,1)$. The threshold-shape condition rules out further crossings; when the
crossings are transverse, the alternating sign pattern $+,-,+,-$ gives the stated scalar-stability signs. A
nondegenerate endpoint of a maximal three-crossing interval is a tangency satisfying the hypotheses of
Lemma~\ref{lem:regular-fold}.

For the converse, let $J_1<J_2<J_3$ be the transverse fixed points at $\kappa_f$. Since $F(0;\kappa_f)>0$ and
the crossings are transverse, $F<0$ on $(J_1,J_2)$ and $F>0$ on $(J_2,J_3)$. Pick any
$J_L\in(J_1,J_2)$ and $J_H\in(J_2,J_3)$: the displayed equivalence gives $\kappa_f>\mathcal K(J_L;p)$ and
$\kappa_f<\mathcal K(J_H;p)$, so $\mathcal K(J_H;p)>\max\{0,\mathcal K(J_L;p)\}$ and
$\kappa_f\in\mathcal I_{LH}(p)$.
\end{proof}

The threshold-shape condition above is exactly what the logit--proportional benchmark delivers globally.
The remainder of this appendix proves Theorem~\ref{thm:analytic-threshold}, stated in
Section~\ref{sec:trap}, and establishes the persistence of its geometry under small enactment feedback.

\begin{proof}[Proof of Theorem~\ref{thm:analytic-threshold}]
Write
\[
 L(J):=\log\frac{J}{1-J},\qquad \varpi(J):=r(1-\bar q_f)-\bar q_fJ.
\]
The restriction $\bar q_f<q_f(1)=r/(r+1)$ implies $\varpi(J)>0$ on $[0,1]$, and
\[
 b_f(J)=\frac{\varpi(J)}{r+J},\qquad b_f'(J)=-\frac{r}{(r+J)^2}.
\]
Define
\[
 \eta(J):=-\frac{b_f'(J)}{b_f(J)}=\frac{r}{(r+J)\varpi(J)}>0
\]
and
\[
 f(J):=L(J)+\frac{L'(J)}{\eta(J)}
 =\log\frac{J}{1-J}
 +\frac{(r+J)[r(1-\bar q_f)-\bar q_fJ]}{rJ(1-J)}.
\]
Since
\[
 \mathcal K(J;p)=\frac{\varsigma}{(1-p)\phi}\frac{c_0(p)-L(J)}{b_f(J)},
\]
differentiation gives the exact sign reduction
\begin{equation}\label{eq:Kprime-sign-reduction}
 \mathcal K'(J;p)=
 \frac{\varsigma\eta(J)}{(1-p)\phi b_f(J)}[c_0(p)-f(J)].
\end{equation}
A second direct differentiation yields
\begin{equation}\label{eq:fprime-canonical}
 f'(J)=
 \frac{[r(1-\bar q_f)-\bar q_fJ][(1+2r)J-r]}
 {rJ^2(1-J)^2}.
\end{equation}
All factors other than the second bracket in the numerator are positive. Hence $f$ is strictly decreasing on
$(0,J^\circ)$ and strictly increasing on $(J^\circ,1)$, where $J^\circ=r/(1+2r)$. Moreover,
$\lim_{J\downarrow0}f(J)=\lim_{J\uparrow1}f(J)=+\infty$,
and substitution of $J^\circ$ gives $f(J^\circ)=m(r,\bar q_f)$. It follows that $f(J)=c_0(p)$ has exactly two
solutions $J_-<J^\circ<J_+$ if and only if \eqref{eq:global-logit-inequality} holds. Equation
\eqref{eq:Kprime-sign-reduction} then gives the global derivative pattern. If the inequality fails, the same
equation gives the stated monotonicity conclusion.

At either stationary point $J_s\in\{J_-,J_+\}$, $c_0-L=L'/\eta$, and therefore
\[
 \mathcal K(J_s;p)=
 \frac{\varsigma L'(J_s)}{(1-p)\phi\,\eta(J_s)b_f(J_s)}>0.
\]
Also, $\lim_{J\downarrow0}\mathcal K(J;p)=+\infty$ and $\lim_{J\uparrow1}\mathcal K(J;p)=-\infty$.
Thus the graph falls from $+\infty$ to its positive minimum, rises to its strictly higher positive maximum,
and then falls to $-\infty$. A horizontal line at a positive $\kappa_f$ therefore has three intersections
exactly under \eqref{eq:canonical-pressure-window}, one intersection outside its closure, and one simple plus
one double intersection at either endpoint.

For stability and fold nondegeneracy, let
\[
 Q(J,\kappa_f):=B_0(p)-H^{-1}(J)-(1-p)\phi\kappa_fb_f(J)
 =(1-p)\phi b_f(J)[\mathcal K(J;p)-\kappa_f].
\]
Writing $F(J,\kappa_f)=\Psi(J;\kappa_f)-J$ gives
\[
 F(J,\kappa_f)=H\!\left(H^{-1}(J)+Q(J,\kappa_f)\right)-J.
\]
At any simple fixed point, $\operatorname{sgn}F_J=\operatorname{sgn}\mathcal K'$. The reduced best response
is increasing, so the negative signs at the two outer intersections imply $0<\Psi'<1$, while the positive sign
at the middle intersection implies $\Psi'>1$. Finally, at either stationary point,
\eqref{eq:fprime-canonical} and \eqref{eq:Kprime-sign-reduction} imply
$\mathcal K''(J_-;p)>0$ and $\mathcal K''(J_+;p)<0$.
At the corresponding pressure endpoint,
\[
 F_{\kappa_f}=-h(H^{-1}(J))(1-p)\phi b_f(J)\neq0,\qquad
 F_{JJ}=h(H^{-1}(J))(1-p)\phi b_f(J)\mathcal K''(J;p)\neq0.
\]
Both endpoints are therefore nondegenerate saddle-nodes.
\end{proof}

\begin{corollary}[Persistence under small enactment feedback]\label{cor:logit-feedback}
Maintain the formation primitives and the strict inequality \eqref{eq:global-logit-inequality}. Let
$B_\epsilon(\cdot;p)\in C^2[0,1]$ be a family of induced zero-pressure formation returns such that
\[
 \lVert B_\epsilon(\cdot;p)-B_0(p)\rVert_{C^2[0,1]}\longrightarrow0
 \qquad\text{as }\epsilon\downarrow0,
\]
and define
\[
 \mathcal K_\epsilon(J;p)=
 \frac{B_\epsilon(J;p)-H^{-1}(J)}{(1-p)\phi b_f(J)}.
\]
There exists $\bar\epsilon>0$ such that, for every $\epsilon\in[0,\bar\epsilon)$,
$\mathcal K_\epsilon$ has exactly two stationary points $J_-(\epsilon)<J_+(\epsilon)$, converging to
$J_-$ and $J_+$, with the same global derivative pattern $(-,+,-)$. Their threshold values remain positive
and strictly ordered. The interval with exactly three equilibria, the stable--unstable--stable ordering, and the
two nondegenerate saddle-nodes therefore persist.

A sufficient primitive construction is
\[
 s_a^\epsilon(A)=s_a^0+\epsilon\widetilde s_a(A),\qquad
 z_\epsilon(A)=b_N-\psi s_a^\epsilon(A),
\]
where $G$ and $\widetilde s_a$ are $C^3$ and the enactment subgame is regular. If $a_\epsilon(J)$ is the unique
solution of $A=JG(z_\epsilon(A))$, then
\[
 B_\epsilon(J;p)=p\Omega(z_\epsilon(a_\epsilon(J)))-(1-p)\gamma
\]
satisfies the displayed $C^2$ convergence. In particular, this applies to a sufficiently small multiple of a
smooth or globally interior proportional Bayesian enactment schedule.
\end{corollary}

The proof is in Online Appendix OA.13.

\section{Primitive coexistence conditions and robustness}\label{app:folds}
This appendix collects the local fold form, robustness of the coexistence interval to small $C^2$ perturbations, the
primitive route to a pressure-threshold pair, and the responsive-opposition scope result, all referenced
from Section~\ref{sec:trap}.

\begin{lemma}[Local normal form at a nondegenerate fold]\label{lem:regular-fold}
Let $\theta\in\{\kappa_f,p\}$; assume the densities $h$ and $g$ are $C^1$ and the sanction branch is $C^2$
near the fold, so that $\Psi$ is $C^2$ in $(J,\theta)$ there, and suppose $\Psi(\cdot;\theta)$ has a
\emph{nondegenerate fold} at $(J^\ast,\theta^\ast)$:
\[
 \Psi(J^\ast;\theta^\ast)=J^\ast,\quad \Psi_J(J^\ast;\theta^\ast)=1,\quad
 \Psi_{JJ}(J^\ast;\theta^\ast)\neq0,\quad \Psi_{\theta}(J^\ast;\theta^\ast)\neq0 .
\]
Then in a neighborhood of $\theta^\ast$ the fixed points of $\Psi(\cdot;\theta)$ near $J^\ast$ number two on
one side of $\theta^\ast$ and none on the other, appearing as a stable--unstable pair: \emph{locally}, two
equilibria are born (or annihilated) as $\theta$ crosses $\theta^\ast$. Transversality holds everywhere for
novelty, $\Psi_{p}=h(\widehat R)\,\widehat R_p>0$ since $\widehat R_p=\Omega_P(a)+\gamma+\phi s_f(J)>0$
(Proposition~\ref{prop:capacity}); for reputational sensitivity it holds on the interior sanctioned branch,
$\Psi_{\kappa_f}=h(\widehat R)\,\widehat R_{\kappa_f}<0$ where $\widehat R_{\kappa_f}=-(1-p)\phi\,(q_f-\bar
q_f)<0$ (and $\Psi_{\kappa_f}=0$ where the rule is tolerated or capped).
\end{lemma}

\begin{proof}
At $(J^\ast,\theta^\ast)$, $F=F_J=0$ while $F_{JJ}=\Psi_{JJ}\neq0$ and $F_\theta=\Psi_\theta\neq0$. Expanding
$F$ and solving $F=0$ near the degenerate point gives
$\theta-\theta^\ast=-\tfrac12(F_{JJ}/F_\theta)(J-J^\ast)^2+o\big((J-J^\ast)^2\big)$, a parabola opening to one
side: two roots $J$ for each $\theta$ on that side, none on the other---the saddle-node normal form. The
transversality expressions are the $p$- and
$\kappa_f$-derivatives of $\widehat R$ (Proposition~\ref{prop:capacity} and the interior form of
\eqref{eq:endogenous-tolerance}), each nonzero on the sanctioned region; in general
$\widehat R_{\kappa_f}=-(1-p)\phi\,\partial_{\kappa_f}s_f$, and the statement's formula specializes it to
the hard interior branch, where $\partial_{\kappa_f}s_f=q_f-\bar q_f$.
\end{proof}

\begin{proposition}[Robustness of the trap to small primitive perturbations]\label{prop:openset}
Let $\theta$ be a scalar parameter, such as $p$ or $\kappa_f$, and let $\pi$ collect the remaining
primitives. Write
\[
 F(J,\theta;\pi):=\Psi_\pi(J;\theta)-J .
\]
Consider primitive perturbations in the induced $C^2$ topology: require uniform convergence of $F$, $F_J$,
$F_{JJ}$, $F_{J\theta}$, and $F_\theta$ on a compact rectangle $[0,1]\times\Theta$. Suppose that, at a baseline
primitive vector $\pi^0$, there is a nonempty open interval
$I^0=(\theta_L^0,\theta_U^0)\subset\Theta$ such that for every $\theta\in I^0$ the equation
$F(J,\theta;\pi^0)=0$ has exactly three transverse roots
\[
 J_\ell(\theta)<J_i(\theta)<J_h(\theta),
\]
with
\[
 F_J(J_\ell,\theta;\pi^0)<0,\qquad
 F_J(J_i,\theta;\pi^0)>0,\qquad
 F_J(J_h,\theta;\pi^0)<0 .
\]
Then for every compact interval $K\subset I^0$ with nonempty interior, there exists an open neighborhood
$\mathcal U_K$ of $\pi^0$ such that, for every $\pi\in\mathcal U_K$ and every $\theta\in K$,
$F(J,\theta;\pi)=0$ has exactly three transverse roots with the same stability signs. If, in addition, the
endpoints of $I^0$ are nondegenerate folds satisfying the conditions of Lemma~\ref{lem:regular-fold}, those folds
also persist under all sufficiently small perturbations of $\pi$. Hence the parameter region supporting two
stable cultures separated by an unstable threshold contains an open neighborhood of $\pi^0$; the result is not
a knife-edge feature of the baseline primitives.
\end{proposition}

The proof, together with a quantified perturbation margin, is in Online Appendix OA.13.

\medskip

The next result gives primitive conditions delivering a pressure-threshold pair in the fully coupled model,
with the enactment channel active; Section~\ref{subsec:trap} states its content informally.

\begin{theorem}[Coexistence with a regular enactment channel, allowing active feedback]\label{thm:coupled-coexistence}
Suppose formation costs are a location--scale family
$H_\varsigma^{-1}(J)=\mu_k+\varsigma Q_0(J)$,
where $Q_0$ is continuous on every compact subinterval of $(0,1)$ with the location normalization
$Q_0(1/2)=0$, so that $\mu_k$ is the median cost, and that the formation-sanction branch is
smooth and multiplicative, $s_f(J;\kappa_f)=\kappa_f b_f(J)$ with $b_f(J)>0$ and $b_f'(J)<0$,
on an interval containing $0<J_L<J_H<1$. Fix $p$ and enactment primitives satisfying
\eqref{eq:enactment-regularity} and with $s_a$ weakly decreasing in $A$---as the tolerance class
\eqref{eq:endogenous-tolerance} delivers---and let $a(J)$ be the induced enactment solution, so that
$B(J;p)=p\,\Omega(b_N-\psi s_a(a(J)))-(1-p)\gamma$ is weakly increasing in $J$. If
\begin{equation}\label{eq:coupled-condition}
 B(J_L;p)>\mu_k ,
\end{equation}
then there is $\bar\varsigma>0$ such that for every $\varsigma\in(0,\bar\varsigma)$ the pair $(J_L,J_H)$ is a
pressure-threshold pair, so the interval $\mathcal I_{LH}(p)$ of formation pressures is nonempty and open, and
every pressure in it yields at least three fixed points of the reduced map. Under the threshold-shape
condition of Proposition~\ref{prop:atmost}---with the branch valid on all of $(0,1)$---there are exactly
three; when the crossings are transverse they
form two stable cultures separated by one unstable threshold, and at any nondegenerate endpoint of a maximal multiplicity interval the entering or exiting
pair has the saddle-node form of Lemma~\ref{lem:regular-fold}.
\end{theorem}

The proof below is a monotonicity argument on the exposure ratio
$\Xi(J;p):=[B(J;p)-\mu_k]/b_f(J)$: wherever the zero-pressure formation return exceeds the median formation cost, the
numerator weakly increases---the enactment channel helps rather than hurts here---while the denominator
strictly falls, so $\Xi$ is strictly increasing. Economically, \eqref{eq:coupled-condition} is the
requirement, glossed in Section~\ref{subsec:trap}, that the zero-stigma formation return clear the median
cost somewhere; concentration converts that surplus into a pressure window. The result is sufficient, not
necessary---Theorem~\ref{thm:trap} supplies necessity for the transverse configuration, and exact
equivalence only under its threshold-shape and transversality hypotheses---and it shows that small dispersion generates
coexistence rather than that large dispersion removes it, which would require the global Lipschitz condition of
Proposition~\ref{prop:mult}.

\begin{proof}[Proof of Theorem~\ref{thm:coupled-coexistence}]
With $s_f(J;\kappa_f)=\kappa_f b_f(J)$,
\[
 \mathcal K_{\varsigma}(J;p)
 =
 \frac{B(J;p)-\mu_k-\varsigma Q_0(J)}{(1-p)\phi b_f(J)} .
\]
Condition \eqref{eq:enactment-regularity} implies that the unique enactment solution $a(J)$ is weakly increasing in $J$; because $s_a$ is weakly
decreasing, the cutoff $z(a(J))=b_N-\psi s_a(a(J))$ and hence $B(J;p)$ are weakly increasing. On every smooth branch,
\[
 B'(J;p)
 =p\,G(z(a(J)))\,\psi[-s_a'(a(J))]a'(J)\ge0 .
\]
Define the primitive exposure ratio
\[
 \Xi(J;p):=\frac{B(J;p)-\mu_k}{b_f(J)} .
\]
Because $b_f'(J)<0$, whenever $B(J;p)>\mu_k$,
\[
 \Xi'(J;p)
 =\frac{B'(J;p)b_f(J)-[B(J;p)-\mu_k]b_f'(J)}{b_f(J)^2}>0 .
\]
The hypothesis $B(J_L;p)>\mu_k$ and monotonicity of $B$ therefore imply
$0<\Xi(J_L;p)<\Xi(J_H;p)$. Since
$\mathcal K_{\varsigma}(J;p)\to \Xi(J;p)/[(1-p)\phi]$ uniformly on the compact interval
$[J_L,J_H]$ as $\varsigma\downarrow0$, the inequalities
$0<\mathcal K_{\varsigma}(J_L;p)<\mathcal K_{\varsigma}(J_H;p)$
hold for all sufficiently small $\varsigma>0$. Thus $(J_L,J_H)$ is a pressure-threshold pair. Theorem~\ref{thm:trap}
gives at least three fixed points for every pressure between those two threshold values, and
Proposition~\ref{prop:atmost} gives exactly three crossings under the threshold-shape condition, with the
stated stability pattern at transverse crossings.
Nondegenerate endpoints of the multiplicity interval are folds by Lemma~\ref{lem:regular-fold}.
\end{proof}

Finally, the scope result on responsive opposition discussed in Section~\ref{sec:trap}.

\begin{proposition}[Scope: responsive opposition]\label{prop:responsive}
For each margin $m$ let the opportunistic pool respond to its sanction through
\[
 \nu_m^{\rho}(s)=\bar\nu_m\,r_m(\rho_m s),
\]
where $r_m$ is $C^2$ and positive on the relevant compact sanction range, $r_m(0)=1$, $r_m'\le0$, and
$\rho_m\ge0$ indexes responsiveness; the smooth sanction then solves
$s=\kappa_m\bigl[\nu_m^{\rho}(s)/(\nu_m^{\rho}(s)+\alpha_mX_m)-\bar q_m\bigr]_\varepsilon$. Suppose the
fixed-pool baseline ($\rho_f=\rho_a=0$) has exactly three transverse roots with the stated stability signs
throughout a compact parameter set $K$ contained in the interior of its coexistence region. Then there is
$\bar\rho>0$ such that, whenever $\max_m\rho_m<\bar\rho$, the responsive-opposition model has exactly three
transverse roots with those stability signs for every parameter in $K$, and every nondegenerate fold on the
boundary of the fixed-pool coexistence region continues locally. Conversely, at any responsiveness and
parameter value for which the induced reduced map satisfies $\sup_J\Psi_{\rho}'(J)<1$, the equilibrium is
unique.
\end{proposition}

The proof, in Online Appendix OA.7, shows that small responsiveness perturbs the induced sanction, hence the
reduced fixed-point function, in the $C^2$ topology, so persistence follows from
Proposition~\ref{prop:openset}; OA.7 also microfounds $r_m$ from expressive opportunists with heterogeneous
benefits and derives two sufficient primitive conditions for weak responsiveness.

\section{Proofs of the dynamic theorems}\label{app:dynproofs}

This appendix proves Theorem~\ref{thm:history-expectations} (extremal paths and
expectation-dependent convergence), Theorem~\ref{thm:band} (the focus condition and the strict band), and
Theorem~\ref{thm:frontier} (duration bounds), stated in Section~\ref{sec:dynamics}.

\begin{proof}[Proof of Theorem~\ref{thm:history-expectations}]
Under the coordinatewise order, $\mathbb X(j)$ is a complete lattice with bottom
$\bot_j=(j,0,0,\ldots)$ and top $\top_j=(j,1,1,\ldots)$. The operator $\mathcal T_j$ maps this lattice into
itself and is isotone because $\widehat R$ and $H$ are nondecreasing. Tarski's theorem therefore makes its
fixed-point set a nonempty complete lattice. Moreover, if a sequence of paths increases or decreases
coordinatewise, bounded convergence in the discounted sum and continuity of $H$ and $\widehat R$ imply that
$\mathcal T_j$ commutes with its pointwise limit. Thus
\[
 \mathcal T_j^n\bot_j\uparrow\underline{\mathbf J}(j),
 \qquad
 \mathcal T_j^n\top_j\downarrow\overline{\mathbf J}(j),
\]
and both limits are fixed points. Induction places every fixed point between every lower and upper iterate, so
the two limits are respectively the least and greatest fixed points. A fixed point of $\mathcal T_j$ is exactly
the $J$ component of a perfect-foresight path, with $V$ recovered from \eqref{eq:forward-value}; this proves (i).

If $j\le j'$, the corresponding bottom and top paths are ordered, and induction through the common coordinate
formula \eqref{eq:path-operator} orders every lower and upper iterate. Passing to limits proves that both
extremal paths are nondecreasing in inherited stock. Every other equilibrium path is squeezed between them,
which proves (ii).

We first record two consequences used in (iii). For any perfect-foresight path, let
$l=\liminf_tJ_t$ and $u=\limsup_tJ_t$. Eventually every future stock lies in an arbitrarily small enlargement
of $[l,u]$. Monotonicity of $\widehat R$ in the discounted sum, followed by the stock equation, gives
$l\ge\Psi(\max\{l-\varepsilon,0\})$ and $u\le\Psi(\min\{u+\varepsilon,1\})$ for every $\varepsilon>0$;
continuity of $\Psi$ then delivers, as $\varepsilon\downarrow0$,
\begin{equation}\label{eq:liminf-limsup}
 l\ge\Psi(l),\qquad u\le\Psi(u).
\end{equation}
Under the stated sign chart, if paths $\mathbf X\le\mathbf Y$ and $Y_t\to J_\ell$, then
$X_t\to J_\ell$: its limsup is at most $J_\ell$, while the first inequality in
\eqref{eq:liminf-limsup} rules out a liminf below $J_\ell$. Symmetrically,
$\mathbf X\ge\mathbf Y$ and $Y_t\to J_h$ imply $X_t\to J_h$.

It follows that $\mathcal E_h^{\max}$ and $\mathcal E_h^{\min}$ are upper sets. For example,
$j'\ge j$ implies $\overline{\mathbf J}(j')\ge\overline{\mathbf J}(j)$, so high convergence of the latter
forces high convergence of the former. Moreover, high convergence of the least path forces that of the
greatest, hence $\mathcal E_h^{\min}\subseteq\mathcal E_h^{\max}$ and the ordering in
\eqref{eq:expectation-band}. If a path from $j<J_-^{RE}$ converged high, the greatest path above it would also
converge high, contradicting the infimum of the upper set $\mathcal E_h^{\max}$. If $j>J_+^{RE}$, the least
path converges high and every other path above it does as well. For $j$ strictly between the cutoffs,
$j>\inf\mathcal E_h^{\max}$ places $j$ in the upper set $\mathcal E_h^{\max}$, so the greatest path
converges high, while $j<J_+^{RE}$ excludes $j$ from $\mathcal E_h^{\min}$, so the least path does not:
convergence depends on the continuation. This proves (iii).
\end{proof}

The next lemma isolates the planar fact behind Theorem~\ref{thm:band}.

\begin{lemma}[Backward spiral crossings]\label{lem:spiral}
Let $F$ be a $C^1$ map from a neighborhood of $s^\ast\in\mathbb R^2$ to $\mathbb R^2$ with
$F(s^\ast)=s^\ast$, and suppose the eigenvalues $\mu,\bar\mu$ of $DF(s^\ast)$ are nonreal with $|\mu|>1$.
Let $\ell$ be any nonzero linear functional on $\mathbb R^2$. Then $F$ is invertible near $s^\ast$, and
there is a neighborhood $B$ of $s^\ast$---a ball in a norm adapted to
$DF(s^\ast)$---such that for every $s\in B\setminus\{s^\ast\}$ the backward
orbit $s_{-k}:=F^{-k}(s)$, $k\ge0$, is defined for all $k$, remains in $B$, converges to $s^\ast$, never
equals $s^\ast$, and satisfies $\ell(s_{-k}-s^\ast)<0$ for infinitely many $k$ and
$\ell(s_{-k}-s^\ast)>0$ for infinitely many $k$.
\end{lemma}

\begin{proof}
Nonreal eigenvalues make $DF(s^\ast)$ invertible, so $F$ is a local $C^1$ diffeomorphism and
$\mathcal B:=F^{-1}$ is $C^1$ near $s^\ast$ with $D\mathcal B(s^\ast)$ having the nonreal eigenvalues
$\nu,\bar\nu$, $\nu:=\mu^{-1}$, of modulus $\rho:=|\mu|^{-1}<1$; label the pair so that
$\nu=\rho e^{i\omega}$ with $\omega\in(0,\pi)$. A real $2\times2$ matrix with nonreal eigenvalues is
real-similar to $\rho$ times the rotation by $\omega$, so after a real-linear change of coordinates,
identified with $\mathbb C$, the backward map reads
\[
 v\;\longmapsto\;\nu v+r(v),\qquad v:=P(s-s^\ast),
\]
with $|r(v)|/|v|\to0$ as $v\to0$, by differentiability at $s^\ast$. Fix $\varepsilon\in(0,1)$ with
$\rho(1+\varepsilon)<1$ and $\varepsilon':=\arcsin\varepsilon<\min\{\omega,\pi-\omega\}/2$, and choose
$\tau>0$ with $|r(v)|\le\varepsilon\rho|v|$ on $\{|v|\le\tau\}$; let $B$ correspond to that disc. For
$|v|\le\tau$ and $v':=\nu v+r(v)$,
\[
 \rho(1-\varepsilon)|v|\;\le\;|v'|\;\le\;\rho(1+\varepsilon)|v|,
\]
so the backward orbit stays in the disc, converges geometrically to $0$, and never reaches $0$; and
$v'/(\nu v)=1+r(v)/(\nu v)$ with $|r(v)/(\nu v)|\le\varepsilon$, so the continuously unwound angle
$\theta_k:=\arg v_k$ advances per step by $\omega_k\in[\omega-\varepsilon',\omega+\varepsilon']
\subset(0,\pi)$ and increases to infinity. Write $\ell(s-s^\ast)=\mathrm{Re}(\bar c\,v)$ for some
$c\neq0$. Around each angle $\arg c+\pi+2\pi n$, $n\in\mathbb N$, place an arc of length
$\omega+2\varepsilon'$; by the choice of $\varepsilon'$ this length is below $\pi$, and each step of
$(\theta_k)$ is strictly shorter than it, so the increasing angle sequence cannot jump across any of
these arcs and visits every arc lying above its initial value---hence all but finitely many of them,
still infinitely many of each family. At any such visit the angular distance to the minimizing direction
is below $(\omega+2\varepsilon')/2<\pi/2$, whence
$\ell(s_{-k}-s^\ast)\le-|c||v_k|\cos\!\bigl((\omega+2\varepsilon')/2\bigr)<0$. The arcs centered at
$\arg c+2\pi n$ give the strictly positive visits symmetrically.
\end{proof}

\begin{proof}[Proof of Theorem~\ref{thm:band}]
\emph{Step 1: the forward system near the middle culture.} Along any perfect-foresight path,
\eqref{eq:forward-value}--\eqref{eq:forward-stock} read $(J_{t+1},V_{t+1})=F(J_t,V_t)$ with
\[
 F(J,V):=\Bigl(dJ+\delta H\bigl(\lambda^{-1}[V-(1-\lambda)\widehat R(J)]\bigr),\;
 \lambda^{-1}[V-(1-\lambda)\widehat R(J)]\Bigr),
\]
and $s^\ast:=(J_i,\widehat R(J_i))$ is a fixed point. Under the stated smoothness $F$ is $C^1$ near
$s^\ast$, its linearization has the characteristic polynomial \eqref{eq:forward-characteristic} at
$m_j=m_i$, and its Jacobian determinant is $d/\lambda$ everywhere; under \eqref{eq:focus-condition} the
roots are nonreal with $|\mu|^2=d/\lambda=1/\beta>1$. Apply Lemma~\ref{lem:spiral} with
$\ell(J,V):=J$ evaluated on displacements, and shrink the resulting ball $B$ so that its $J$-projection
lies in $(0,1)$.

\emph{Step 2: an exact backward prefix.} Choose $\eta>0$ small enough that $M:=J_i+\eta<J_h$,
$m:=J_i-\eta>J_\ell$, and both $(M,\widehat R(M))$ and $(m,\widehat R(m))$ lie in $B$; neither equals
$s^\ast$, their first coordinates differing from $J_i$. By Lemma~\ref{lem:spiral} the backward orbit
$(x^{(k)},v^{(k)}):=F^{-k}(M,\widehat R(M))$ stays in $B$ and has $x^{(N)}<J_i$ for some $N$. Define
\[
 x_t:=x^{(N-t)},\quad V_t:=v^{(N-t)}\ (0\le t\le N),\qquad
 x_t:=M,\quad V_t:=\widehat R(M)\ (t>N).
\]
Consecutive backward-orbit points are $F$-images of each other, so $V_t=(1-\lambda)\widehat
R(x_t)+\lambda V_{t+1}$ for every $t<N$, and the constant continuation extends this identity to all $t$;
solving it forward against its bounded tail gives $V_t=(1-\lambda)\sum_{n\ge0}\lambda^n\widehat
R(x_{t+n})$: the $V_t$ are exactly the values induced by the path $\mathbf x$ itself.

\emph{Step 3: subsolution and upward pasting.} Let $j^-:=x_0=x^{(N)}<J_i$. For $t+1\le N$,
$(\mathcal T_{j^-}\mathbf x)_{t+1}=dx_t+\delta H(V_{t+1})=x_{t+1}$, the first coordinate of $F$; for
$t+1>N$, $(\mathcal T_{j^-}\mathbf x)_{t+1}=dM+\delta H(\widehat R(M))=dM+\delta\Psi(M)>M=x_{t+1}$,
because the sign chart gives $\Psi(M)>M$ on $(J_i,J_h)$. Hence $\mathbf x\le\mathcal T_{j^-}\mathbf x$.
The operator is isotone and commutes with monotone pointwise limits (proof of
Theorem~\ref{thm:history-expectations}), so the iterates $\mathcal T^n_{j^-}\mathbf x$ increase to a
perfect-foresight path $\mathbf y\ge\mathbf x$ from $j^-$ with $\liminf_ty_t\ge M>J_i$. By
\eqref{eq:liminf-limsup} its liminf $l$ satisfies $l\ge\Psi(l)$, which with $l\ge M$ and the sign chart
forces $l\ge J_h$; its limsup satisfies $u\le\Psi(u)$, which the sign chart forbids on $(J_h,1]$, so
$u\le J_h$ and $\mathbf y\to J_h$. The comparison argument in the proof of
Theorem~\ref{thm:history-expectations} then makes the greatest path from $j^-$ converge to $J_h$ as
well: $j^-\in\mathcal E_h^{\max}$.

\emph{Step 4: the mirror witness.} Repeating Steps 2--3 from $(m,\widehat R(m))$, with a backward index
$N'$ at which $j^+:=x^{(N')}>J_i$ and the constant continuation at $m$, yields a supersolution: on the
plateau, $dm+\delta\Psi(m)<m$ because $\Psi(m)<m$ on $(J_\ell,J_i)$. Downward iteration gives a
perfect-foresight path $\mathbf z$ from $j^+$ with $\limsup_tz_t\le m<J_i$. The squeeze forces
$u\le J_\ell$ (the sign chart forbids $u\le\Psi(u)$ on $(J_\ell,J_i)$) and then $l=u=J_\ell$, since
$l\ge\Psi(l)$ fails on $(0,J_\ell)$ and at $l=0$, where $\Psi(0)>0$. So $\mathbf z\to J_\ell$, and by
comparison the least path from $j^+$ converges to $J_\ell$: $j^+\notin\mathcal E_h^{\min}$.

\emph{Step 5: conclusion.} $\mathcal E_h^{\max}$ is an upper set containing $j^-$, and
$\mathcal E_h^{\min}$ is an upper set not containing $j^+$, so
$[j^-,j^+]\subseteq\mathcal E_h^{\max}\setminus\mathcal E_h^{\min}$. The display's stronger claim also
holds: both extremal paths are nondecreasing in the inherited stock
(Theorem~\ref{thm:history-expectations}(ii)), so for every $j\in[j^-,j^+]$ the greatest path from $j$
dominates the greatest path from $j^-$ and, by the comparison argument in the proof of
Theorem~\ref{thm:history-expectations}, converges to $J_h$, while the least path from $j$ lies below the
least path from $j^+$ and converges to $J_\ell$---not merely away from $J_h$. When $\mathcal E_h^{\min}$
is nonempty, \eqref{eq:expectation-band} gives $J_-^{RE}\le j^-<J_i<j^+\le J_+^{RE}$.
\end{proof}

\begin{proof}[Proof of Theorem~\ref{thm:frontier}]
Iterating $J_{t+1}=dJ_t+\delta B_t$ yields
\eqref{eq:general-stock-bounds}. Setting $B_t\le1$ and summing the geometric series gives
\eqref{eq:physical-stock-ceiling}. Once policy is withdrawn, Theorem~\ref{thm:history-expectations} applies from
the inherited stock $J_T$. A ceiling strictly below $J_-^{RE}$ rules out possible high convergence, and one
strictly below $J_+^{RE}$ rules out high convergence under every admissible continuation. A ceiling weakly below either cutoff
only rules out strict crossing.

For either threshold $c\in(J_0,1)$, the physical ceiling lies strictly above $c$ exactly when
\[
 1-(1-J_0)d^T>c
 \quad\Longleftrightarrow\quad
 T>\frac{\log[(1-c)/(1-J_0)]}{\log d}.
\]
The smallest integer satisfying the strict inequality is the floor plus one in
\eqref{eq:Tpossible}--\eqref{eq:Trobust}.

Because the continuation-value weights are nonnegative and sum to one, $V_{t+1}\ge\underline R$ on every
bounded path. Under a uniform direct shift, therefore,
$B_t=H(V_{t+1}+\tau)\ge H(\underline R+\tau)$. The lower bound obtained from
\eqref{eq:general-stock-bounds} weakly reaches $c$ at $\tau_T^{\mathrm{env}}(c)$ and strictly exceeds it for
$\tau>\tau_T^{\mathrm{env}}(c)$. Full support implies that
a finite such shift exists whenever the physical ceiling has strict slack above $c$. It also implies
$B_t<1$ for every finite shift, while arbitrarily large finite shifts make $B_t$ arbitrarily close to one.
At any shorter duration, therefore, a finite direct shift places $J_T$ strictly below the physical ceiling,
which is weakly below the relevant cutoff. Under the policy schedule, the isotone policy operator has a
greatest and a least fixed point. For $c=J_-^{RE}$, the greatest fixed point's post-withdrawal tail equals
the greatest no-policy path from its terminal stock $J_T$:
otherwise, pasting that greater tail to the policy prefix produces a subsolution strictly above the greatest
policy fixed point, a contradiction. Hence $J_T>J_-^{RE}$ yields at least one high-converging policy path;
when $J_T>J_+^{RE}$, every post-withdrawal tail converges high, and by the uniform bound
$B_t\ge H(\underline R+\tau)$ the terminal stock of \emph{every} policy fixed point then exceeds
$J_+^{RE}$, so the guarantee holds across all admissible continuations. For the necessity of
$T_{\mathrm{robust}}$, the mirror-image pasting applies: the least policy fixed point's post-withdrawal tail
equals the \emph{least} no-policy path from its terminal stock---otherwise, pasting that smaller tail to the
policy prefix produces a supersolution strictly below the least policy fixed point, and downward iteration
of the isotone operator from it reaches a fixed point strictly below the least one, a contradiction. At any
$T<T_{\mathrm{robust}}$ and any finite shift, the least policy fixed point's terminal stock lies strictly
below $J_+^{RE}$, so its tail---the least no-policy path from that stock---does not converge to $J_h$ by the
definition of $J_+^{RE}$ as $\inf\mathcal E_h^{\min}$: the robust guarantee fails at every shorter duration.
This proves feasibility and sharpness
within the stated policy class. Online Appendix OA.10 supplies the full continuation-value envelope and the
detailed pasting construction.
\end{proof}

\section{Proof of the welfare theorem}\label{app:welfare-proof}

\begin{proof}[Proof of Theorem~\ref{thm:central-welfare}]
(i) Differentiating \eqref{eq:welfare}, the planner's enactment first-order condition is
\[
 G^{-1}(A/J)=b_N+\Delta(\sigma(A))+A\sigma'(A)\Delta'(\sigma(A)),
\]
because
\[
 \partial_A[A\Delta(\sigma(A))]=\Delta(\sigma(A))+A\sigma'(A)\Delta'(\sigma(A)),
 \qquad \partial_A[pJK(A/J)]=pG^{-1}(A/J).
\]
The decentralized cutoff is $z(A)=b_N-\psi s_a(A)$, so the planner's value differs from the private cutoff by
$\Delta(\sigma(A))+A\sigma'(A)\Delta'(\sigma(A))+\psi s_a(A)$; this differs from the
fixed-competence wedge $\Delta(\sigma)+\psi s_a(A)$ by the learning term $A\sigma'(A)\Delta'(\sigma(A))$,
which is nonnegative and, since $\Delta'=V_N+L>0$, strictly positive wherever $A>0$ and $\sigma'(A)>0$,
proving (i). For part~(iii), consider the competent benchmark. Since
$z(A)\le b_N<b_N+V_N+D$, $w(A_j)<w^\ast$ and $\Omega_P(A_j)<\Omega_S$, so
\[
 R(J_j,A_j)<\msv(p),\qquad
 J_j=H(R)<J^\ast,\qquad
 A_j=J_jw(A_j)<J_jw^\ast<J^\ast w^\ast=A^\ast.
\]
For the same-technology comparison in part~(ii), define
\[
 c(A):=b_N+\Delta(\sigma(A))+A\sigma'(A)\Delta'(\sigma(A)).
\]
At an interior optimum of the planner with
endogenous competence, the first-order conditions give
\[
 G^{-1}(A^{\ast\ast}/J^{\ast\ast})=c(A^{\ast\ast}),\qquad
 H^{-1}(J^{\ast\ast})=p\Omega(c(A^{\ast\ast}))-(1-p)\gamma.
\]
The stated nonnegativity condition makes $c(A^{\ast\ast})\ge b_N$. Since every equilibrium has
$z(A_j)\le b_N$ and
\[
 H^{-1}(J_j)=p\Omega(z(A_j))-(1-p)[\gamma+\phi s_f(J_j)],
\]
an effectively sanctioned
margin makes the comparison strict: $\psi s_a(A_j)>0$ gives $z(A_j)<b_N\le c(A^{\ast\ast})$ and hence
$\Omega(z(A_j))<\Omega(c(A^{\ast\ast}))$, because $\Omega'=G>0$ makes $\Omega$ strictly increasing, while
$\phi s_f(J_j)>0$ makes the exposure term strictly negative with
$\Omega(z(A_j))\le\Omega(c(A^{\ast\ast}))$. Either way $H^{-1}(J_j)<H^{-1}(J^{\ast\ast})$, so
$J_j<J^{\ast\ast}$. It then gives
$A_j=J_jG(z(A_j))<J^{\ast\ast}G(c(A^{\ast\ast}))=A^{\ast\ast}$. Online Appendix OA.1 gives the planner
derivations and boundary conditions.
\end{proof}

\section*{Statements and Declarations}

\noindent\textit{Funding and competing interests.} The author declares no competing interests; no external
funding was received for this work.

\smallskip
\noindent\textit{Data and code availability.} No empirical data are used. A complete replication package
containing \texttt{replicate\_paperF3.py}, the companion verification and figure scripts, versioned software
requirements, expected outputs, and documentation accompanies this submission as supplementary research material.

\smallskip
\noindent\textit{Use of generative AI.} During the preparation of this
work, the author used Claude and OpenAI Codex as an AI-assisted tool for language and structural revision, consistency
checking, literature organization, and computational verification. The author reviewed and edited all output, verified the
mathematical arguments, references, and numerical results, and takes full responsibility for the content of
the manuscript.

\singlespacing

\clearpage
\normalsize
\onehalfspacing
\setcounter{section}{0}
\setcounter{equation}{0}
\setcounter{figure}{0}
\setcounter{table}{0}
\renewcommand{\thesection}{OA.\arabic{section}}
\renewcommand{\theequation}{OA.\arabic{equation}}
\renewcommand{\thefigure}{OA.\arabic{figure}}
\renewcommand{\thetable}{OA.\arabic{table}}

\begin{center}
{\Large\bfseries Online Appendix}\\[5pt]
{\large Conformity Traps and the Formation of Independent Judgment}\\[3pt]
{\normalsize Hector Galindo-Silva}
\end{center}
\medskip

\noindent This appendix is part of the same document: sections are numbered
OA.1--OA.13, matching every ``Online Appendix OA.$n$'' pointer in the main text.
\medskip

This appendix is organized in five blocks. Sections~\ref{oa:accuracy}--\ref{oa:covert} develop the
accuracy technology and the welfare and output accounting; Sections~\ref{oa:robustness}--\ref{oa:scope}
collect robustness panels, the responsive-opposition microfoundation, and the scope conditions;
Sections~\ref{oa:forward-local}--\ref{oa:patience} contain the local classification of forward-looking
cultures and the policy-duration material; Section~\ref{oa:certification} records the certification
details; and Section~\ref{oa:identification} states the identification framework, with the omitted proofs
collected last in Section~\ref{oa:omitted}.
Table~\ref{tab:oa-notation} collects the notation used most often in the main paper and here.

\begin{table}[ht]
\centering
\caption{Core notation.}\label{tab:oa-notation}
\footnotesize
\setlength{\tabcolsep}{5pt}
\begin{tabular}{L{0.24\textwidth}L{0.66\textwidth}}
\toprule
Object & Meaning \\
\midrule
$J, A$ & Masses of formed and enacted judgment. \\
$k,H;\ \ell,G$ & Formation net cost and its c.d.f.; resolve cost and its c.d.f. \\
$p,\gamma$ & Long-run frequency of novel episodes; expected diagnostic cost, equivalently the saving from immediate routine reliance, on familiar episodes. \\
$q_m,\nu_m,\alpha_m$ & Posterior probability of opposition after an event at margin $m$; background and authentic event-rate parameters. \\
$s_m,\kappa_m,\bar q_m$ & Audience sanction; its sensitivity and tolerance threshold. \\
$\phi_0,v_f,\phi;\ \psi$ & Raw formation visibility, prompt vindication, unresolved formation exposure $\phi=(1-v_f)\phi_0$; enactment-stigma weight. \\
$b_N,z(A),w(A)$ & Private value of faithful enactment; enactment cutoff; enactment rate among formed agents. \\
$\sigma(A),\Delta(\sigma)$ & Competence generated by enacted judgment; output value of departure relative to routine use in a novel episode. \\
$C,K,\Omega$ & Aggregate formation cost, aggregate resolve cost, and private option-value functions. \\
$\delta,\beta$ & Revision opportunity rate and discount factor in the forward-looking dynamics. \\
$\rho$ & Fraction of output consequences internalized by a private enactor. \\
$V_F,V_N,D,L$ & Familiar-output level, correct-departure value, routine loss under novelty, and incorrect-departure loss. \\
\bottomrule
\end{tabular}
\end{table}

The welfare sections use
$C(J)=\int_0^JH^{-1}(u)\,du$, $K(q)=\int_0^qG^{-1}(u)\,du$, and
$\Omega(x)=\E[(x-\ell)_+]$, with $\E|k|<\infty$ and $\E|\ell|<\infty$ assumed so that $C$ and $K$ are
finite and continuous on $[0,1]$; the left tail of $\ell$ is already controlled by the finiteness of
$\Omega$, while the right tail is what $K(1)$ requires. Standing losses
are treated as transfers unless
Sections~\ref{oa:real-accounting}--\ref{oa:screening} explicitly restore their real matching consequences.

\section{Accuracy and public experimentation}\label{oa:accuracy}

The learning schedule in the main paper is the reduced form of the following environment of public
experimentation. Novelty arrives in \emph{regimes}: an unknown parameter $\xi$ characterizes how novel tasks
currently differ from the routine, and regimes turn over at a fixed hazard, at which point $\xi$ is redrawn and
learning restarts at common Gaussian prior precision $\tau_0$. The routine's payoff does not vary with $\xi$, so
choosing it generates no data. A departure is an experiment. Within a regime, public experiments accumulate as
a finite point process whose expected count in a stationary culture with enacted prevalence $A$ is $\lambda A$
(absorbing arrival intensity and expected regime age into $\lambda$), and each produces a public signal
\[
 x_i=\xi+\epsilon+\epsilon_i,
 \qquad
 \epsilon\sim N(0,\tau_c^{-1}),
 \quad
 \epsilon_i\sim N(0,\tau_e^{-1}).
\]
The common noise $\epsilon$ is drawn once per regime and the idiosyncratic noises are independent across
experiments. Averaging $N$ signals eliminates idiosyncratic but not common noise. The reduced form evaluates
precision at the mean count---a mean-count approximation, not the exact expectation of nonlinear accuracy over
random $N$---and, absorbing $\lambda$ into $\tau_e$, gives shared posterior precision
\[
 \tau(A)=\tau_0+\frac{\tau_c\tau_eA}{\tau_eA+\tau_c},
\]
which is increasing, strictly concave, and bounded above by $\tau_0+\tau_c$. Common noise prevents exact learning
within a regime, while regime turnover prevents accumulation across regimes. The finite point process is also
what keeps idiosyncratic noise relevant; exact averaging over a positive mass in a non-atomic continuum would
otherwise eliminate it.

Let $\Sigma(\tau)$ be the probability that a formed agent who consults the shared model identifies the correct
novel action, and suppose it is nondecreasing in precision. Then $\sigma(A)=\Sigma(\tau(A))$ is nondecreasing
and bounded; it is concave whenever $\Sigma$ is also concave (including an affine rule with nonnegative slope).
In the affine case,
\[
 \sigma(A)=\underline\sigma+
 (\overline\sigma_{\max}-\underline\sigma)\frac{A}{A+m},
 \qquad m=\tau_c/\tau_e,
\]
the hyperbolic schedule used in the worked example. Learning saturates faster when individual experiments are
precise relative to the common shock. Signals become part of the shared model inherited by successors, rather
than information used by the audience sanctioning the current act. Because signals are public and the private
enactor values conviction rather than the informational spillover, experimentation is under-supplied as in
strategic experimentation \citep{BoltonHarris1999,KellerRadyCripps2005}.

With competence $\sigma(A)$,
\[
 Y_N^L(A)=A\Delta(\sigma(A))-D,
 \qquad
 \Delta(\sigma)=\sigma(V_N+L)-(L-D).
\]
If the agent internalizes fraction $\rho$ of output, her enactment cutoff is
$b_N+\rho\Delta(\sigma(A))-\psi s_a(A)$, whereas the planner's cutoff is
$b_N+\Delta(\sigma(A))+A\sigma'(A)\Delta'(\sigma(A))$.

\begin{proposition}[Accuracy and reputation]\label{prop:oa-accuracy}
An enacted departure has positive incremental value iff
$\sigma(A)>\bar\sigma=(L-D)/(V_N+L)$; the novel-state level is positive iff
$A\Delta(\sigma(A))>D$. Taking competence as given, the local enactment wedge is
\[
 W_a(A;\rho)=(1-\rho)\Delta(\sigma(A))+\psi s_a(A).
\]
With practice-built competence it gains the learning term
$A\sigma'(A)\Delta'(\sigma(A))$. At an interior planner allocation satisfying
$G^{-1}(A/J)=b_N+\Delta(\sigma(A))+A\sigma'(A)\Delta'(\sigma(A))$, the formation wedge is
\[
\begin{aligned}
 W_f(J,A;\rho)&=
 p\bigl\{\Omega\bigl(b_N+\Delta(\sigma(A))+A\sigma'(A)\Delta'(\sigma(A))\bigr)\\
 &\qquad-\Omega\bigl(b_N+\rho\Delta(\sigma(A))-\psi s_a(A)\bigr)\bigr\}
 +(1-p)\phi s_f(J).
\end{aligned}
\]
At an arbitrary $(J,A)$, the first $\Omega$ argument is instead $G^{-1}(A/J)$. Holding competence fixed removes
the term $A\sigma'(A)\Delta'(\sigma(A))$ from the planner cutoff. Pressure is
corrective exactly where the relevant wedge is negative, which requires low internalization and an
output loss large enough to overcome stigma and, when competence is endogenous, the learning externality.
\end{proposition}

In the worked example, $L=0.80$ and $\bar\sigma=0.573$, while
$\sigma(A_\ell)=\sigmalow<\bar\sigma<\sigmahigh=\sigma(A_h)$. Nevertheless, at $\rho=0$ the enactment wedge remains
positive in both cultures ($\wedgelow$ in the trap and $\wedgehigh$ in the questioning culture), because
standing-loss distortions are concentrated exactly where competence is lowest. Pressure is corrective there only under
the screening qualification in Section~\ref{oa:screening}.

\subsection*{The endogenous-competence planner}

For completeness, consider the stationary planner objective from the main paper,
\[
 \mathcal W(J,A;p)=(1-p)[V_F+\gamma(1-J)]
 +p[A b_N+A\Delta(\sigma(A))-D]-C(J)-pJ K(A/J)
\]
on $0\le A\le J\le1$,
with the convention $JK(A/J):=0$ at $J=0$ and maintaining the baseline private-payoff specification
($\rho=0$) throughout, as in the main text.
Define the learning-adjusted social cutoff
\[
 c(A):=b_N+\Delta(\sigma(A))+A\sigma'(A)(V_N+L).
\]

\begin{proposition}[Endogenous-competence planner]\label{prop:oa-planner-endog}
Let $\sigma$ be increasing, concave, and $C^1$. The planner problem attains an interior maximum, and every
interior maximum $(J^{\ast\ast},A^{\ast\ast})$ satisfies
\[
 G^{-1}(A^{\ast\ast}/J^{\ast\ast})=c(A^{\ast\ast}),\qquad
 H^{-1}(J^{\ast\ast})=p\Omega(c(A^{\ast\ast}))-(1-p)\gamma.
\]
If
$\Delta(\sigma(A^{\ast\ast}))+A^{\ast\ast}\sigma'(A^{\ast\ast})(V_N+L)\ge0$, every equilibrium with at
least one effectively sanctioned margin ($\psi s_a(A_j)>0$ or $\phi s_f(J_j)>0$) satisfies
$J_j<J^{\ast\ast}$ and $A_j<A^{\ast\ast}$. If
$A\Delta(\sigma(A))$ is concave, the objective is jointly strictly concave and the optimum unique. In the
nonconcave hyperbolic worked example, the exhaustive computer-assisted certificate computed with Arb establishes uniqueness of the joint
first-order solution at $p=0.10$.
\end{proposition}

\begin{proof}
Continuity on the compact feasible set gives existence. Full support implies
$C'(0^+)=-\infty$ and $C'(1^-)=+\infty$, ruling out $J=0,1$. Given $J>0$, the enactment marginal tends to
$+\infty$ as $A\downarrow0$ and to $-\infty$ as $A\uparrow J$, so the solution is interior. Differentiation
gives the displayed first-order conditions, using
$qG^{-1}(q)-K(q)=\Omega(G^{-1}(q))$.

The nonnegativity condition gives $c(A^{\ast\ast})\ge b_N$. At any equilibrium,
\begin{align*}
 H^{-1}(J_j)
 &=p\Omega(z(A_j))-(1-p)[\gamma+\phi s_f(J_j)]\\
 &\le p\Omega(b_N)-(1-p)\gamma
 \le p\Omega(c(A^{\ast\ast}))-(1-p)\gamma
 =H^{-1}(J^{\ast\ast}),
\end{align*}
with strict first inequality at an effectively sanctioned margin: $\psi s_a(A_j)>0$ gives
$z(A_j)<b_N$ and hence $\Omega(z(A_j))<\Omega(b_N)$, because $\Omega'=G>0$ makes $\Omega$ strictly
increasing, while $\phi s_f(J_j)>0$ makes the exposure term strictly negative. Hence $J_j<J^{\ast\ast}$, and
$A_j=J_jG(z(A_j))<J^{\ast\ast}G(c(A^{\ast\ast}))=A^{\ast\ast}$.

If $[A\Delta(\sigma(A))]''\le0$, the benefit term is concave while $C$ and the perspective $JK(A/J)$ are
convex. On the interior the Hessian of $-C(J)-pJK(A/J)$ has determinant $C''pK''/J>0$ before adding the
nonpositive enactment curvature, so the full objective is strictly concave.
\end{proof}

Endogenous competence also sharpens the output reversal. If
$\sigma(A_\ell)<\bar\sigma<\sigma(A_h)$, departures in the trap lower novel-state output while those in the
questioning culture raise it. This does not by itself make pressure corrective: in the worked example the
enactment wedge remains positive in both cultures, $\wedgelow$ and $\wedgehigh$ at $\rho=0$, because stigma
is largest where competence is lowest. Pressure becomes corrective only if under-internalized error or real screening
value outweighs stigma and the steady-state learning externality.

\section{The transfer benchmark and full match-surplus accounting}\label{oa:real-accounting}

Ex post, a sanction is a real withdrawal of cooperation. A full welfare accounting for the real matching consequences of withdrawal therefore adds two flow
terms omitted from the transfer benchmark: destroyed match surplus for authentic agents and the partners'
matching payoff $U_m(s_m;q_m)$ per visible event. At an interior audience optimum the partners' payoff loss from
a small enactment-protection intervention is second order by the envelope property. By contrast, the reduction in
match-surplus destruction borne by inframarginal authentic enactors, $pA\psi\varepsilon$, is a first-order real gain.

The additional qualification concerns the released marginal act. Under real accounting, and with fully
competent released departures ($\Delta(1)=V_N+D$), its social value is
\[
 V_N+D-\psi\left[\Lambda_as_a+\frac{c_a}{2}s_a^2\right].
\]
The local first-order enactment-protection comparison therefore survives whenever
\[
 V_N+D>
 \psi\left[\Lambda_as_a+\frac{c_a}{2}s_a^2\right];
\]
at inherited competence $\sigma$ the same comparison replaces $V_N+D$ by $\Delta(\sigma)$.
This condition relaxes as prevalence rises and $s_a$ falls, and the inframarginal term reinforces it. The
formation margin is analogous, with destruction terms scaled by the enactment rate $G(z(A))$. Thus the
transfer benchmark isolates the judgment-capacity and learning channels from the match-related welfare effects of
withdrawal. Full match-surplus accounting adds an interpretable protection term of the same family as screening value.

The planner--equilibrium ordering is moreover \emph{globally} robust to adding sufficiently small
match-surplus terms.

\begin{proposition}[Local robustness to small match-surplus perturbations]\label{prop:realsurplus}
Let $\Upsilon\in C^0$ on the compact feasible set $\{(J,A):0\le A\le J\le1\}$ collect the match-surplus
flows omitted from the transfer benchmark---the destroyed match surplus of sanctioned authentic agents and
the partners' matching payoff $U_m(s_m;q_m)$ per visible event---and let
$\mathcal W_\varrho:=\mathcal W+\varrho\Upsilon$ for $\varrho\ge0$. Suppose the hypotheses of
Theorem~\ref{thm:central-welfare}(ii) hold, $(J^{\ast\ast},A^{\ast\ast})$ is the unique maximizer of
$\mathcal W$, and the reduced equilibrium equation has finitely many fixed points (as under the
threshold-shape condition of Proposition~\ref{prop:atmost}). Then there exists $\bar\varrho>0$ such that, for
every $\varrho\in[0,\bar\varrho)$, every maximizer $(J^{\ast\ast}_\varrho,A^{\ast\ast}_\varrho)$ of
$\mathcal W_\varrho$ satisfies $J_j<J^{\ast\ast}_\varrho$ and $A_j<A^{\ast\ast}_\varrho$ at every equilibrium
with at least one effectively sanctioned margin ($\psi s_a(A_j)>0$ or $\phi s_f(J_j)>0$).
\end{proposition}

\paragraph{Scale of the correction.}
At the ex-post optimal withdrawal the audience optimum of \eqref{eq:audience-payoff} has value
$U_m=\tfrac{c_m}{2}W_m(q)^2$ with $W_m(q)\le\pi_m/c_m$, so $U_m\le\pi_m^2/(2c_m)$. This disciplines the
omitted flows per event; it does not by itself place the economy with full match-surplus accounting
($\varrho=1$) inside $[0,\bar\varrho)$. The proof of the proposition is in Appendix~\ref{oa:omitted}.

\section{Screening and the value of pressure}\label{oa:screening}

Suppose an undeterred oppositional act at margin $m$ imposes real harm $D_m^{\mathrm o}$ and the mass of such acts
responds smoothly to stigma, with $\nu_m'(s_m)<0$. Let $\mathcal M_m(s_m)$ denote the first-order benchmark
judgment surplus destroyed by a marginal increase in pressure, including the induced loss of formation or
enactment. The local welfare derivative is
\[
 \frac{\partial\mathcal W}{\partial s_m}
 =D_m^{\mathrm o}[-\nu_m'(s_m)]-\mathcal M_m(s_m).
\]
Pressure is locally valuable exactly when the real screening value of deterring oppositional acts exceeds the
judgment surplus it destroys. This is an accounting decomposition rather than an independent theorem, because
$\mathcal M_m$ is defined as the surplus destroyed. Even then, the audience's ex-post sanction need not coincide
with the social optimum: the audience weighs private standing after observing an act and omits the formation
externality.

\subsection*{The complete conditional map}

Let $w^\ast=G(b_N+V_N+D)$ and $\Omega_S=\E[(b_N+V_N+D-\ell)_+]$. In the competent transfer benchmark,
\[
 \msv(p)-R(J,A)
 =p[\Omega_S-\Omega_P(A)]+(1-p)\phi s_f(J)\ge0.
\]
Thus every strict-pressure equilibrium has $J_j<J^\ast$ and $A_j<J^\ast w^\ast$: enactment pressure lowers
both action and the option value of formation, while formation pressure adds a separate wedge. The whole gap
should not be attributed to reputation. With sanctions removed,
$J^0=H(p\Omega(b_N)-(1-p)\gamma)$ still falls short of the competent planner because private conviction omits
$V_N+D$, and
\[
 J^\ast-J_j=(J^\ast-J^0)+(J^0-J_j).
\]
In the worked example this decomposition is $0.043+0.112$ in the questioning culture and $0.043+0.835$ in the
trap.

\begin{table}[ht]\centering
\caption{The conditional welfare map.}
\label{tab:oa-welfare}
\footnotesize\setlength{\tabcolsep}{4pt}
\begin{tabular}{L{0.20\textwidth}L{0.34\textwidth}L{0.36\textwidth}}
\toprule
regime & local welfare object & verdict on pressure\\
\midrule
Competent transfer benchmark
 & $\msv-R=p[\Omega_S-\Omega_P]+(1-p)\phi s_f\ge0$
 & distortionary; each strict-pressure equilibrium under-provides formation and enactment\\[2pt]
Competence built by practice
 & $\Delta(\sigma)+A\sigma'(A)\Delta'(\sigma)+\psi s_a(A)$
 & distortionary where positive; corrective only if under-internalized error overcomes stigma and learning\\[2pt]
Real oppositional harm
 & $D_m^{\mathrm o}[-\nu_m'(s_m)]-\mathcal M_m(s_m)$
 & valuable iff screening value exceeds the judgment surplus destroyed\\
\bottomrule
\end{tabular}
\end{table}

The audience's ex-post sanction need not coincide with the social optimum even in the third row: it weighs
private standing after observing the act and omits the formation externality. The main paper's welfare result is
therefore conditional, not an unconditional anti-pressure ranking.

\section{Capacity gap from practice-built competence}\label{oa:capacity-gap}

Retain the main paper's output notation
\[
 \Delta(\sigma)=\sigma(V_N+L)-(L-D),
 \qquad
 \bar\sigma=\frac{L-D}{V_N+L},
 \qquad
 Y_N^L(A)=A\Delta(\sigma(A))-D.
\]
The worked example uses the hyperbolic learning family
\[
 \sigma(A)=\underline\sigma+
 (\overline\sigma_{\max}-\underline\sigma)\frac{A}{A+m},
 \qquad m>0.
\]

\begin{proposition}[Prevalence and accuracy compound the capacity gap]\label{prop:oa-capacity-gap}
Let $0<A_\ell<A_h$ and let $\sigma$ be strictly increasing, with
$\sigma(A_h)>\bar\sigma$. Then
\[
 Y_N^L(A_h)-Y_N^L(A_\ell)
 =A_h\Delta(\sigma(A_h))-A_\ell\Delta(\sigma(A_\ell))>0.
\]
Holding $A_\ell$ and $A_h$ fixed, this gap increases in the learning ceiling
$\overline\sigma_{\max}$ and in steepness $1/m$ within the hyperbolic family. If
$\sigma(A_\ell)<\bar\sigma<\sigma(A_h)$, an enacted departure has negative incremental value in the trap and
positive incremental value in the high culture, so
\[
 Y_N^L(A_h)+D>0>Y_N^L(A_\ell)+D.
\]
The high-culture output level $Y_N^L(A_h)$ itself is positive if and only if
$A_h\Delta(\sigma(A_h))>D$.
\end{proposition}

\begin{proof}
Let $F(A)=A\Delta(\sigma(A))$. If $\sigma(A_\ell)<\bar\sigma$, then
$F(A_\ell)<0<F(A_h)$. If instead $\sigma(A_\ell)\ge\bar\sigma$, then
$\Delta(\sigma(A))\ge0$ for every $A\in[A_\ell,A_h]$, and
\[
 F'(A)=\Delta(\sigma(A))+A\sigma'(A)(V_N+L)>0.
\]
Thus $F(A_h)>F(A_\ell)$ in either case, which proves the level-gap claim because the common $-D$ cancels.

For a learning parameter $\eta$, the direct partial effect on the gap is
\[
 (V_N+L)\left[
 A_h\,\partial_\eta\sigma(A_h)-A_\ell\,\partial_\eta\sigma(A_\ell)
 \right].
\]
In the hyperbolic family,
\[
 A\frac{\partial\sigma(A)}{\partial\overline\sigma_{\max}}
 =\frac{A^2}{A+m},
 \qquad
 A\frac{\partial\sigma(A)}{\partial(1/m)}
 =(\overline\sigma_{\max}-\underline\sigma)
 \frac{A^2m^2}{(A+m)^2},
\]
and both expressions are strictly increasing in $A>0$. This proves the two partial comparative statics.
Finally, $\Delta(\sigma)$ has the sign of $\sigma-\bar\sigma$, so the threshold inequalities give the signs of
$Y_N^L(A)+D=A\Delta(\sigma(A))$; subtracting $D$ gives the last level condition.
\end{proof}

At the worked high culture, $A_h\Delta(\sigma(A_h))=\AhDelta$ and therefore
$Y_N^L(A_h)=\AhDelta-D=\YNLhigh$. These learning-built values are distinct from the competent output
$Y_N^C(A_h)=\YNChigh$ reported in the main paper's culture table.

\section{Covert novelty as an accounting perturbation}\label{oa:covert}

Suppose a fraction $\rho_c\in[0,1)$ of the episodes that the cue labels familiar are in fact novelties that
only a formed agent's diagnostic detects, and that detection yields a net output gain $\Delta_c>0$ relative to
applying the routine but creates no additional registered enactment event and leaves decision and posterior
technologies unchanged. Because the gain is constant in the aggregates $(J,A)$, any privately internalized
part of it is absorbed into the location of the net formation cost $k$, so the equilibrium fixed-point system
of the main text is unchanged.

Per calm-labeled episode, the output of a formed agent differs from that of a routine-reliant agent by
$-\gamma$ on the fraction $1-\rho_c$ of truly familiar episodes and by $+\Delta_c$ on the fraction $\rho_c$ of
covert novelties. Calm-labeled output is therefore
$V_F+\gamma+[-(1-\rho_c)\gamma+\rho_c\Delta_c]J$, which reduces to $Y_F(J)$ at $\rho_c=0$ and is strictly
decreasing in $J$ if and only if
\[
 (1-\rho_c)\gamma>\rho_c\Delta_c .
\]
The remaining steps establishing the normal-performance ranking reversal are unchanged. The reversal is
therefore confined to cases in which covert novelty is rare, or its gains small or slow to attribute.

\section{Additional robustness panels}\label{oa:robustness}

Figure~\ref{fig:oa-region} maps the two-culture region over two parameter slices. Color gives trap depth
$J_h-J_\ell$; white contours mark the multiplicity boundary and the star marks the baseline. The left panel
varies formation sensitivity and background opposition, and the right panel varies formation sensitivity and
the logit formation-cost scale.

\begin{figure}[ht]\centering
\caption{Two-dimensional robustness of the trap.}
\label{fig:oa-region}
\includegraphics[width=0.92\textwidth]{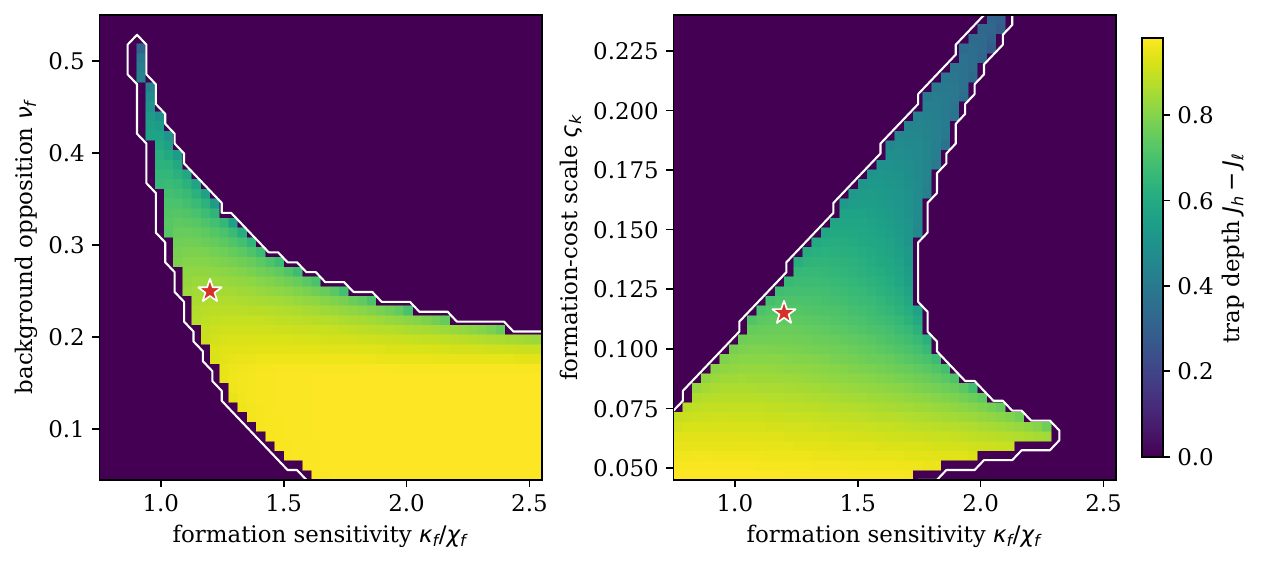}
\end{figure}

Figure~\ref{fig:oa-robustness} adds one-dimensional slices and structural variants. The first three panels vary
one primitive at a time and report trap depth
$J_h-J_\ell$; dashed lines mark the baseline. The final panel compares the baseline with smooth zero-threshold,
smooth-cap, and responsive-opposition variants.

\begin{figure}[ht]\centering
\caption{Robustness of the two-culture region.}
\label{fig:oa-robustness}
\includegraphics[width=0.92\textwidth]{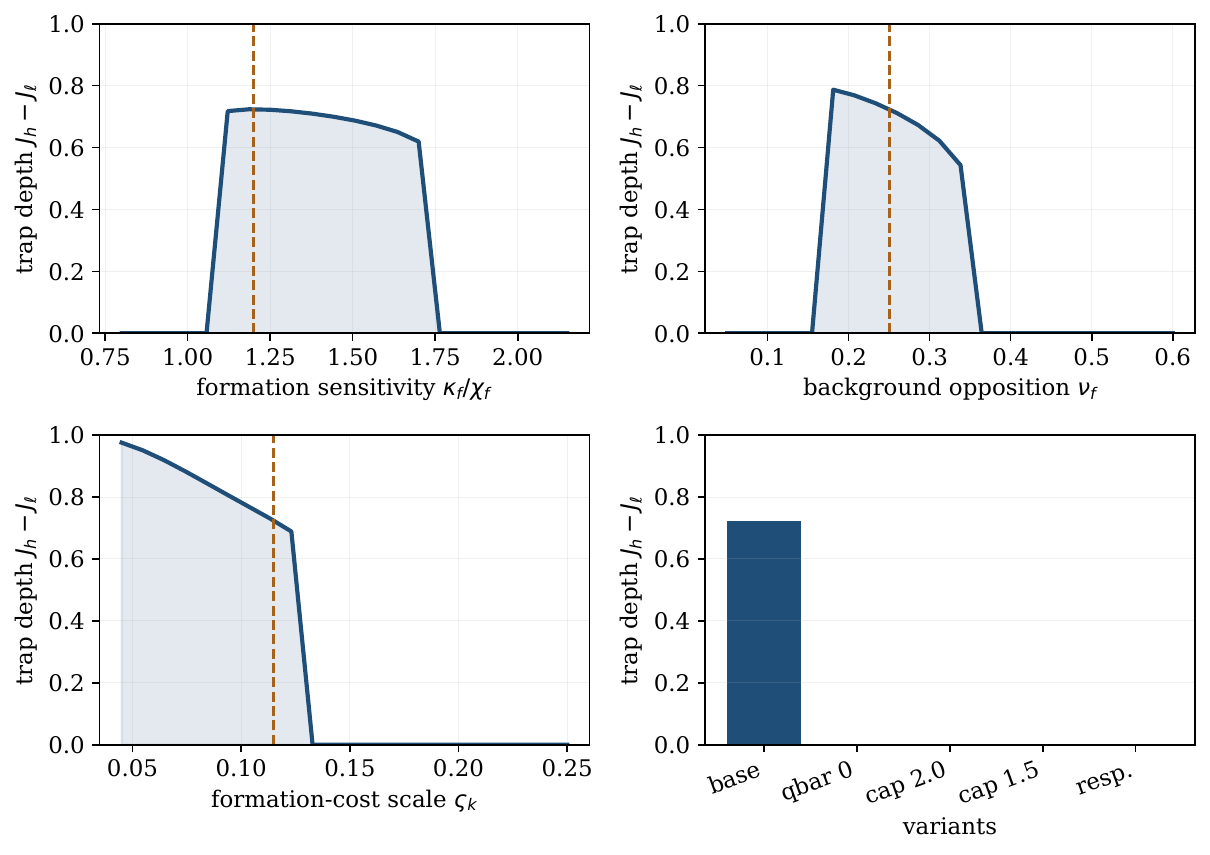}
\end{figure}

\paragraph{Joint robustness in cost dispersion and pool responsiveness.} How much economic room the mechanism
has before the scope conditions of Section~3 of the main text bind is a quantitative question, and Figure~\ref{fig:oa-jointrobust}
answers it for the two most exposed dimensions jointly: the dispersion $\varsigma_k$ of formation costs and the
responsiveness $1/d_\nu$ of the opportunistic pool, parametrized exponentially as
$\nu_m(s)=\bar\nu_m e^{-s/d_\nu}$ at both margins. Three features stand out. From the benchmark, coexistence
survives raising cost dispersion by about $13\%$ (to $\varsigma_k\approx0.130$) at a fixed pool; and at the
benchmark dispersion it survives pool responsiveness up to $d_\nu^\dagger=3.4524$, the threshold located by the
one-dimensional bisection of Section~\ref{oa:responsive}---the point at which trap-level stigma deters roughly $30\%$ of
opportunists---with the grid frontier of the figure approximating that boundary from within. And the two
margins substitute: with costs concentrated below $\varsigma_k\approx0.10$, coexistence tolerates every
responsiveness tested. The boundary is a trade-off frontier rather than a corner. This is a parametric
sweep---high-precision numerical rather than a rigorous interval enclosure; only the benchmark root counts and
folds are certified.

\begin{figure}[ht]\centering
\caption{The joint robustness region of coexistence.}
\label{fig:oa-jointrobust}
\includegraphics[width=0.62\textwidth]{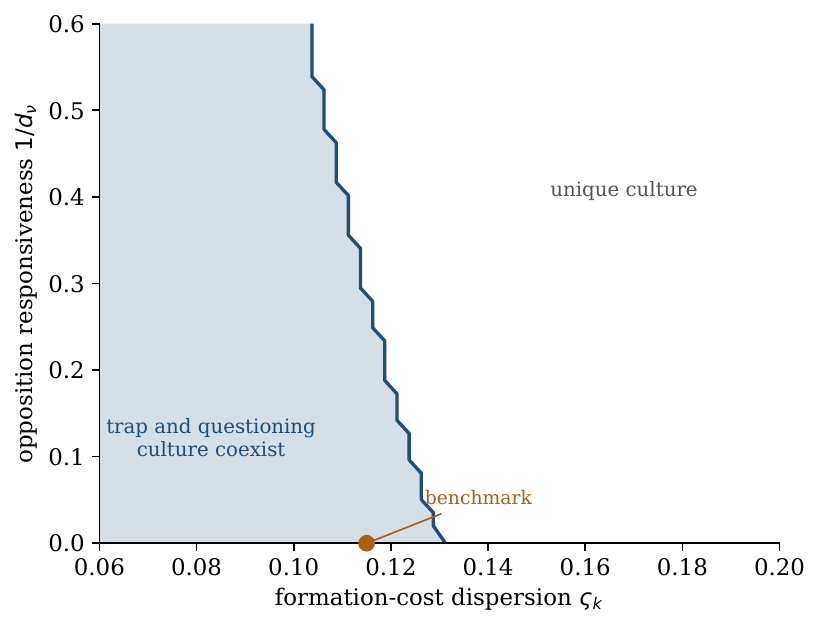}\\[4pt]
\begin{minipage}{0.88\textwidth}
\footnotesize\emph{Notes:} Shaded: pairs of formation-cost dispersion $\varsigma_k$ and opposition
responsiveness $1/d_\nu$ at which the responsive-opposition model retains three equilibria, holding every
other benchmark primitive fixed; $1/d_\nu=0$ is the fixed-pool baseline. The dot marks the benchmark
$(\varsigma_k,1/d_\nu)=(0.115,0)$. Root counts on a $57\times40$ grid with a composite $1{,}300$-point
$J$-grid; the displayed boundary anchors are unchanged under two- and four-fold denser root grids.
Numerical rather than certified; in this figure, only the benchmark root counts and folds are certified.
\end{minipage}
\end{figure}

\paragraph{A competence-aware audience.} The matching payoff of the main text prices the future
relationship with the actor: $\Lambda_a$, the loss from excluding an authentic member, is a primitive
constant. A natural objection is that the audience should also price what
Sections~\ref{subsec:sequencing} and~\ref{sec:welfare} of the main text attach to enacted judgment: if
prevalent enactment sustains competence $\sigma(A)$, excluding an authentic dissenter in the trap costs
less than excluding one in the questioning culture. The next proposition makes the audience
competence-aware and shows that the objection strengthens the mechanism. Let
\[
 \Lambda_a(\sigma(A);\varepsilon):=
 \Lambda_a^0\Bigl[1-\varepsilon+\varepsilon\,\frac{\sigma(A)}{\sigma_{\mathrm{ref}}}\Bigr],
 \qquad \varepsilon\in[0,1],
\]
with $\sigma$ increasing and $\sigma_{\mathrm{ref}}>0$ a fixed anchor; $\pi_a$ and $c_a$ are unchanged,
and $\varepsilon=0$ is the baseline.

\begin{proposition}[Competence-aware audience]\label{prop:competence-audience}
Suppose $\sigma$ is increasing and $C^1$.
\begin{enumerate}
\item[(i)] The decentralized ex-post sanction of the main text becomes
\[
 s_a^\varepsilon(A)=
 \frac{\pi_a q_a(A)-\bigl(1-q_a(A)\bigr)\Lambda_a(\sigma(A);\varepsilon)}{c_a}
\]
on its nonnegative branch, and wherever that branch is positive,
\[
 \frac{ds_a^\varepsilon}{dA}
 =\frac{\bigl(\pi_a+\Lambda_a(\sigma(A);\varepsilon)\bigr)q_a'(A)
 -\bigl(1-q_a(A)\bigr)\,\varepsilon\,\Lambda_a^0\,\sigma'(A)/\sigma_{\mathrm{ref}}}{c_a}<0.
\]
Both terms are strictly negative, so the sanction remains strictly decreasing in prevalence and
Assumption~\ref{ass:socialmeaning} continues to hold; competence awareness adds the second term---lower
enacted prevalence now also lowers the audience's stake in retaining authentic dissenters---and therefore
strengthens the prevalence feedback rather than weakening it.
\item[(ii)] Suppose in addition that $\sigma$ is $C^2$, that at $\varepsilon=0$ the enactment subgame
satisfies the regularity condition \eqref{eq:enactment-regularity} strictly, and that the reduced map has
exactly three transverse fixed points with the stable--unstable--stable pattern. Then there is
$\bar\varepsilon>0$ such that for every $\varepsilon\in[0,\bar\varepsilon)$ the competence-aware model
has a unique regular enactment solution and exactly three transverse fixed points with the same pattern,
and every nondegenerate fold continues.
\end{enumerate}
\end{proposition}

\begin{proof}
Part (i) is the displayed derivative: $q_a'<0$, $\Lambda_a\ge0$, and $\sigma'>0$ make both numerator terms
strictly negative, and Assumption~\ref{ass:socialmeaning} asks only continuity, weak monotonicity, and the
strict signs on the sanctioned interior branch. For (ii), the perturbation is linear in $\varepsilon$ with
a fixed $C^2$ coefficient:
\[
 s_a^\varepsilon=s_a^0+\varepsilon\,\widetilde s_a,
 \qquad
 \widetilde s_a(A)=-\bigl(1-q_a(A)\bigr)\,\Lambda_a^0
 \Bigl[\frac{\sigma(A)}{\sigma_{\mathrm{ref}}}-1\Bigr]\Big/c_a\;\in\;C^2[0,1].
\]
The parameter-dependent implicit-function theorem under the strict regularity condition gives
$a_\varepsilon\to a_0$ in $C^2[0,1]$ as $\varepsilon\downarrow0$, exactly as in the proof of
Corollary~\ref{cor:logit-feedback}, so the induced reduced fixed-point function converges in the $C^2$
topology of Proposition~\ref{prop:openset}, which delivers the persistence of the three transverse roots,
their stability signs, and the continuation of nondegenerate folds.
\end{proof}

At the benchmark, full strength is not needed as a limit. Setting
$\sigma_{\mathrm{ref}}=\sigma(A_h)$---so that the audience of the questioning culture behaves exactly as in
the baseline---and $\varepsilon=1$, the enactment regularity constant is $\compCa$ (against $1$), the
sanction remains strictly decreasing, and the reduced map retains exactly three fixed points, at
$J=\compJlow$, $\compJmid$, and $\Jhigh$ against the certified $\Jlow$, $\Jmid$, $\Jhigh$: the questioning
culture is unchanged, the trap is slightly deeper (its enactment sanction rises from $\compSaTrapBase$ to
$\compSaTrap$), and the unstable culture rises from $\Jmid$ to $\compJmid$, so the trap's basin widens.
The effect is quantitatively modest for an instructive reason: in the trap the posterior already attributes
almost every departure to the background source ($q_a=\qalow$), so the audience expects to exclude few
authentic members, and its stake in their competence carries little weight where dissent is rarest.
Competence awareness therefore deepens the trap exactly where the objection presumed it would soften it.
The full-strength numbers are high-precision numerical checks
(\texttt{verify\_extensions\_B\_F3.py}), not interval certificates; the $\varepsilon$-sweep at $0.25$,
$0.5$, and $0.75$ preserves monotonicity, regularity, and the root count throughout. Two remarks delimit
the exercise: the sanction is evaluated at the steady-state competence of
Section~\ref{subsec:sequencing} of the main text, so it inherits the stationarity qualification stated
there; and the welfare benchmark of Section~\ref{sec:welfare} continues to treat standing losses as
transfers---competence awareness changes the audience's private stake, not that accounting convention.

\paragraph{Joint robustness in the conviction stake and formation pressure.} The worked benchmark sets the
private enactment benefit at $b_N=1.81$, roughly $1.8$ familiar episodes' output, and the main text reports
that coexistence fails below roughly $b_N\approx1.3$ at the benchmark pressure.
Figure~\ref{fig:oa-bnkappa} maps the two-culture region over the conviction stake and formation pressure
jointly. The region is a widening wedge, not a corner: at the benchmark pressure $\kappa_f/\chi_f=1.2$,
coexistence requires $b_N>\bnDagger$ and then survives to the top of the tested range, while every tested
stake down to $b_N=1.20$ supports a nonempty pressure window, whose width grows with the
stake---$\bnWindowLow$ at $b_N=1.2$ and $\bnWindowTwo$ at $b_N=2.0$, against the certified window
$(\Kminchi,\Kmaxchi)$ at the benchmark stake. Lower stakes do not extinguish the mechanism; they narrow
the band of audience sensitivities that supports it. The panel is a parametric sweep at the same precision
conventions as Figure~\ref{fig:oa-jointrobust}; at $b_N=1.81$ its bisected window reproduces the certified
folds to $10^{-4}$.

\begin{figure}[ht]\centering
\caption{Coexistence in the conviction stake and formation pressure.}
\label{fig:oa-bnkappa}
\includegraphics[width=0.62\textwidth]{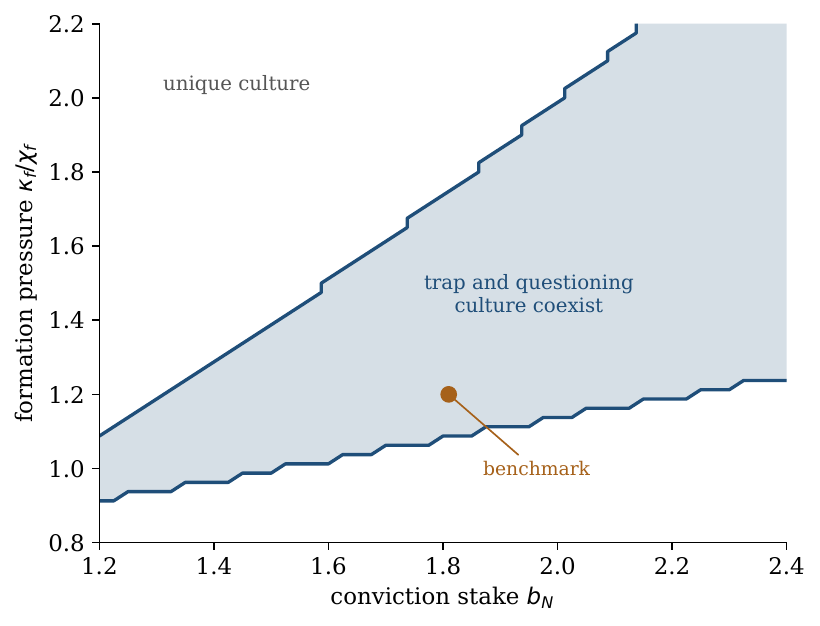}\\[4pt]
\begin{minipage}{0.88\textwidth}
\footnotesize\emph{Notes:} Shaded: pairs of conviction stake $b_N$ and formation pressure
$\kappa_f/\chi_f$ at which the reduced map retains three equilibria, holding every other benchmark
primitive fixed. The dot marks the benchmark $(b_N,\kappa_f/\chi_f)=(1.81,1.2)$. Root counts on a
$49\times57$ grid with the composite $1{,}300$-point $J$-grid of Figure~\ref{fig:oa-jointrobust};
boundaries at quoted stakes refined by bisection. Numerical rather than certified; only the benchmark
root counts and folds are certified.
\end{minipage}
\end{figure}

\section{Responsive opposition}\label{oa:responsive}

\paragraph{Background departures across all states.}
If the opportunistic pool is active in all states and the audience cannot condition on context, the enactment
posterior becomes
\[
 q_a(A)=\frac{\nu_a}{\nu_a+\alpha_a pA},
\]
the same family with the authentic arrival rate rescaled by novelty. All results carry over after replacing
$\alpha_a$ by $p\alpha_a$; novelty then lowers enactment stigma directly and reinforces the capacity comparative
static in the main paper.

If the opportunistic mass $\nu_m(s_m)$ falls with stigma, the smooth-branch sanction solves
\[
 s_m^\ast(X_m)=\kappa_m\left[
 \frac{\nu_m(s_m^\ast)}{\nu_m(s_m^\ast)+\alpha_mX_m}-\bar q_m
 \right]_\varepsilon.
\]
With $\Phi_\varepsilon(x)=(1+e^{-x/\varepsilon})^{-1}$, implicit differentiation gives
\[
 |s_m^{\ast\prime}(X_m)|
 =\frac{\kappa_m\Phi_\varepsilon\nu_m\alpha_m}
 {(\nu_m+\alpha_mX_m)^2+
   \kappa_m\Phi_\varepsilon|\nu_m'|\alpha_mX_m}
 <\frac{\kappa_m\Phi_\varepsilon\nu_m\alpha_m}
 {(\nu_m+\alpha_mX_m)^2}.
\]
Responsive opposition therefore preserves the negative reputational-sanction slope. At the current solution, the
endogenous-response term in the denominator attenuates the direct prevalence slope; this comparison does not
claim global monotonicity in response intensity.

\begin{proposition}[Responsive opposition: persistence and collapse]\label{prop:oa-responsive}
For each margin let
\[
 \nu_m^{\boldsymbol\rho}(s)=\bar\nu_m r_m(\rho_m s),
\]
where $r_m$ is $C^2$, positive on the relevant compact sanction range, $r_m(0)=1$, and $r_m'\le0$. Let
$s_m^{\boldsymbol\rho}(X_m)$ denote the solution of the responsive-sanction equation above. Suppose the
fixed-pool baseline $\boldsymbol\rho=\mathbf0$ has exactly three transverse roots with the same stability signs
throughout a compact parameter set $K$ contained in the interior of its coexistence region. Then there is
$\bar\rho>0$ such that, whenever $\max_m\rho_m<\bar\rho$, the responsive-opposition model has exactly three
transverse roots with those stability signs for every parameter in $K$. Separately, every nondegenerate fold on the
boundary of the fixed-pool coexistence region continues locally in $(\boldsymbol\rho,\theta)$, where $\theta$ is
the fold parameter. Conversely,
for any $\boldsymbol\rho$ and parameter value at which the induced reduced map satisfies
$\sup_J\Psi_{\boldsymbol\rho}'(J)<1$, the
equilibrium is unique. Thus mild responsive opposition preserves the trap, while any response that makes the
reduced map a contraction eliminates it.
\end{proposition}

\begin{proof}
Write
\[
 E_m(s,X,\boldsymbol\rho)=s-\kappa_m\left[
 \frac{\nu_m^{\boldsymbol\rho}(s)}{\nu_m^{\boldsymbol\rho}(s)+\alpha_mX}-\bar q_m
 \right]_\varepsilon.
\]
At $\boldsymbol\rho=\mathbf0$, $\nu_m^{\boldsymbol\rho}$ is independent of $s$ and $E_{m,s}=1$. The
implicit-function theorem therefore makes each $s_m^{\boldsymbol\rho}(X)$ a $C^2$ function of
$(X,\boldsymbol\rho)$ for all sufficiently small $\max_m\rho_m$, converging to the fixed-pool sanction in the
$C^2$ topology on compact $X$ intervals. The induced reduced fixed-point functions
$F_{\boldsymbol\rho}(J,\theta)=\Psi_{\boldsymbol\rho}(J;\theta)-J$ converge in the same topology to the fixed-pool function. The main
paper's robustness result then preserves the three transverse roots uniformly on $K$. At a nondegenerate boundary fold,
the system $(F_{\boldsymbol\rho},\partial_JF_{\boldsymbol\rho})=(0,0)$ has nonsingular Jacobian with respect to
$(J,\theta)$, so the implicit-function theorem continues that fold locally in $\boldsymbol\rho$. The derivative
criterion gives uniqueness whenever $\sup_J\Psi_{\boldsymbol\rho}'(J)<1$.
\end{proof}

\begin{proposition}[Expressive opposition and weak stigma elasticity]\label{prop:oa-expressive}
Let the background pool at margin $m$ consist of mass $\bar\nu_m$ of opportunists who value being seen to defy.
An opportunist draws expressive benefit $e\ge0$ from a $C^2$ distribution $\Phi_m$ with $\Phi_m(0)=0$ and
density $\varphi_m$, places weight $\omega_m\in[0,1]$ on standing, and acts iff $e\ge\omega_ms_m$. Then
\[
 \nu_m(s)=\bar\nu_m[1-\Phi_m(\omega_ms)].
\]
The response is weak enough for the regular trap to persist under either of two sufficient conditions:
\begin{enumerate}
\item[(i)] for any $\omega_m\le1$, on each compact equilibrium sanction range $\mathcal S_m$,
\[
 \max_m\sup_{s\in\mathcal S_m}
 \left\{\Phi_m(\omega_ms),\ \varphi_m(\omega_ms),\
 |\varphi_m'(\omega_ms)|\right\}\le\bar\varepsilon
\]
for a sufficiently small $\bar\varepsilon>0$; or
\item[(ii)] each $\Phi_m$ is $C^2$ on a neighborhood of the relevant compact range and
$\max_m\omega_m$ is sufficiently small.
\end{enumerate}
In either case the three roots keep their stability signs and the nondegenerate folds continue locally.
\end{proposition}

\begin{proof}
The acting mass is $\bar\nu_m[1-\Phi_m(\omega_ms)]=\bar\nu_mr_m(\omega_ms)$ with
$r_m=1-\Phi_m$, which is $C^2$, positive on the relevant range, satisfies $r_m(0)=1$, and has
$r_m'=-\varphi_m\le0$. Part (ii) follows directly from the responsive-opposition proposition. Under (i), the
pool and its first two derivatives are uniformly within $\bar\nu_m\bar\varepsilon$ of the constant pool on the
compact sanction range. The implicit-function argument above makes the induced sanction and reduced fixed-point
function converge in $C^2$ to their fixed-pool counterparts; robustness to small $C^2$ perturbations preserves the roots, stability
signs, and folds.
\end{proof}

For the numerical diagnostic, the robustness script solves the full system with the hard-interior audience rule
and $\nu_m(s)=\bar\nu_me^{-s/d_\nu}$. This exponential hard-interior exercise is separate from the smooth
no-cap family used in the local $C^2$ proposition above. Small $\max_m\rho_m$ in that local family and large
$d_\nu$ in the exponential family both describe weak responsiveness. Bisection
places the transition near $d_\nu^\dagger\approx3.45$: the sampled higher-decay-scale side has three roots and
the lower side one; at $d_\nu=5$, $J=(0.028,0.091,0.758)$, while at $d_\nu=2$ a unique questioning culture
survives. At the boundary, the trap-level formation sanction $s_f\approx1.24$ leaves
$e^{-1.24/3.45}\approx0.70$ of the pool active, so a response that deters about $30\%$ of opportunists at
trap-level stigma is enough to remove coexistence in this example. This is a sampled fragility check,
not a global comparative-statics theorem; \texttt{verify\_planner\_threshold\_F3.py} reproduces the threshold.

\section{Scope conditions: audience memory and responsive opposition}\label{oa:scope}

This appendix collects the two results that delimit the mechanism's scope: how much individual
memory the anonymous-matching assumption tolerates, and how much of the opportunistic pool may
respond to the stigma it attracts.

\subsection{Bounded audience memory}

Anonymous matching is the assumption on which
the posterior channel rests, and it deserves a quantitative rather than a rhetorical defense. Agents here are
long-lived and forward-looking, so a partner who could accumulate a personal record would eventually tell an
authentic member from an opportunist, and $\nu_m\downarrow0$ would collapse the channel---the logic of
reputational identification in \citet{ElyValimaki2003}. The question is how fast that happens. Suppose a
partner observes, in addition to the current event, the individual's registered record over the last $n$
ordinary episodes, and that an authentic member leaves a registered probe in such an episode with probability
$\phi_0$ while an opportunist, who does not wait for a diagnostic cue, does so with probability
$\chi\phi_0$, $\chi>1$. Write $k$ for the number of recorded probes in the window.

\begin{proposition}[Bounded audience memory]\label{prop:memory}
Assume $0<\phi_0<1$ and $1<\chi<1/\phi_0$, so that both per-episode recording rates are probabilities, and let
$\Lambda_n(k):=\chi^k\bigl[(1-\chi\phi_0)/(1-\phi_0)\bigr]^{n-k}$ be the likelihood ratio of the record. Then:
\begin{enumerate}
\item[(i)] the posterior after a probe is the baseline posterior \eqref{eq:posterior} with residual ambiguity
$\nu_f$ replaced by the record-dependent $\nu_f\Lambda_n(k)$, so $n=0$ is the anonymous baseline;
\item[(ii)] the expected formation stigma borne by an authentic member obeys
\[
 \E\bigl[s_f\bigr]\;\le\;\kappa_f\sqrt{\nu_f/J}\;\varkappa^{\,n},
 \qquad
 \varkappa:=\sqrt{\chi}\,\phi_0+\sqrt{(1-\phi_0)(1-\chi\phi_0)}\;<\;1;
\]
\item[(iii)] hence, \emph{for each fixed} $J>0$, expected formation stigma converges geometrically to zero
as the record lengthens.
\end{enumerate}
\end{proposition}

\begin{proof}
Part (i) is Bayes' rule applied to the product of the current event and the independent record. For (ii),
$q_f^{(n)}=\nu_f\Lambda_n/(\nu_f\Lambda_n+J)\le\sqrt{\nu_f\Lambda_n/J}$ by the arithmetic--geometric mean
inequality, and under the authentic measure
$\E[\sqrt{\Lambda_n}]=\sum_k\sqrt{\Pr_a(k)\Pr_o(k)}=\varkappa^{\,n}$, the $n$-fold Bhattacharyya coefficient of
two distinct Bernoulli measures, which is strictly below one by Cauchy--Schwarz. Since
$s_f=\kappa_f[q_f-\bar q_f]_+\le\kappa_fq_f$, the bound follows, and (iii) is immediate from
$\varkappa<1$.
\end{proof}

What the bound does and does not settle is worth stating precisely, because the two are easy to conflate. It
settles the \emph{level} of expected formation stigma at any fixed prevalence: a long individual record makes
authentic questioning self-identifying there. Three things it does not settle. First, the bound is pointwise
rather than uniform: the factor $\sqrt{\nu_f/J}$ diverges as $J\downarrow0$, so it degrades precisely in the
low-prevalence region where the trap sits, and $\sup_{J\in(0,1]}$ of the right-hand side is infinite. It
therefore does not deliver uniform convergence of the formation map, of the equilibrium correspondence, or of
the extremal cultures. Second, a bound on the level of a sanction does not control its slope in $J$, which is
what the perturbation argument of Proposition~\ref{prop:openset} needs to transport folds; persistence of
coexistence over an explicit range of record lengths would require $C^1$ control of the induced formation map
in $n$, which I do not claim. Third, even the limit is not a uniqueness result:
the discussion following Corollary~\ref{cor:anonymity} shows that at $\phi=0$ the reduced map remains
increasing through the enactment channel, so multiplicity can survive private formation in
principle---though not at the certified benchmark, where the $\phi=0$ map is unique
(Section~\ref{oa:invisible})---and uniqueness would in general additionally require
$\sup_J\Psi_0'(J)<1$. The proposition is a rate, not a survival threshold.

The mechanism therefore applies where counterparties turn over quickly relative to the rate at which a record
identifies a type: rotating committee seats, arm's-length partners, standing held at the level of a market or
profession rather than a dyad. It does not apply to a small stable group that accumulates long individual
records, and the bound in (ii) gives the rate. The same reading disciplines policy: an anonymous channel
lowers $\phi_0$ and so slows identification, while any measure that makes questioning individually traceable
speeds it up.

\subsection{Opposition that responds to its own sanction}

\begin{proposition}[Responsive opposition: a primitive survival condition]\label{prop:responsive-primitive}
Under the responsive pool $\nu_f(s)=\bar\nu_f[1-F_v(\phi s)]$ (the construction applies symmetrically at
enactment; parts (ii)--(iii) hold the enactment pool fixed), with $F_v$
continuous:
\begin{enumerate}
\item[(i)] for each $J$ the induced sanction is the unique fixed point of a continuous weakly decreasing
self-map, and $\nu_f(s)\in[\underline\nu_f,\bar\nu_f]$ at every attainable stigma;
\item[(ii)] the induced reduced map is sandwiched, $\Psi_{\bar\nu_f}(J)\le\Psi_{\mathrm{resp}}(J)\le
\Psi_{\underline\nu_f}(J)$, where $\Psi_{\nu}$ denotes the fixed-pool map at ambiguity $\nu$;
\item[(iii)] consequently, if there are $0<J_L<J_H<1$ with
\[
 \mathcal K_{\bar\nu_f}(J_H;p)\;>\;\kappa_f\;>\;\mathcal K_{\underline\nu_f}(J_L;p),
\]
the responsive-opposition model has at least three fixed points at the same pressure; if, in addition,
$\Psi_{\mathrm{resp}}$ is nondecreasing and the smallest and largest crossings are transverse
($\Psi_{\mathrm{resp}}'\neq1$ there), those two fixed points are stable under the reduced scalar adjustment.
If instead $\underline\nu_f=0$ and $\sup_J\Psi_{\mathrm{resp}}'(J)<1$, the equilibrium is unique.
\end{enumerate}
\end{proposition}

\begin{proof}
For (i), $\nu_f(\cdot)$ is weakly decreasing and $q_f$ is increasing in $\nu_f$, so
$s\mapsto\kappa_f[q_f(J;\nu_f(s))-\bar q_f]_+$ is continuous and weakly decreasing on the compact
$[0,\bar s_f]$; a decreasing continuous self-map of an interval has exactly one fixed point. The bracket on
$\nu_f$ follows from $s\in[0,\bar s_f]$. For (ii), monotonicity of $q_f$ in $\nu_f$ gives the corresponding
sanction bracket, and $\Psi$ is decreasing in $s_f$. Part (iii) applies \eqref{eq:psi-K} to each bound
separately: the upper bound at $J_L$ gives $\Psi_{\mathrm{resp}}(J_L)<J_L$, the lower bound at $J_H$ gives
$\Psi_{\mathrm{resp}}(J_H)>J_H$, and the intermediate value theorem with the boundary signs of
Theorem~\ref{thm:trap} supplies the three crossings. For the stability clause, the boundary signs place
$\Psi_{\mathrm{resp}}$ above the diagonal immediately to the left of the smallest crossing and below it
immediately to the right; at a transverse crossing of a nondecreasing map this sign pattern forces
$\Psi_{\mathrm{resp}}'<1$ there, which is stability under the reduced scalar adjustment, and the symmetric
argument applies at the largest crossing.
\end{proof}

\subsection{Invisible maintenance}\label{oa:invisible}

The formation-stigma channel prices the visible probe, so the natural deviation is to keep the practice
and hide it: substitute private verification for the public request, the transparency-substitution logic
of \citet{Prat2005}. The baseline treats diagnosis as interactive by technology---a probe requests a
justification held by others, and the knowledge needed to test a premise is dispersed
\citep{Hayek1945}---so a private substitute exists only at a cost premium. This subsection makes the
menu explicit and gives the premium at which the mechanism is safe.

Extend the environment so that each maintaining agent chooses, along with maintenance, a mode. The
\emph{visible} mode is the baseline diagnostic practice, with familiar-state expected cost $\gamma$ and
raw exposure $\phi_0$. The \emph{invisible} mode has identical diagnostic power and familiar-state
expected cost $\gamma_{\mathrm{priv}}>\gamma$, and generates no public event. The audience observes
events as before, so $X_f$ is the visible authentic mass and hidden maintainers do not enter the
posterior \eqref{eq:posterior}.

\begin{proposition}[Invisible maintenance]\label{prop:invisible}
\begin{enumerate}
\item[(i)] If
\[
 \gamma_{\mathrm{priv}}-\gamma\;>\;\phi\,\kappa_f\,(1-\bar q_f),
\]
the visible mode strictly dominates at every visible prevalence: every maintainer chooses it in any
equilibrium of the extended game, whose equilibrium set therefore coincides with the baseline's. In
particular, the certified three-culture configuration survives every such menu unchanged.
\item[(ii)] If instead $\gamma_{\mathrm{priv}}-\gamma<\phi\,s_f(J_\ell)$ at a stationary all-visible
culture with stock $J_\ell$, a maintainer there strictly prefers the invisible mode, so the all-visible
configuration unravels and the formation-stigma channel closes at that culture.
\end{enumerate}
\end{proposition}

\begin{proof}
For (i), $q_f(X_f)\le1$ gives $s_f(X_f)=\kappa_f[q_f(X_f)-\bar q_f]_+\le\kappa_f(1-\bar q_f)$, so the
visible mode's familiar-state cost obeys
$\gamma+\phi s_f(X_f)\le\gamma+\phi\kappa_f(1-\bar q_f)<\gamma_{\mathrm{priv}}$ at every
$X_f\in[0,1]$: strict dominance, mode by mode, independent of the profile. With every maintainer visible,
the extended game's equilibrium conditions reduce to those of
Definition~\ref{def:equilibrium}. Part (ii) is the one-line deviation check: at the stationary culture a
visible maintainer pays $\gamma+\phi s_f(J_\ell)$ per familiar episode against
$\gamma_{\mathrm{priv}}$ for the same diagnostic power.
\end{proof}

At the benchmark the uniform bound is $\phi\kappa_f(1-\bar q_f)=\invisUnif$, or $\invisUnifxg$ times the
diagnostic cost $\gamma$, and the trap-level exposure $\phi s_f(J_\ell)$ is $\invisTrapxg$ times
$\gamma$: the mechanism is safe against any private substitute costing roughly an order of magnitude more
than interactive diagnosis, and threatened by cheaper ones. Between the two thresholds, a conservative
diagnostic replaces the exposure by $\min\{\phi s_f(J),\bar c\}$ with
$\bar c=\gamma_{\mathrm{priv}}-\gamma$, the cost the cheaper mode would imply at unchanged visible
prevalence; the three-root configuration survives every cap down to $\invisCapxg\,\gamma$
(\texttt{verify\_extensions\_B\_F3.py}), so the numerical margin is wider than the dominance bound. The
limiting menu is instructive: free invisibility is the case $\phi=0$ of
Corollary~\ref{cor:anonymity}, and at the benchmark the reduced map is then unique, with a questioning
culture at $J=\phiZeroJ$---larger than the baseline's $\Jhigh$, since everyone forms without exposure.
The scope condition is therefore sharp in both directions. The trap requires diagnosis to be interactive,
so that hiding the practice forfeits the information others hold; where cheap private substitutes exist,
the formation-stigma channel closes---and closes toward the questioning culture. This is the
individual-choice counterpart of the institutional visibility margin of Corollary~\ref{cor:anonymity}:
institutions choose $\phi_0$; agents, offered a menu, choose whether to be seen at all.

\section{Local classification of forward-looking cultures}\label{oa:forward-local}

For reference, this section gives the complete local classification omitted from the main text. Let $H$ be the
formation-cost c.d.f., with density $h$, and let $\widehat R(J)$ be the current return to formed judgment after
solving the enactment margin. A share $\delta\in(0,1)$ can revise each period and agents discount at
$\beta\in(0,1)$. Write
\[
 d:=1-\delta,
 \qquad
 \lambda:=\beta d\in(0,d).
\]
A perfect-foresight path satisfies
\begin{align}
 V_t&=(1-\lambda)\widehat R(J_t)+\lambda V_{t+1},
 \label{eq:oa-forward-value}\\
 J_{t+1}&=dJ_t+\delta H(V_{t+1}).
 \label{eq:oa-forward-stock}
\end{align}
Define the static map $\Psi(J):=H(\widehat R(J))$.

\begin{theorem}[Forward-looking cultures: complete local classification]\label{thm:oa-forward-local}
Suppose $\widehat R$ is $C^1$ and nondecreasing, $H$ is $C^1$ with continuous density $h$ near the
stationary values $\widehat R(J_j)$, and $\Psi$ has exactly three transverse fixed points
$J_\ell<J_i<J_h$, with
\[
 m_\ell:=\Psi'(J_\ell)<1,
 \qquad m_i:=\Psi'(J_i)>1,
 \qquad m_h:=\Psi'(J_h)<1.
\]
Then:
\begin{enumerate}
\item[(i)] A stationary perfect-foresight culture is exactly a static cultural equilibrium:
$V^\ast=\widehat R(J^\ast)$ and $J^\ast=\Psi(J^\ast)$. Forward-looking formation therefore preserves the
three stationary cultures.
\item[(ii)] Define
\[
 \vartheta_j:=d+\lambda^{-1}-\frac{\delta(1-\lambda)m_j}{\lambda},
 \qquad
 m^{\mathrm{flip}}:=\frac{(1+d)(1+\lambda)}{\delta(1-\lambda)}.
\]
The local characteristic roots at $J_j$ solve
\begin{equation}\label{eq:oa-forward-characteristic}
 P_j(r)=r^2-\vartheta_jr+\frac d\lambda=0.
\end{equation}
Each outer culture is a saddle, with one root $r_j^s\in(0,1)$ and one root $r_j^u>1$. The middle culture is a
source when $1<m_i<m^{\mathrm{flip}}$: real roots then both lie outside the unit circle on the same side, while
complex roots have common modulus $\beta^{-1/2}>1$. At $m_i=m^{\mathrm{flip}}$, one root equals $-1$; above
that value, the middle culture is a saddle with a stable root in $(-1,0)$, so convergence along its stable arm is oscillatory.
\item[(iii)] Within the source region, define
\begin{equation}\label{eq:oa-node-focus}
 \widetilde Q_i(\lambda):=
 [1-\delta m_i+\lambda(d+\delta m_i)]^2-4d\lambda.
\end{equation}
The two algebraic node--focus boundaries are
\begin{equation}\label{eq:oa-lambda-node-focus}
 \lambda_i^\pm=
 \frac{[\sqrt d\pm\delta\sqrt{m_i(m_i-1)}]^2}{(d+\delta m_i)^2},
\end{equation}
with only roots in $(0,d)$ economically feasible. A negative value of \eqref{eq:oa-node-focus} gives a focus,
a positive value gives a node, and zero is a repeated-root boundary.
\item[(iv)] Along the saddle path into an outer culture, the local convergence factor is
\begin{equation}\label{eq:oa-forward-stable-root}
 r_j^s=\frac12\left[\vartheta_j-
 \sqrt{\vartheta_j^2-\frac{4d}{\lambda}}\right].
\end{equation}
At a nondegenerate fold where $m_j\uparrow1$, $r_j^s\uparrow1$, so the half-life
$\log(1/2)/\log r_j^s$ diverges.
\end{enumerate}
\end{theorem}

\begin{proof}
At a stationary path, \eqref{eq:oa-forward-value} gives $V^\ast=\widehat R(J^\ast)$, and
\eqref{eq:oa-forward-stock} then gives $J^\ast=H(\widehat R(J^\ast))$; the converse is immediate. Let
$h_j=h(\widehat R(J_j))$, so $m_j=h_j\widehat R'(J_j)$. Solving the value recursion forward gives
\[
 V_{t+1}=\lambda^{-1}V_t-(1-\lambda)\lambda^{-1}\widehat R(J_t),
 \qquad
 J_{t+1}=dJ_t+\delta H(V_{t+1}).
\]
Its Jacobian at culture $j$ is
\[
 M_j=\begin{pmatrix}
 d-\delta(1-\lambda)m_j/\lambda&\delta h_j/\lambda\\
 -(1-\lambda)\widehat R'(J_j)/\lambda&1/\lambda
 \end{pmatrix},
\]
and its characteristic polynomial is \eqref{eq:oa-forward-characteristic}. In particular,
\[
 P_j(0)=\frac d\lambda=\frac1\beta>1,
 \qquad
 P_j(1)=\frac{\delta(1-\lambda)}{\lambda}(m_j-1).
\]
For either outer culture, $P_j(1)<0$, so the two positive real roots lie on opposite sides of one. For the
middle culture, $P_i(1)>0$, $P_i(0)>1$, and
\[
 P_i(-1)=\frac{\delta(1-\lambda)}{\lambda}(m^{\mathrm{flip}}-m_i).
\]
When $1<m_i<m^{\mathrm{flip}}$, the signs at $-1$, $0$, and $1$, together with the product of the roots
$d/\lambda=1/\beta>1$, put every real pair outside the unit circle; complex roots have modulus
$\sqrt{d/\lambda}=\beta^{-1/2}$. Equality puts one root at $-1$. Above the $-1$-root threshold, a root lies in
$(-1,0)$ and the other remains outside the unit circle. Multiplying the discriminant by $\lambda^2$ gives
\eqref{eq:oa-node-focus}; solving its two zeros as a quadratic in $\sqrt\lambda$ gives
\eqref{eq:oa-lambda-node-focus}. Finally, the smaller outer root is \eqref{eq:oa-forward-stable-root}. As
$m_j\uparrow1$, $P_j(1)\uparrow0$, that root tends to one, and the other tends to $d/\lambda=1/\beta>1$.
The half-life therefore diverges.
\end{proof}

At the exact-decimal baseline, $\delta=\deltaval$ and $m_i=1.4993045\ldots$, certified by interval
arithmetic with enclosure radius below $5\times10^{-10}$, and $m^{\mathrm{flip}}=761$ exactly. The
discriminant $\vartheta_i^2-4d/\lambda$ is certified strictly negative under outward rounding, with
enclosure $-0.0082948419\pm4.1\times10^{-11}$; equivalently, the scaled discriminant
\eqref{eq:oa-node-focus} is
$\widetilde Q_i=\lambda^2(\vartheta_i^2-4d/\lambda)=-0.0067562005\ldots$
(script \texttt{certify\_forward\_dynamics\_F3.py}). With $\beta=0.95$, the middle roots are
$1.0249673\pm0.0455380i$ (common modulus $\beta^{-1/2}$, i.e.\ $|\cdot|^2=1.0526316$), so the middle
culture is an unstable focus (and hence a source). This classification is local; the global
convergence regions in the main paper come from the constructive monotone-path argument, not from the
linearization. The focus case, however, does carry global content: by Theorem~\ref{thm:band} of the main
text, nonreal middle-culture roots imply that the expectation-dependent band is strictly open around
$J_i$; the certified enclosures above make that hypothesis itself a certificate rather than a
floating-point reading.

\section{Direct formation support and patience comparison}\label{oa:patience}

To make the numerical statement precise, for inherited stock $j$ define the monotone path operator
\begin{align*}
 (\mathcal T_j\mathbf x)_0&=j,\\
 (\mathcal T_j\mathbf x)_{t+1}
 &=dx_t+\delta H\!\left((1-\lambda)\sum_{n=0}^\infty
 \lambda^n\widehat R(x_{t+1+n})\right).
\end{align*}
Iteration from $(j,0,0,\ldots)$ and $(j,1,1,\ldots)$ yields the least and greatest perfect-foresight stock
paths. Assume both high-convergence sets are nonempty. Let $J_-^{RE}$ be the infimum inherited stock from which
the greatest path converges to $J_h$, and let $J_+^{RE}$ be the infimum inherited stock from which the least
path converges to $J_h$. For initial stock $J_0$ satisfying
$J_0<J_-^{RE}\le J_+^{RE}<1$, and for $d=1-\delta$, define the physical duration thresholds
\[
 T_{\mathrm{possible}}
 =\left\lfloor\frac{\log[(1-J_-^{RE})/(1-J_0)]}{\log d}\right\rfloor+1,
 \qquad
 T_{\mathrm{robust}}
 =\left\lfloor\frac{\log[(1-J_+^{RE})/(1-J_0)]}{\log d}\right\rfloor+1.
\]

\subsection*{Sufficient direct formation shifts}

The instrument in this subsection shifts only the formation cutoff of a revising cohort. It does not change
$\widehat R$, the post-withdrawal primitives, or the rate at which existing formed judgment turns over.

\begin{theorem}[Sharp duration bounds and sufficient temporary cutoff shifts]\label{thm:oa-direct-formation}
Suppose $H$ is strictly increasing with full support on $\mathbb R$, $\widehat R$ is continuous and
nondecreasing, the post-withdrawal reduced map satisfies the three-culture alternating sign chart of
Theorem~\ref{thm:history-expectations}(iii) with $\mathcal E_h^{\min}$ nonempty, and
$J_0<J_-^{RE}\le J_+^{RE}<1$ for the post-withdrawal cutoffs defined above. Let
\[
 \underline R:=\min_{J\in[0,1]}\widehat R(J),
 \qquad
 \overline R:=\max_{J\in[0,1]}\widehat R(J).
\]
Fix $T\ge1$. At date $t<T$, let a direct formation instrument $\tau_t\ge0$ shift the revising cohort's cutoff
to $k\le V_{t+1}+\tau_t$, and set $\tau_t=0$ for $t\ge T$. For every anticipated stock sequence,
the revising share $B_t$ obeys
\begin{equation}\label{eq:oa-direct-formation-envelope}
 H(\underline R+\tau_t)\le B_t\le H(\overline R+\tau_t).
\end{equation}
Hence any target $b_t\in(0,1)$ is guaranteed by
\begin{equation}\label{eq:oa-direct-formation-implement}
 \tau_t\ge\left[H^{-1}(b_t)-\underline R\right]_+.
\end{equation}

For a terminal target $c\in(J_0,1)$, define
\[
 C_T:=1-(1-J_0)d^T,
 \qquad
 b_T(c):=\frac{c-d^TJ_0}{1-d^T}.
\]
If $C_T\le c$, no finite direct formation instrument can attain $J_T>c$. If $C_T>c$, a uniform policy
$\tau_t\equiv\tau$ for $t<T$ guarantees $J_T>c$ for every admissible continuation whenever
\begin{equation}\label{eq:oa-direct-formation-subsidy}
 H(\underline R+\tau)>b_T(c),
 \qquad\text{equivalently}\qquad
 \tau>H^{-1}(b_T(c))-\underline R.
\end{equation}
The uniform worst-case support bound over nonnegative shifts is
\begin{equation}\label{eq:oa-direct-formation-infimum}
 \tau_T^{\mathrm{env}}(c)
 :=\left[H^{-1}(b_T(c))-\underline R\right]_+.
\end{equation}
Every $\tau>\tau_T^{\mathrm{env}}(c)$ suffices, while $\tau=0$ also suffices when
$H(\underline R)>b_T(c)$. This is the infimum for the pointwise continuation-value envelope; restrictions
linking continuation values across dates may weakly lower the sufficient cutoff in the equilibrium-constrained
problem.

For the logit formation distribution
\[
 H(x)=\left[1+\exp\left(-\frac{x-\mu_k}{\varsigma_k}\right)\right]^{-1},
\]
the envelope infimum becomes
\begin{equation}\label{eq:oa-direct-formation-logit}
 \tau_T^{\mathrm{env}}(c)
 =\left[\mu_k+\varsigma_k\log\frac{b_T(c)}{1-b_T(c)}-\underline R\right]_+.
\end{equation}
Consequently, a finite direct formation policy can make high-culture convergence possible at
$T_{\mathrm{possible}}$ and guarantee it under every admissible continuation at $T_{\mathrm{robust}}$. These are sharp duration bounds when
every finite direct cutoff shift in this policy class has $B_t<1$ and finite shifts can make the finitely many
pre-withdrawal cohort responses uniformly arbitrarily close to one; sharpness is within this class.
\end{theorem}

\begin{proof}
Because the weights in the continuation value are nonnegative and sum to one, every bounded anticipated path
satisfies
\[
 V_{t+1}=(1-\lambda)\sum_{n=0}^{\infty}\lambda^n
 \widehat R(J_{t+1+n})\in[\underline R,\overline R].
\]
The cutoff shift gives $B_t=H(V_{t+1}+\tau_t)$, so monotonicity of $H$ proves
\eqref{eq:oa-direct-formation-envelope}--\eqref{eq:oa-direct-formation-implement}. For a policy schedule
$\boldsymbol\tau$, define
\begin{align*}
 (\mathcal T_{J_0}^{\boldsymbol\tau}\mathbf x)_0&=J_0,\\
 (\mathcal T_{J_0}^{\boldsymbol\tau}\mathbf x)_{t+1}
 &=dx_t+\delta H\!\left((1-\lambda)\sum_{n=0}^{\infty}\lambda^n
 \widehat R(x_{t+1+n})+\tau_t\right).
\end{align*}
This is an isotone self-map of the complete lattice of paths with first coordinate $J_0$, so it has least and
greatest fixed points. Under a uniform instrument, iteration of the stock law and the lower envelope give
\[
 J_T\ge d^TJ_0+(1-d^T)H(\underline R+\tau).
\]
This lower bound exceeds $c$ exactly under \eqref{eq:oa-direct-formation-subsidy}. Treating the pointwise value
interval as a Cartesian relaxation gives \eqref{eq:oa-direct-formation-infimum}; inversion of the logit c.d.f.
gives \eqref{eq:oa-direct-formation-logit}.

Full support implies $H(x)<1$ at every finite $x$. Hence any finite instrument has $B_t<1$ and
$J_T<C_T$, proving impossibility when $C_T\le c$. If $C_T>c$, then $b_T(c)\in(0,1)$ and
$H(\underline R+\tau)\to1$ as $\tau\to\infty$, so a finite sufficient shift exists.

It remains to connect terminal crossing to post-withdrawal selection. For $c=J_+^{RE}$, every policy path has
$J_T>c$, and every post-withdrawal perfect-foresight tail from that stock converges to $J_h$. For
$c=J_-^{RE}$, take the greatest policy path and the greatest post-withdrawal path from its terminal stock. The
latter weakly dominates the tail of the former. If domination were strict, pasting the greater tail onto the
policy prefix would produce a subsolution strictly above the greatest fixed point of
$\mathcal T_{J_0}^{\boldsymbol\tau}$; monotone iteration would then yield a still greater fixed point, a
contradiction. The greatest policy path therefore inherits the greatest post-withdrawal tail, which converges to
$J_h$ when $J_T>J_-^{RE}$. For necessity at $T_{\mathrm{robust}}$, the mirror-image pasting applies to the
least fixed point of $\mathcal T_{J_0}^{\boldsymbol\tau}$: its post-withdrawal tail equals the least
post-withdrawal path from its terminal stock---otherwise, pasting the smaller tail onto the policy prefix
would produce a supersolution strictly below the least fixed point, and downward monotone iteration from it
would yield a still smaller fixed point, a contradiction. At any duration whose physical ceiling lies weakly
below $J_+^{RE}$, every finite instrument leaves the least fixed point's terminal stock strictly below
$J_+^{RE}$, so its tail---the least post-withdrawal path from that stock---does not converge to $J_h$,
because stocks strictly below $\inf\mathcal E_h^{\min}$ lie outside the upper set $\mathcal E_h^{\min}$:
the robust guarantee fails at every shorter duration. Applying these arguments at the first strict
physical-ceiling crossings gives
possibility at $T_{\mathrm{possible}}$ and convergence under every admissible continuation at $T_{\mathrm{robust}}$. The limiting response
$H(\underline R+\tau)\to1$ and strict inequality at every finite $\tau$ complete sharpness.
\end{proof}

\begin{corollary}[Capped-instrument duration]\label{cor:oa-capped-duration}
Fix a cap $\bar\tau>0$ and restrict the instrument to $\tau_t\in[0,\bar\tau]$, and let
$\bar H(\bar\tau):=H(\underline R+\bar\tau)<1$. Under the envelope
\eqref{eq:oa-direct-formation-envelope}, the guaranteed stock after $T$ dates of the uniform capped policy
$\tau_t\equiv\bar\tau$ is
\[
 J_T\ \ge\ d^TJ_0+(1-d^T)\bar H(\bar\tau).
\]
For a target $c\in(J_0,1)$: if $\bar H(\bar\tau)\le c$, no duration makes the envelope guarantee
$J_T>c$ available at cap $\bar\tau$; if $\bar H(\bar\tau)>c$, the least duration at which the envelope
guarantees $J_T>c$ is
\begin{equation}\label{eq:oa-capped-duration}
 T^{\mathrm{env}}(c,\bar\tau)
 =\left\lfloor\frac{\log\bigl[(\bar H(\bar\tau)-c)/(\bar H(\bar\tau)-J_0)\bigr]}{\log d}\right\rfloor+1 .
\end{equation}
$T^{\mathrm{env}}(c,\bar\tau)$ is nonincreasing in $\bar\tau$ and recovers the uncapped durations
$T_{\mathrm{possible}}$ and $T_{\mathrm{robust}}$ displayed above as $\bar\tau\uparrow\infty$.
\end{corollary}

\begin{proof}
The lower bound is the uniform-instrument display in the proof of
Theorem~\ref{thm:oa-direct-formation} at $\tau=\bar\tau$. When $\bar H(\bar\tau)>c>J_0$ it is strictly
increasing in $T$ and exceeds $c$ exactly when
$d^T<[\bar H(\bar\tau)-c]/[\bar H(\bar\tau)-J_0]$; the least such integer is
\eqref{eq:oa-capped-duration}. When $\bar H(\bar\tau)\le c$ the bound never exceeds $c$. Monotonicity in
$\bar\tau$ follows from monotonicity of $H$, and $\bar H(\bar\tau)\uparrow1$ recovers the uncapped
formulas.
\end{proof}

As with the theorem, the statement is envelope-conservative: it bounds what the pointwise
continuation-value envelope can guarantee at cap $\bar\tau$, and does not assert that crossing is
impossible at smaller caps or shorter durations.

\begin{proposition}[Patience changes selection, not stationary cultures]\label{prop:oa-beta-selection}
The discount factor $\beta$ enters neither stationary equilibrium condition, so the set of stationary cultures
$\{J_\ell,J_i,J_h\}$ and the associated $(A,q_m,s_m)$ are invariant to $\beta$. It does enter the continuation
operator through $\lambda=\beta d$, and therefore can move the global selection boundaries $J_-^{RE}$ and
$J_+^{RE}$.
\end{proposition}

\begin{proof}
On a stationary path, \eqref{eq:oa-forward-value} and $\lambda<1$ give $V=\widehat R(J)$; because $\delta>0$,
\eqref{eq:oa-forward-stock} gives $J=H(V)$. Neither condition involves $\beta$, which proves stationary
invariance. The operator $\mathcal T_j$ of the main text depends on $\beta$ only through $\lambda$, and its
extremal fixed points are not pinned down by the stationary conditions, so the boundaries defined from them
need not be invariant.
\end{proof}

\paragraph{Direction of the effect at the worked primitives.} How the boundaries move is a quantitative
question, and at these primitives it is settled numerically rather than by certificate. The superseded
benchmark carried a companion interval certificate over $\beta$--$\delta$ rectangles; the present package
replaces it with \texttt{compare\_patience\_F3.py}, which locates the boundaries by the same drift criterion in
floating point and is labelled non-rigorous throughout. At $\delta=0.05$ it returns
\[
 \beta=0.92:\ J_-^{RE}\approx0.154,\ J_+^{RE}\approx0.170;
 \qquad
 \beta=0.95:\ J_-^{RE}\approx0.122,\ J_+^{RE}\approx0.181 .
\]
More patient members therefore lower the stock above which high convergence becomes possible and raise the stock
above which it occurs under every admissible continuation, widening the expectation-dependent region from about $0.015$ to about
$0.058$: patience makes a group both easier to rescue and easier to lose. Starting from the benchmark trap
stock, the constructive duration pair is $(3,4)$ at both values, because the two boundaries move within the
same steps of the physical ceiling. None of this establishes a derivative sign or global monotonicity in
$\beta$, and the main text reports it accordingly.

\section{Certification details and smooth variants}\label{oa:certification}

\paragraph{Exact-decimal benchmark inputs.}
The decimal values in this subsection are part of a constructed example, not estimates. The Arb
programs read them as exact rational inputs and round only the resulting evaluations outward. Formation and
resolve costs are
\[
 k\sim\mathrm{Logistic}(-0.08,0.115),
 \qquad
 \ell\sim\mathrm{Logistic}(-0.06,0.32).
\]
The audience primitives are
\[
\begin{gathered}
 (\nu_f,\nu_a)=(0.25,0.33),\qquad
 \alpha_f=\alpha_a=1,\qquad
 \kappa_a=2.32,\\
 \phi=0.44,\qquad \psi=1,\qquad
 (\bar q_f,\bar q_a)=(0.16,0.207).
\end{gathered}
\]
Formation pressure is $\kappa_f=1.2\chi_f=1.62$, where $\chi_f=1.35$ is only the reporting normalization
used in the robustness figures. The remaining positive and output primitives are
\[
 b_N=1.81,\qquad \gamma=0.069,\qquad
 V_F=1,\qquad V_N=0.30,\qquad D=0.17,
 \qquad p=0.10.
\]
The smooth-rule robustness comparison uses $\varepsilon=0.02$; the certified baseline itself uses the hard
withdrawal rule on its uncapped sanctioned interior branch. Indeed,
\[
 q_f(J)\ge q_f(1)=\qfone>\bar q_f,\qquad
 q_a(A)\ge q_a(1)=\qaone>\bar q_a
\]
throughout the physical domain. The inactive kink is deliberate: it rules out threshold crossing as the source
of the baseline folds. For endogenous competence the exact inputs are
\[
 \underline\sigma=0.30,\qquad
 \overline\sigma_{\max}=0.95,\qquad
 m=0.15,\qquad L=0.80,
\]
and the dynamic benchmark uses
\[
 \delta=\deltaval,\qquad \beta=0.95.
\]

\paragraph{Worked-example diagnostics.}
The enactment-regularity modulus is $0.2170$ in floating point, and the rigorous certificate puts it
below $0.25$, hence below one. At $p=0.10$ the equilibrium roots are
$(J_\ell,J_i,J_h)=(\Jlow,\Jmid,\Jhigh)$ after rounding. The competent planner chooses
$(J^\ast,A^\ast)=(\plannerCompJ,\plannerCompA)$ and the no-pressure equilibrium is
$(J^0,A^0)=(0.856,0.854)$; the corresponding level decompositions are reported at the unrounded roots.

Under endogenous competence,
\[
 \sigma(A_\ell)=\sigmalow<\bar\sigma=\sigmabar<\sigmahigh=\sigma(A_h),
\]
with novel-state values $\YNLlow$ and $\YNLhigh$. The learning-externality term at the high culture is $0.100$,
and the endogenous-competence planner chooses
$(J^{\ast\ast},A^{\ast\ast})=(\plannerJ,\plannerA)$. Because the hyperbolic objective is not concave, uniqueness here
is supplied by the exhaustive interval certificate of the joint first-order system together with the
existence and interiority argument in OA.1.

\paragraph{Certified fold geometry.}
In the main paper's worked example, the formation equation in the explicit enactment coordinate $A$ is affine
in novelty: $Q(A,p)=pD(A)-N(A)$ with $D(A)=Q_p>0$. Hence all solutions of the fold system $Q=Q_A=0$ are exactly
the zeros of $N'D-ND'$. Exhaustive root isolation with Arb over the full physical domain, with analytic endpoint-tail
bounds, proves that there are exactly two and no other folds:
\[
 (J,p)=(0.411405\ldots,0.067104\ldots)
 \quad\text{and}\quad
 (0.057587\ldots,0.142866\ldots).
\]
The certificate also proves $J'(A)>0$ on the full physical enactment branch and excludes zero from $Q_p$,
$Q_{AA}$, and the Jacobian determinant of $(Q,Q_A)$ at both boxes. Hence the $A\leftrightarrow J$ change of
coordinate is locally regular, these nondegeneracy conditions transfer to the original $(J,p)$ system, and both
points are nondegenerate saddle-nodes. Certified representative counts, together with the absence of any other
singular root or boundary entry, give exactly one equilibrium below, three inside, and one above the
coexistence window $p\in(\pfoldlo\ldots,\pfoldhi\ldots)$.

Holding $p=0.10$, \texttt{certify\_global\_geometry\_F3.py} proves that the full two-margin baseline has exactly
two stationary points of $\mathcal K$, at
\[
 J_-=0.068784\ldots,\qquad J_+=0.550367\ldots,
\]
with global derivative pattern $(-,+,-)$, nonzero curvature at both points, and certified pressure values
\[
 \mathcal K(J_-;\pbench)/\chi_f=\Kminchi\ldots,\qquad
 \mathcal K(J_+;\pbench)/\chi_f=\Kmaxchi\ldots.
\]
The exact-decimal baseline pressure $\kappa_f/\chi_f=1.2$ lies strictly between them. The same proof isolates
exactly three equilibrium roots, certifies their stable--unstable--stable slope pattern, and proves the unique
joint first-order solution of the endogenous-competence planner. Thus the exact-count result at $p=0.10$ is a
computer-assisted proof on the full domain, not a dense-grid diagnostic. The $\mathcal K$ geometry and
three-root statements fix $p=0.10$; the novelty-fold proof fixes every other exact-decimal primitive and varies
only $p$ over $[0,1/2]$, the range in which familiar tasks are at least as frequent as novel ones. The main paper's perturbation-robustness result gives qualitative persistence under small smooth
perturbations, but the certificate is not a quantified multidimensional parameter box.

\paragraph{Finite certificate for the convergence cutoffs.}
For completeness, this paragraph states the finite conditions used to turn the forward-looking path operator
into certified enclosures for the convergence cutoffs. Let $d:=1-\delta$, $\lambda:=\beta d$, and, for an inherited stock $j$,
let $\mathbb X(j)$ be the sequences in $[0,1]^{\mathbb N_0}$ whose first coordinate is $j$. Define
\begin{align*}
 (\mathcal T_j\mathbf x)_0&=j,\\
 (\mathcal T_j\mathbf x)_{t+1}
 &=dx_t+\delta H\!\left((1-\lambda)\sum_{n=0}^{\infty}
 \lambda^n\widehat R(x_{t+1+n})\right).
\end{align*}
Write $\bot_j=(j,0,0,\ldots)$ and $\top_j=(j,1,1,\ldots)$. Iteration from these two paths gives the least and
greatest perfect-foresight paths, denoted $\underline{\mathbf J}(j)$ and $\overline{\mathbf J}(j)$.

The following comparison principle is used repeatedly and is worth stating explicitly. If two paths are
ordered coordinatewise, $\mathbf x\le\mathbf y$, then $\limsup_tx_t\le\limsup_ty_t$ and
$\liminf_tx_t\le\liminf_ty_t$; in particular, every equilibrium path from $j$ is squeezed between
$\underline{\mathbf J}(j)$ and $\overline{\mathbf J}(j)$ (Theorem~\ref{thm:history-expectations}(i) of the
main text), so when a
dominating extremal path converges, the limit superior of every dominated path is bounded by its limit---and
these bounds hold whether or not the dominated path itself converges.

\begin{lemma}[Finite certificate for the convergence cutoffs]\label{lem:oa-finite-destination}
Suppose $H$ and $\widehat R$ are continuous and nondecreasing and the static map
$\Psi=H\circ\widehat R$ has exactly three fixed points with sign pattern $+,-,+,-$. Choose
\[
 J_\ell<a_-<a_+<J_i<b_-<b_+<J_h.
\]
Suppose that, for finite integers $n_-,n_+$,
\[
 [\mathcal T_{a_-}^{n_-}\top_{a_-}]_1\le a_-,
 \qquad
 [\mathcal T_{b_+}^{n_+}\bot_{b_+}]_1\ge b_+.
\]
Suppose also that there is a subsolution
$\mathbf s\in\mathbb X(a_+)$ satisfying
$\mathbf s\le\mathcal T_{a_+}\mathbf s$ and $\liminf_ts_t>J_i$, and a supersolution
$\mathbf q\in\mathbb X(b_-)$ satisfying
$\mathcal T_{b_-}\mathbf q\le\mathbf q$ and $\limsup_tq_t<J_i$. Then both high-convergence sets are nonempty.
Define
\[
 J_-^{RE}:=\inf\{j:\overline J_t(j)\to J_h\},
 \qquad
 J_+^{RE}:=\inf\{j:\underline J_t(j)\to J_h\}.
\]
Then
\begin{equation}\label{eq:oa-constructive-global-regions}
\begin{aligned}
 j\le a_-&\Longrightarrow\text{ every path converges to }J_\ell,\\
 a_+\le j\le b_-&\Longrightarrow
 \underline J_t(j)\to J_\ell\ \text{and}\ \overline J_t(j)\to J_h,\\
 j\ge b_+&\Longrightarrow\text{ every path converges to }J_h,
\end{aligned}
\end{equation}
and
\[
 J_-^{RE}\in[a_-,a_+],
 \qquad
 J_+^{RE}\in[b_-,b_+].
\]
\end{lemma}

\begin{proof}
The finite upper-iterate condition implies
$\overline J_1(a_-)\le a_-$. The shifted tail of the greatest path is a path from
$\overline J_1(a_-)$ and is bounded above by
$\overline{\mathbf J}(\overline J_1(a_-))\le\overline{\mathbf J}(a_-)$. The greatest path is therefore
nonincreasing; its stationary limit is $J_\ell$ because $a_-<J_i$. Monotonicity in inherited stock extends low
convergence of every path to $j\le a_-$. The finite lower-iterate condition symmetrically makes the least path
from $b_+$ nondecreasing and convergent to $J_h$, and extends high convergence of every path to $j\ge b_+$.

Iteration from $\mathbf s\le\mathcal T_{a_+}\mathbf s$ increases to a fixed path whose liminf exceeds $J_i$;
the sign pattern forces that path to converge to $J_h$. Hence the greatest path converges high from $a_+$ and,
by monotonicity, from every larger stock. Iteration down from $\mathbf q$ yields a fixed path whose limsup is
below $J_i$, so it converges to $J_\ell$; the least path therefore converges low from every $j\le b_-$. These
four implications give \eqref{eq:oa-constructive-global-regions} and the cutoff brackets.
\end{proof}

The conditions are finitely checkable. If
$\underline R:=\min_J\widehat R(J)$ and
$\overline R:=\max_J\widehat R(J)$, truncating the continuation sum after term $N$ and bounding the remainder
by its extrema gives an enclosure of width no larger than
\begin{equation}\label{eq:oa-tail-error}
 \lambda^{N+1}(\overline R-\underline R).
\end{equation}
Thus the drift tests are exact finite iterates with constant tails, while the sub- and supersolutions are finite
backward recursions followed by constant tails. Neither step extrapolates convergence from a simulated path.

At the dynamic benchmark, the extremal-path convergence tests classify inherited stocks $0.05$, $0.15$, and $0.25$
as low--low, low--high, and high--high, respectively. Starting from the low stationary culture, the physical
ceilings after zero through four intervention dates are $0.020671$, $0.069638$, $0.116156$, $0.160348$, and
$0.202330$.

The independent dynamic certificate \texttt{certify\_forward\_dynamics\_F3.py} re-certifies the unique enactment
branch, the three stationary roots, and the global $+,-,+,-$ sign chart. It then evaluates finite Tarski
iterates with constant tails and two finite backward-recursion constructions. The $900$th upper iterate at
$j=0.12$ has first successor $0.119774<0.12$, while the $400$th lower iterate at $j=0.19$ has first successor
$0.191751>0.19$. A $70$-step subsolution with tail $0.5$ reconstructs $j=0.124815<0.13$, and an $80$-step
supersolution with tail $0.10$ reconstructs $j=0.176387>0.17$. These two backward recursions are finite
nonlinear instances of the spiral construction behind Theorem~\ref{thm:band} of the main text: both loop
around the middle culture before anchoring on their constant tails, which is how they reach inherited
stocks on the far side of $J_i$. The finite certificate therefore yields
\[
 J_-^{RE}\in[\REminlo,\REminhi],\qquad J_+^{RE}\in[\REpluslo,\REplushi],
\]
without shooting interpolation or convergence extrapolation. The same certificate verifies the physical
ceilings, the exact duration split $(T_{\mathrm{possible}},T_{\mathrm{robust}})=(\Tpossible,\Trobust)$, the return floor
$\widehat R(J)>-0.577249$, and the strict sufficient-shift inequalities
\[
 \tau=0.648\Longrightarrow J_3>0.130063,
 \qquad
 \tau=0.802\Longrightarrow J_4>0.190090 .
\]

\paragraph{Patience comparison (numerical, not certified).}
The superseded benchmark carried a companion interval certificate over $\beta$--$\delta$ rectangles; at the
present primitives that exercise is replaced by the explicitly non-rigorous
\texttt{compare\_patience\_F3.py}. Section~\ref{oa:patience} reports the comparison and its scope.

Script \texttt{verify\_planner\_threshold\_F3.py} reproduces the responsive-opposition
threshold. Objects outside the scope of these Arb-based certificates
use high-precision diagnostics, with residuals and transversality terms reported.

For reference, the hard audience rule is
$s_m^W(X_m)=[\kappa_m(q_m(X_m)-\bar q_m)]_0^{\bar s_m}$, with
$\kappa_m=(\pi_m+\Lambda_m)/c_m$ in the matching microfoundation. The lower truncation is the audience's
tolerance region and the upper truncation is the finite partner mass; the no-cap convention applies when the
equilibrium sanction range remains below that mass. Replacing the positive part by its logit smoothing (the ``softplus'') gives the
no-cap smooth approximation in the main paper. At zero tolerance it converges to proportional reputational sanctions
as the smoothing parameter vanishes. A finite exponential saturation is a distinct smooth-cap technology, not
the hard cap.

The hard interior rule is not the only nearby case with a trap. The vulnerability script replaces stigma by
the softened, optionally capped family
\[
 \tilde s_f(J)=\bar s\left(1-\exp\left\{
 -\frac{\kappa_f\,\varepsilon\log(1+\exp[(q_f(J)-\bar q)/\varepsilon])}{\bar s}
 \right\}\right),
\]
with the no-cap case obtained as $\bar s=\infty$. With finite $\bar s$ this is a smooth saturating-cap
technology; its curvature need not satisfy $s''\ge0$ globally, so it is reported as robustness only. The
zero-threshold limit leaves a unique culture at $J=0.0052$. With the benchmark tolerance $\bar q_f=0.16$ and
$\varepsilon=0.02$, the softened approximation gives $J=(0.0188,0.1966,0.7124)$, reproducing the hard-rule
baseline's three-root structure and stability pattern. Smooth saturating caps, by contrast, remove it: at
$\bar s=2.0$ and at $\bar s=1.5$ a unique questioning culture survives. The reason is visible in the
schedule---the saturating form is concave, so at $J\to0$ it delivers $0.99$ against the linear rule's $1.36$,
cutting roughly a quarter of the trap-level sanction on which the feedback depends. The example with nondegenerate folds
therefore survives softening of the kink but not caps that blunt the high sanction at low prevalence.

\section{Identification: limits and design}\label{oa:identification}

This appendix states the two identification results summarized in Section~5 of the main text, maintaining
throughout that the analyst does not know the recording and background rates (were those primitives known,
event frequencies alone would recover the stocks). Two layers of non-identification should be kept
distinct: with the recording and background rates unknown, the stocks $(J,A)$ themselves are not identified
from event frequencies; and even granted $(J,A)$, splitting non-enacted judgment into its silenced and
privately withheld parts requires the background enactment propensity $G(b_N)$.

\subsection{What conduct cannot reveal}\label{subsec:limits}

\begin{proposition}[Limits to identifying latent adaptive capacity]\label{thm:hidden}
Maintain the no-prompt-vindication baseline $v_f=0$. Consider any environment in which \emph{(A1)} visible independent conduct pools authentic and background
sources that no single act distinguishes: in calm states, independent events (probes) mix authentic formers
with the opportunistic pool, and all other conduct is the common routine; \emph{(A2)} in rare states an agent
departs from the routine only if she is both able and, given exposure, willing; and \emph{(A3)} the population
partitions into never-formed $1-J$, formed-but-silenced $S^R$, formed-but-privately-withheld $N$, and
formed-and-enacting $A$, with $A+S^R+N=J$; the authentic flow $A$ enters the observed novel-state events
together with the background pool, as in (A1). The two-margin model is one
such environment. Consider an analyst who observes conduct (event frequencies and actions) but neither
motives nor the primitives $(\phi_0,\psi,\nu_m,H,G,b_N)$, with the visibility normalization $\alpha_m=1$
maintained. Then:
\begin{enumerate}
\item[(i)] \emph{(No act classifies its actor.)} Whenever $X_m>0$, every observed independent act has a
posterior probability of being authentic strictly inside $(0,1)$, and a conforming agent may be unformed,
silenced, or privately withheld.
\item[(ii)] \emph{(Aggregates are confounded.)} The calm-state event frequency $\phi_0(\nu_f+J)$ is
a single moment in which the stock $J$ is confounded with the background composition and visibility, and the
registered rare-state event mass $\psi(\nu_a+A)$ is likewise confounded with $\nu_a$ and $\psi$.
\item[(iii)] \emph{(The composition is counterfactual even given aggregates.)} Granting the analyst
$(J,A)$, the split of non-enacted judgment $J-A$ into its reputationally silenced and privately withheld
parts requires the no-stigma enactment rate $w^0$, which no conduct under a single exposure regime reveals;
under (A1)--(A3) alone, conduct identifies at most the line-segment outer bound
$(S^R,N)\in\{(S^R,N)\ge0:\ S^R+N=J-A\}$.
\item[(iv)] \emph{(Sharpness, within the two-margin model.)} The bound is sharp on its relative interior
whenever $J>0$, the enactment margin is strictly sanctioned (so that $z(A)<b_N$), and the analyst's
uncertainty admits any full-support cost primitives: for every target split with $S^R>0$ and $N>0$ there
exist admissible primitives, differing only in the unobserved $(G,H)$, whose equilibrium generates identical
conduct at every margin and assigns that split. There this line segment is the identified set up to its boundary.
\end{enumerate}
\end{proposition}

\subsection{Identification by design}\label{subsec:design}

Point identification is possible within the two-margin model, but---within the maintained environment in
which the recording and background rates are unknown---only by design: the required instruments
are assumptions about measurement rather than objects the model supplies, and conditions (D2)--(D3) below
come close to observing the stocks directly. The next statement is therefore a
proposition under explicit design conditions, not a theorem of the model.

\begin{proposition}[Point identification under design conditions]\label{prop:pointid}
Maintain $v_f=0$ and consider the following design conditions:
\begin{enumerate}
\item[(D1)] \emph{Immediate enactment protection:} an exogenous no-stigma enactment-protection intervention is
applied before the inherited stock $J$ can adjust, leaving $G$, $b_N$, the background flow $\nu_a$, and the
measurement technology unchanged; the novel-state departure mass is measured in both regimes, either directly
as acts (the total mass $\nu_a+A$, at the normalization $\alpha_a=1$) or as registered events with a
regime-invariant registration rate $\psi$;
\item[(D2)] \emph{Truthful elicitation:} a private elicitation instrument reports formation status
truthfully;
\item[(D3)] \emph{Measured nuisance levels:} (a) the level $\nu_a$ is independently measured; and (b), when
departures are registered as events rather than directly counted as acts, the registration rate $\psi$ is known.
\end{enumerate}
Under (D1) the background cancels from the change induced by enactment protection: acts give $A^0-A=S^R$ in levels, while events
give $\psi S^R$. With acts, (D2) gives $J$ and hence $(1-J,S^R,A+N)$, where $A+N=J-S^R$; events require the
known scale in (D3b) for the same conclusion. If $\psi$ is unknown but regime-invariant, event changes identify
only $\psi S^R$ and its sign, and rank $S^R$ across groups only when they share $\psi$. The change alone does
not split $A+N$. Clause (D3a) recovers $A=(\nu_a+A)-\nu_a$ from acts and, together with (D3b),
$A=y/\psi-\nu_a$ from events; $A^0$ follows likewise, so $S^R=A^0-A$, $N=J-A^0$, and the full static
composition are point-identified. A small dated exposure shock around an outer culture further separates
impact timing and, from local decay along its stable arm, identifies the composite stable root characterized
in Section~\ref{sec:dynamics}.
\end{proposition}

This is composite identification: the local decay rate identifies $r_j^s$, while separating turnover requires
external knowledge of discounting and the local slope of the cultural map.

The design conditions are instruments taken as available rather than derived: truthful private elicitation
is treated as given (modeling why a silenced member reports truthfully is a mechanism-design question beyond
this paper); under responsive opposition the enactment-protection comparison also counts deterred opportunists, which
is why (D1) requires the pool to remain fixed; and separating enacted from privately withheld judgment requires
its level---and, for events, the registration scale---to be measured, not merely stable. In the committee application, enactment protection is implemented by a chair who protects
dissenting votes on unusual cases, elicitation is an anonymous canvass of members' private assessments, and the
exposure shock is a dated change in how visibly questioning enters the record.

\subsection*{Auxiliary rank separation under exclusion restrictions}\label{oa:rank}

This subsection records the formal rank comparison used only as an auxiliary diagnostic. It uses the main paper's
notation: $J$ is formed judgment, $A$ is enacted judgment, and $(\theta_f,\theta_a)$ are pressure shifters at
the formation and enactment margins. The comparison is deliberately restricted by exclusion restrictions; it
does not rule out one-state models with direct instrument effects or additional hidden states.

\begin{lemma}[Rank separation under exclusion restrictions]\label{lem:oa-rank}
Consider a comparison class in which a scalar state $X$ solves $X=\Gamma(X;\vartheta)$ and a vector of
observable statistics is $\zeta(X)$. Suppose the instruments $\vartheta$ have no direct effect on $\zeta$
conditional on $X$. At every regular solution,
\[
 \frac{\partial\zeta}{\partial\vartheta}
 =\zeta'(X)\frac{\partial X}{\partial\vartheta}
\]
has rank at most one.

In the two-margin model, write the unique within-period enactment solution as
$A=a(J;\theta_a)$; formation pressure has no direct effect on $A$ conditional on $J$. At a regular equilibrium
on strictly sanctioned smooth interior branches at both margins, suppose
$J_{\theta_f}\ne0$ and $a_{\theta_a}\ne0$, as established by the main paper's comparative statics. Then
\[
 \det\frac{\partial(J,A)}{\partial(\theta_f,\theta_a)}
 =J_{\theta_f}a_{\theta_a}\ne0,
\]
so the response of $(J,A)$ has rank two.

For a further fixed-stock implication, let $w^0$ be the enactment rate absent enactment stigma and define
reputationally suppressed judgment by $S^R=Jw^0-A$. On a strictly sanctioned enactment branch,
\[
 S^R>0,
 \qquad
 \left.\frac{\partial S^R}{\partial\theta_a}\right|_J=-a_{\theta_a}>0.
\]
If a candidate single-index model has a locally invertible stock equation $J=f(X)$, holding $J$ fixed also fixes
$X$ and therefore fixes every statistic $\zeta(X)$. This fixed-stock response is another implication of the
same exclusion restrictions.
\end{lemma}

\begin{proof}
The chain rule gives
$\partial\zeta/\partial\vartheta=\zeta'(X)\partial X/\partial\vartheta$, the product of a column and a row;
its rank is therefore at most one. In the two-margin model, total differentiation of
$A=a(J;\theta_a)$ gives the two Jacobian rows
\[
 (J_{\theta_f},J_{\theta_a})
 \quad\text{and}\quad
 (a_JJ_{\theta_f},a_JJ_{\theta_a}+a_{\theta_a}).
\]
The cross terms cancel in the determinant, leaving
$J_{\theta_f}a_{\theta_a}\ne0$. Finally, strict enactment pressure implies an enactment rate below $w^0$ and
hence $S^R>0$. Holding $J$ fixed in $S^R=Jw^0-A$ gives
$\partial S^R/\partial\theta_a=-a_{\theta_a}>0$. If $J=f(X)$ is locally invertible, a fixed $J$ fixes $X$,
which proves the last comparison.
\end{proof}

\section{Omitted proofs}\label{oa:omitted}

\begin{proof}[Proof of Proposition~\ref{prop:practice}]
At stationary aggregates, maintaining readiness yields $R(J,A;p)-k$ relative to not maintaining in every
period, independently of history. Thus the unique optimum maintains every period when $k<R$ and
never maintains when $k>R$. At $k=R$ any history-contingent policy is optimal; the maintained tie-breaking rule
selects stationary maintenance, and continuity of $H$ gives the tie type zero mass.
\end{proof}

\begin{proof}[Proof of Lemma~\ref{lem:reputation}]
For the no-cap rule, let $\Phi_\varepsilon(x)=(1+\exp[-x/\varepsilon])^{-1}$. Then
\[
 s_m'=\kappa_m\Phi_\varepsilon(q_m-\bar q_m)q_m'<0,
\]
and
\[
 s_m''=\kappa_m\left[
 \frac{\Phi_\varepsilon(q_m-\bar q_m)(1-\Phi_\varepsilon(q_m-\bar q_m))}{\varepsilon}(q_m')^2
 +\Phi_\varepsilon(q_m-\bar q_m)q_m''\right]>0,
\]
because $q_m'<0$ and $q_m''>0$. Positivity and $\partial s_m/\partial\kappa_m>0$ are immediate. The hard
tolerance rule $[\kappa_m(q_m(X_m)-\bar q_m)]_0^{\bar s_m}$ is the $\varepsilon\downarrow0$ version with a
linear interior branch and flat tolerated or capped regions.
\end{proof}

\begin{proof}[Proof of Proposition~\ref{prop:tolerance}]
Partner $u$ withdraws exactly when $u\le[(\pi_m+\Lambda_m)q-\Lambda_m]/c_m$; truncating this cutoff to
$[0,\bar s_m]$ gives \eqref{eq:endogenous-tolerance}. The same expression maximizes the concave aggregate
audience payoff \eqref{eq:audience-payoff}. Since
$q_m'(X_m)=-\nu_m\alpha_m/(\nu_m+\alpha_mX_m)^2<0$, the derivative on an interior region is the displayed
expression. If $X_h>X_\ell$, then $q_m(X_h)<q_m(X_\ell)$. Any threshold between the two posteriors makes
$W_m(q_m(X_h))=0$ and $W_m(q_m(X_\ell))>0$.
\end{proof}

\begin{proof}[Proof of Proposition~\ref{prop:mult}]
The map $\mathcal T$ is continuous, isotone, and maps the complete lattice $[0,1]^2$ into itself; Tarski's
theorem gives a nonempty complete lattice of fixed points. The argument preceding the proposition supplies
the unique regular enactment solution $a(J)$; \eqref{eq:aprime} follows by implicit differentiation.
Substitution yields the increasing scalar map $\Psi$. The scalar uniqueness claim follows from the contraction
mapping theorem; on a smooth branch, the derivative bound implies the stated Lipschitz condition. For the
three-root claim, note that $\Psi(0)>0$ and $\Psi(1)<1$ because $H$ has full
support; the assumed derivative pattern and signs at the two extrema then imply exactly three crossings,
with the stated local stability for the reduced scalar adjustment. For the necessity direction of the
equivalence, write $F:=\Psi-\mathrm{id}$ and suppose $F(\underline J)\ge0$: since $F$ decreases on
$(0,\underline J)$ and increases on $(\underline J,\overline J)$, it is nonnegative on
$(0,\overline J]$ with at most the single zero $\underline J$, and on $(\overline J,1)$ it decreases
strictly from $F(\overline J)\ge0$ to $F(1)<0$, contributing exactly one further zero---at most two fixed
points in total. Symmetrically, $F(\overline J)\le0$ allows at most two. Exactly three fixed points
therefore require $F(\underline J)<0<F(\overline J)$.
\end{proof}

\begin{proof}[Proof of Proposition~\ref{prop:nosocialmeaning}]
With constant sanctions the enactment cutoff $\bar z$ is independent of $A$, so \eqref{eq:AS} gives
$A=\bar wJ$ directly, and the formation return is independent of $(J,A)$; $J^0$ is the unique fixed point of
a constant map.
\end{proof}

\begin{proof}[Proof of Proposition~\ref{prop:capacity}]
Implicit differentiation of \eqref{eq:AS} gives \eqref{eq:aprime} and
\[
 a_{\theta_a}=\frac{-Jg(z)\psi\,\partial s_a/\partial\theta_a}{1-Jg(z)\psi[-s_a'(a)]}\le0
\]
by Assumption~\ref{ass:socialmeaning} and \eqref{eq:enactment-regularity}, with strict inequality on a strictly
sanctioned smooth interior branch. Along $A=a(J)$, with $\Omega'=G$,
\[
 \widehat R_J=pG(z)\psi[-s_a'(a)]\,a'(J)+(1-p)\phi[-s_f'(J)]\ge0,\qquad
 \widehat R_{\theta_f}=-(1-p)\phi\,\frac{\partial s_f}{\partial\theta_f}\le0,
\]
and $\widehat R_{\theta_a}=pG(z)\{-\psi\,\partial s_a/\partial\theta_a-\psi s_a'(a)\,a_{\theta_a}\}\le0$.
The corresponding inequalities are strict under the branch conditions stated in the proposition;
$\widehat R_J$ is strict whenever the enactment schedule is strictly decreasing, or the formation
schedule is strictly decreasing and $\phi>0$. The derivative
$\widehat R_p$ is the displayed direct derivative of
\eqref{eq:invest}. Total differentiation of $J=H(\widehat R(J;\theta))$ gives
$dJ_j/d\theta=h\widehat R_\theta/(1-\Psi'(J_j))$. At a regular stable equilibrium the denominator is positive.
Enacted judgment then satisfies
\[
 \frac{dA}{d\theta_f}=a_J\frac{dJ}{d\theta_f}\le0,\qquad
 \frac{dA}{d\theta_a}=a_J\frac{dJ}{d\theta_a}+a_{\theta_a}\le0,\qquad
 \frac{dA}{dp}=a_J\frac{dJ}{dp}>0,
\]
with the stated strictness. Finally, at fixed $J$, $S^R=Jw^0-a(J;\theta_a)$, so
$\partial S^R/\partial\theta_a=-a_{\theta_a}\ge0$, strictly on a strictly sanctioned interior branch.
\end{proof}

\begin{proof}[Proof of Corollary~\ref{cor:logit-feedback}]
Let $C_p:=(1-p)\phi$. Since $b_f$ is smooth and bounded strictly away from zero on $[0,1]$,
\[
 \mathcal K_\epsilon-\mathcal K_0=\frac{B_\epsilon-B_0}{C_pb_f}\longrightarrow0
\]
in $C^2[0,1]$. The two zeros of $\mathcal K_0'$ are simple because
$\mathcal K_0''(J_-)>0$ and $\mathcal K_0''(J_+)<0$. Choose disjoint neighborhoods of those points on which
the corresponding second-derivative signs are uniform. Outside those neighborhoods, $\mathcal K_0'$ is
bounded away from zero: on a compact middle remainder this follows by continuity, and on sufficiently short
endpoint tails it follows from $\mathcal K_0'(J)\to-\infty$ as $J\downarrow0$ or $J\uparrow1$.
Sufficiently small $C^2$ perturbations therefore preserve the sign of the first derivative outside the two
neighborhoods and the strict monotonicity of that derivative inside them. Each neighborhood contains exactly
one zero of $\mathcal K_\epsilon'$, and there are no others. Convergence of the critical points and values
follows from the implicit-function theorem and uniform convergence. Positivity, strict ordering, stability
signs, and fold nondegeneracy are open strict inequalities.

For this primitive construction, at $\epsilon=0$ the cutoff is the constant
$z_0=b_N-\psi s_a^0$, so $a_0(J)=JG(z_0)$ and
$B_0(p)=p\Omega(z_0)-(1-p)\gamma$. The parameter-dependent implicit-function theorem, together with condition
\eqref{eq:enactment-regularity}, gives $a_\epsilon\to a_0$ in $C^2[0,1]$. Because the prevalence-dependent part of $z_\epsilon$
is multiplied by $\epsilon$, $z_\epsilon(a_\epsilon(\cdot))\to z_0$ in $C^2[0,1]$. Smoothness of $\Omega$
then gives the required convergence of $B_\epsilon$.
\end{proof}

\begin{proof}[Proof of Proposition~\ref{prop:openset}]
Fix $K\subset I^0$ compact with nonempty interior. Because roots are transverse and exactly three for each
$\theta\in K$, the implicit-function theorem represents them locally as smooth graphs. Compactness of $K$
and the ordering of the roots yield three pairwise disjoint neighborhoods of these root graphs, chosen narrow
enough that $F_J$ is uniformly separated from zero and has the appropriate sign throughout each neighborhood.
On the two $J$-boundaries of each neighborhood, and on the compact remainder of $[0,1]\times K$, the baseline function
$F(\cdot,\cdot;\pi^0)$ is bounded away from zero. Uniform $C^0$ closeness therefore preserves the opposite
boundary signs and excludes roots on the remainder, while uniform $C^1$ closeness preserves strict
monotonicity inside each neighborhood. Hence there is exactly one root in each neighborhood, with the two outer roots locally
stable and the middle root unstable.

For the endpoints, apply the implicit-function theorem to
\[
 Q(J,\theta;\pi):=(F(J,\theta;\pi),F_J(J,\theta;\pi))=0
\]
in the unknowns $(J,\theta)$. At a nondegenerate fold, $F_J=0$ and the Jacobian determinant
$F_JF_{J\theta}-F_\theta F_{JJ}$ equals $-F_\theta F_{JJ}\neq0$; uniform convergence of all four
derivatives---including $F_{J\theta}$---keeps the perturbed Jacobian nonsingular near the fold. Thus each
fold has a unique nearby continuation under small $C^2$ perturbations of the primitives.
\end{proof}

\paragraph{Quantified margin.} Write $F^0:=F(\cdot,\cdot;\pi^0)$. For the closures of these three pairwise
disjoint neighborhoods of the root graphs over a compact $K\subset I^0$, let $\eta>0$ be the minimum of
$|F^0|$ on their boundaries and compact complement, and let $m>0$ be the minimum of $|F_J^0|$ on the closed
neighborhoods. The bounds
$\|F-F^0\|_\infty<\eta/2$ and $\|F_J-F_J^0\|_\infty<m/2$ preserve exactly three roots and their signs.
Around each endpoint fold, let $\zeta>0$ bound $|F_\theta^0F_{JJ}^0|$ away from zero; the additional bound
$\|F_\theta F_{JJ}-F_\theta^0F_{JJ}^0\|_\infty<\zeta/2$ preserves nondegeneracy of the continued fold.

\begin{proof}[Proof of Proposition~\ref{thm:hidden}]
Under the maintained $v_f=0$ baseline, with $\nu_m>0$ and $X_m>0$, Bayes' rule \eqref{eq:posterior} places the posterior that any single observed
act is authentic strictly inside $(0,1)$, so no act classifies its actor; a conforming agent's conduct is the
routine in every calm state and non-departure in rare states, both consistent with all three latent
non-enacting types. The observed calm-state event frequency $\phi_0(\nu_f+J)$ is one moment in the
unknowns $(\phi_0,\nu_f,J)$, and the registered departure mass $\psi(\nu_a+A)$ adds $A$ confounded with
$\nu_a$ and $\psi$; neither
moment involves $w^0=G(b_N)$, so the accounting identity \eqref{eq:decomposition} restricts $(S^R,N)$ only to the
stated line segment even given $(J,A)$.

\emph{Sharpness.} Fix an equilibrium $(J,A,q_f,q_a,s_f,s_a)$ with a strictly sanctioned enactment branch and a
target split $(\tilde S^R,\tilde N)$ with $\tilde S^R,\tilde N>0$ and $\tilde S^R+\tilde N=J-A$. Set
$\tilde w^0:=(A+\tilde S^R)/J\in(A/J,1)$. Because $s_a(A)>0$ and $\psi>0$, the cutoff satisfies $z(A)<b_N$, so
there exists a strictly increasing $C^1$ c.d.f.\ $\tilde G$ with everywhere positive density such that
$\tilde G(z(A))=A/J$ and $\tilde G(b_N)=\tilde w^0$ (both interpolation values are interior, so such a
$\tilde G$ exists). The identification theorem requires this observed equilibrium, not global uniqueness of
the perturbed enactment subgame. Audience primitives are unchanged, so the posteriors $q_m$ and sanctions $s_m$ are unchanged, and the
enactment condition $A=J\tilde G(z(A))$ holds by construction. Choose $\tilde H$ strictly increasing with full
support and $\tilde H(\tilde R)=J$, where $\tilde R$ is the formation return computed under $\tilde G$. Then
$(J,A,q_f,q_a,s_f,s_a)$ is an equilibrium of the perturbed primitives; every observable is unchanged (the
calm-state event frequency $\phi_0(\nu_f+J)$, the registered novel-state event mass $\psi(\nu_a+A)$, all sanctions
and posteriors), while the decomposition \eqref{eq:decomposition} now assigns
$(S^R,N)=(\tilde S^R,\tilde N)$. Conduct therefore cannot rank points in the relative interior of this line segment.
\end{proof}

\begin{proof}[Proof of Proposition~\ref{prop:pointid}]
Under (D1), $A^0-A=J[w^0-w(A)]=S^R$ at fixed $J$. Background invariance makes the measured change equal to
$S^R$ for acts and $\psi S^R$ for events. Under (D2), acts therefore give $J$, $1-J$, and
$A+N=J-S^R$; events also need (D3b), while unknown $\psi$ gives only the scaled mass and its sign. Under
(D3a), $A=(\nu_a+A)-\nu_a$ for acts or, using (D3b), $A=y/\psi-\nu_a$ for events; applying the same formula
after enactment protection yields $S^R=A^0-A$ and $N=J-A^0$. A small dated shock that displaces the
stock---a one-date formation subsidy, say; a purely contemporaneous enactment protection leaves $J_t$
unchanged and traces no transition---identifies impact timing around an outer culture and, along its
stable arm, the composite local root in Proposition~\ref{thm:forward-local}, not its components separately.
\end{proof}

\begin{proof}[Proof of Corollary~\ref{cor:flip}]
The sign of $-(1-p)\gamma(J_h-J_\ell)+p(V_N+D)(A_h-A_\ell)$ gives \eqref{eq:pbreak}; monotonicity of
$Y_F$, $Y_N^C$, and $a(\cdot)$ gives the ranking and the branchwise claim.
\end{proof}

\begin{proof}[Proof of Proposition~\ref{thm:forward-local}]
Stationarity in \eqref{eq:forward-value}--\eqref{eq:forward-stock} gives
$V^\ast=\widehat R(J^\ast)$ and $J^\ast=H(\widehat R(J^\ast))$, and the converse is immediate. Linearization
gives \eqref{eq:forward-characteristic}. At an outer culture $m_j<1$, so
$P_j(0)=1/\beta>1$ and $P_j(1)=\delta(1-\lambda)(m_j-1)/\lambda<0$; since the leading coefficient is positive,
the two roots are positive and lie on opposite sides of one. The smaller root is
\eqref{eq:forward-stable-root}. As $m_j\uparrow1$, that root tends to one, so the half-life diverges.
\end{proof}

The proofs of Theorems~\ref{thm:history-expectations} and~\ref{thm:frontier} are in Appendix~D of the main text.

\begin{proof}[Proof of Proposition~\ref{prop:realsurplus}]
The equilibrium set does not depend on $\varrho$, and under the finiteness assumption the set $E$ of
equilibria with at least one effectively sanctioned margin is finite. If $E$ is empty the conclusion is
vacuous; suppose it is not. Theorem~\ref{thm:central-welfare}(ii)
gives $J_j<J^{\ast\ast}$ and $A_j<A^{\ast\ast}$ for every $(J_j,A_j)\in E$, so
\[
 \epsilon:=\min_{(J_j,A_j)\in E}\min\{J^{\ast\ast}-J_j,\ A^{\ast\ast}-A_j\}>0 .
\]
Let $N_\epsilon$ be the intersection of the feasible set with the open sup-norm ball of radius $\epsilon/2$
around $(J^{\ast\ast},A^{\ast\ast})$. The maintained integrability assumptions make $\mathcal W$ continuous
on the compact feasible set, and $(J^{\ast\ast},A^{\ast\ast})$ is its unique maximizer, so the gap
$g:=\mathcal W(J^{\ast\ast},A^{\ast\ast})-\max_{\text{feasible}\setminus N_\epsilon}\mathcal W$ is strictly
positive. Set $\bar\varrho:=g/(2\|\Upsilon\|_\infty)$ if $\|\Upsilon\|_\infty>0$ and
$\bar\varrho:=\infty$ otherwise. For $\varrho\in[0,\bar\varrho)$ and any feasible $(J,A)\notin N_\epsilon$,
\[
 \mathcal W_\varrho(J,A)\le\mathcal W(J,A)+\varrho\|\Upsilon\|_\infty
 <\mathcal W(J^{\ast\ast},A^{\ast\ast})-\varrho\|\Upsilon\|_\infty
 \le\mathcal W_\varrho(J^{\ast\ast},A^{\ast\ast}),
\]
so every maximizer of the continuous $\mathcal W_\varrho$ on the compact feasible set lies in $N_\epsilon$
and therefore satisfies $J^{\ast\ast}_\varrho>J^{\ast\ast}-\epsilon/2>J_j$ and
$A^{\ast\ast}_\varrho>A^{\ast\ast}-\epsilon/2>A_j$ for every equilibrium in $E$.
\end{proof}

\end{document}